\documentclass[11pt,a4paper]{article}

\usepackage[T1]{fontenc}
\usepackage[utf8]{inputenc}
\usepackage[margin=1in]{geometry}
\usepackage{amsmath,amssymb,amsthm,mathtools,amsfonts}
\usepackage[backend=bibtex,style=alphabetic]{biblatex}
\usepackage{listings}
\usepackage[hidelinks]{hyperref}
\usepackage{orcidlink}
\usepackage{thmtools} 
\usepackage{thm-restate}
\usepackage[shortlabels]{enumitem}
\usepackage{mathrsfs}
\usepackage{bookmark}

\title{Generalizing LCL Complexity Gaps to Unbounded Degree via Monadic Second-Order Properties}
\author{%
	Chiara Piombi\,\orcidlink{0000-0001-5662-4443}\\
	CISPA Helmholtz Center for Information Security\\
    Saarbr\"ucken, Germany\\
	\texttt{chiara.piombi@cispa.de}}
\date{}

\theoremstyle{plain}
\newtheorem{theorem}{Theorem}[section]
\newtheorem*{theorem*}{Theorem}
\newtheorem{corollary}[theorem]{Corollary}
\newtheorem{lemma}[theorem]{Lemma}
\newtheorem*{lemma*}{Lemma}
\newtheorem{proposition}[theorem]{Proposition}

\theoremstyle{remark}
\newtheorem*{remark}{Remark}

\theoremstyle{definition}
\newtheorem{definition}[theorem]{Definition}

\begin{document}
\maketitle
\begin{abstract}
The last decade of research on the $\mathsf{LOCAL}$ model has seen tremendous progress in
understanding \emph{locally checkable labeling $(\mathsf{LCL})$ problems}, culminating in
an almost complete classification of the possible complexities $\mathsf{LCL}$ problems can
exhibit. In particular, on undirected trees, Chang and Pettie showed that there is no
$\mathsf{LCL}$ problem with complexity between $\omega(\log n)$ and $n^{o(1)}$ \cite{changHierarchy2019}
and Chang showed that, for every positive integer $k$, there is no $\mathsf{LCL}$ problem
with complexity between $\omega(n^{1/(k+1)})$ and $o(n^{1/k})$ \cite{treesChang2020};
additionally, which side of each \emph{gap} a problem is found on is decidable.

While the class of $\mathsf{LCL}$ problems -- which, roughly speaking, consists of problems for which
the correctness of a solution can be described by a finite set of allowed node
configurations, which in turn can be locally verified by a constant-time algorithm --
includes many important problems, it has one major restriction: problems can be defined
only on bounded degree graphs, which consequently restricts all the classification and gap
results mentioned above.

In this work, we propose a generalization of $\mathsf{LCL}$ problems to unbounded
degree using \emph{Presburger monadic second-order $(\mathsf{PMSO})$ formulas}; more
specifically, we consider what we call \emph{Local $\mathsf{PMSO}$ $(\mathsf{LPMSO})$
problems}, i.e., problems whose correct solutions are both \emph{finitely described} by a
$\mathsf{PMSO}$ formula and \emph{locally verifiable} by a $\mathsf{LOCAL}$
algorithm in constant time -- this class contains many of the important problems studied in
the $\mathsf{LOCAL}$ model but defines them on unbounded degree graphs. As our main result we prove that,
on unbounded degree rooted trees, the aforementioned $\omega(\log n)$--$n^{o(1)}$ and
$\omega(n^{1/(k+1)})$--$o(n^{1/k})$ complexity gaps (and their decidability) extend to the class of
$\mathsf{LPMSO}$ problems.

We obtain our gap results by proving two theorems that may be of independent
interest:
\begin{enumerate}
    \item We show that computing just a solution of an $\mathsf{LPMSO}$ problem is
    asymptotically as hard as computing \emph{both} a solution and a certificate for its correctness, if the
    certification is based on the minimal automaton verifying the $\mathsf{PMSO}$ formula.
    \item By modifying the minimal verification automaton slightly, we create a so-called rational
    multiset tree automaton whose run on the input rooted tree encodes both a solution and a
    certificate for its correctness; then we show that computing a run of \emph{any} rational multiset tree automaton on a rooted tree is asymptotically equivalent to solving the same problem on a rooted tree whose
    degree is bound by a constant depending only on the problem description.
\end{enumerate}
\end{abstract}
\paragraph*{Funding.} This work is funded by the Deutsche Forschungsgemeinschaft (DFG, German Research Foundation) -- Project number 534005335.
\paragraph*{Acknowledgements.} I want to thank Sebastian Brandt and Gustav Schmid for insightful
discussions, and Jonathan Osser and Noa Izsak for answering my questions about logical
formalism and automata theory.
\clearpage
\section{Introduction}
One of the central and most fruitful areas of research in the $\mathsf{LOCAL}$ model in
the last decade is the study of the possible complexities that problems from the class of
so-called Locally Checkable Labeling ($\mathsf{LCL}$) can exhibit. This class includes many interesting
problems, but has the major drawback of being only defined on bounded-degree graphs. In
this work, we propose a generalization of the class of $\mathsf{LCL}$ problems to the
unbounded-degree setting based on finite model theory and automata theory, and provide
\emph{gap results} for our characterization -- i.e., we show that there are intervals of
the complexity landscape into which no problems from our generalized class falls.

We begin by introducing the $\mathsf{LOCAL}$ model and discussing the literature on
$\mathsf{LCL}$ problems and their generalizations.
\subsection{Background and previous work}
\paragraph*{The LOCAL model}
We consider Linial's $\mathsf{LOCAL}$ model of distributed computing \cite{linialDGA1987},
where the input is a graph $G=(V,E)$ representing a network with processors
on each vertex $v\in V$ which can communicate over the edges $e\in E$. The computation
proceeds in synchronous rounds where each processor sends one message to each of its
neighbors; there are no limitations on message size or local computation power, and the
complexity of an algorithm is the number of rounds until all nodes terminate and output
their own part of a global solution.

We additionally assume that all nodes either know the graph size (the number $n$ of
vertices of the graph) or a polynomial bound on it and are equipped with some way to
distinguish between nodes: in the deterministic $\mathsf{LOCAL}$ model, nodes
are equipped with globally unique IDs composed of $O(\log n)$ bits, while in the randomized $\mathsf{LOCAL}$ model
nodes are equipped with infinite strings of random bits generated independently. Under those 
assumptions, a $\mathsf{LOCAL}$ algorithm that runs in $T(n)$ rounds is equivalent to an
algorithm that for each node $v$ gathers the $T(n)$-hop neighborhood of $v$ to determine
$v$'s output. We then refer to the complexity of a problem in the $\mathsf{LOCAL}$ model as its
\emph{locality}. 
\paragraph*{Locally Checkable Labelling problems}
A problem is truly \emph{local} when its locality is a constant which does not depend on the size of
the input graph; local problems are considered to be efficiently solvable in the $\mathsf{LOCAL}$
model and they are, in a way, a distributed analogue of the centralised problem class
$\mathsf{P}$. Similarly, a natural analogue of the
$\mathsf{NP}$ class is the class of \emph{Locally Checkable
Labelings} ($\mathsf{LCL}$) problems defined by Naor and Stockmeyer \cite{naorLocally1995}; these are problems that can be
\emph{verified} in constant time by a distributed algorithm. Explicitly, $\mathsf{LCL}$s are described via a \emph{finite} set of \emph{configurations} (labeled $r$-hop
neighborhoods for some finite $n\geq 0$) such that the overall solution is correct if and
only if for every node $v$ in the graph there is an isomorphism between the $r$-hop
neighborhood of $v$ and one of the configurations. This description naturally induces two restrictions: $\mathsf{LCL}$
problems are required to have finitely many input and output labels and can only be defined on graphs whose maximum
degree is bounded by a constant; this is because all possible labels and degrees
have to appear in the finitely many configurations.
\paragraph*{Complexity landscape}
The complexity of $\mathsf{LCL}$ problems has been widely studied, especially in the last
10 years (see \cite{balliuPaths2019,balliuRooted2022,brandtGrids2017} amongst others). The
most interesting results to come out of this study are \emph{gaps}, results that show no
problem can have a complexity strictly between two classes; gap results are extremely powerful,
as they can be used to automatically improve both upper bounds and lower bounds.

Conversely, \emph{density} results (like the ones found in \cite{balliuNewclasses2018})
show the lack of major gaps by constructing a set of problems whose
complexities form a dense set of the interval between two functions. These two types of
results together form what is referred to as the \emph{complexity landscape} of
$\mathsf{LCL}$ problems; this complexity landscape is today almost completely understood
(\cite{suomelaLandscapetalk2020} and Figure~\ref{fig:landgeneral} for an overview).
\begin{figure}[th]
\centering
\includegraphics[width=0.8\textwidth]{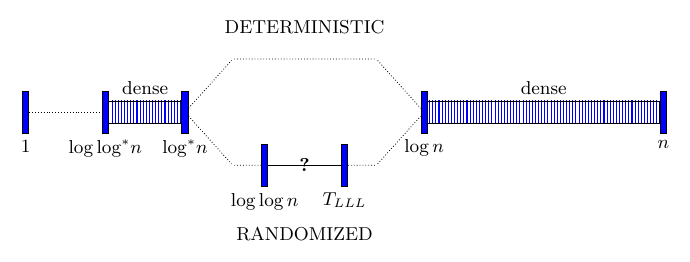}
\caption{The complexity landscape of $\mathsf{LCL}$s on general graphs.}
\label{fig:landgeneral}
\end{figure}
Even more information can be gathered by restricting the types of graphs the problem is
defined on: the complexity landscape of $\mathsf{LCL}$ problems is fully understood for many graph
classes, such as paths \cite{balliuPaths2019}, undirected trees
\cite{treesChang2020}, rooted trees \cite{balliuRooted2022}, and grids
\cite{brandtGrids2017}. We focus especially on undirected and rooted trees; they have
very similar complexity landscapes, as a problem on rooted trees can be encoded as a problem on
unrooted trees with inputs on the nodes -- additionally, both have the interesting property of containing
\emph{countably many discrete} complexity classes. Specifically, an $\mathsf{LCL}$ problem on (both
rooted and undirected) trees has deterministic complexity either
$$O(1),\Theta(\log^*n),\Theta(\log n),\text{ or }\Theta(n^{1/k})\text{ for some
}k\in\mathbb{N}_{\geq 1},$$
or is unsolvable. The main difference lies in the randomized complexity: there exist
problems on undirected trees with deterministic complexity $\Theta(\log n)$ and randomized
complexity $\Theta(\log\log n)$, such as sinkless orientation \cite{balliuSinkless2023} but no such
problems exist for rooted trees \cite{balliuRooted2022}.
\begin{figure}[th]
\centering
\includegraphics[width=0.8\textwidth]{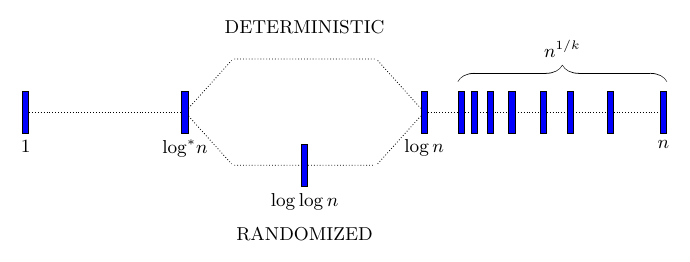}
\caption{The complexity landscape of $\mathsf{LCL}$s on undirected trees.}
\label{fig:landtrees}
\end{figure}
\begin{figure}[th]
\centering
\includegraphics[width=0.8\textwidth]{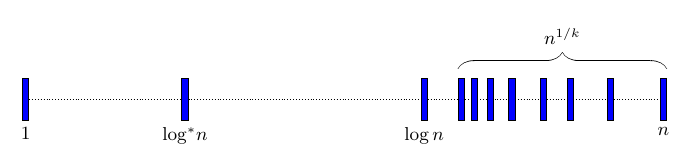}
\caption{The complexity landscape of $\mathsf{LCL}$s on rooted trees; note that the
randomized and deterministic landscapes are the same.}
\label{fig:landroottrees}
\end{figure}
Another focus of the research on distributed complexity gaps is \emph{decidability}, by which we
mean whether an algorithm exists that for a given problem decides which side of a gap it
lies on. Herein lies another difference between the complexity landscape of rooted and
undirected trees: the complexity of $\mathsf{LCL}$s on rooted trees is fully decidable,
while only the ``polynomial'' gaps (the $\omega(\log n)-n^{o(1)}$ and $\omega(n^{1/(k+1)})-o(n^{1/k})$ gaps)
are known to be decidable for unrooted trees.
\paragraph*{Generalizing LCL problems}
As we have seen, Locally Checkable Labeling problems have a nice and clean complexity
landscape; however they are limited by the nature of their description. We focus
specifically on the fact they can only be defined for \emph{bounded degree graphs}; most
natural problems, such as finding a maximal independent set (MIS) or a maximal matching,
have graph-theoretic descriptions that do not depend on the problem description -- and
even those who do, such as $(\Delta+1)$-coloring, can be defined for non-constant values
of the maximum degree $\Delta$.

Our objective is to generalize $\mathsf{LCL}$s in a way that preserves some of the gap results described
above and contains as many problems as possible: this would allow us to use the gaps to
boost upper bound and lower bounds on the complexity of a problem on unbounded degrees as
well. What properties should this generalization have?

First of all, we require our problems to be \emph{locally verifiable}: we can trivially
create problems that are not locally verifiable but have complexity $T(n)$ for any
sublinear function $T$, by requiring the nodes to output their $T(n)$-hop neighborhood.

We could try to naturally define the class of all problems that are locally
verifiable, without any restriction on having a finite description, finite labels or
bounded degree; however it was shown in \cite[Theorem 1]{schmidLFL2026} that for any real
number $r\in(0,1]$ there is such a problem with complexity $\Theta(n^r)$ in unbounded
degree trees. We then require our problems to have some form of \emph{finite description}
-- this is also an advantage when trying to \emph{decide} the complexity of a problem, as
it allows the decidability algorithm to read the whole description.

Lastly, we observe that even these restrictions are not enough to recover \emph{all} of
the above mentioned graph results: we know that finding a MIS on unbounded degree trees has complexity
$\Theta(\log n/\log\log n)$ (\cite{barenboimMIS2009,balliuMIS2019}), which shows that the
$\omega(\log^*n)$ to $o(\log n)$ gap does not hold (but a smaller one, i.e. a
$\omega(\log^*n)$ to $o(\sqrt{\log n})$ or $o(\log\log n)$ could still hold). In this
paper we focus on the polynomial gaps for rooted trees.
\paragraph*{Locally Finite Labeling problems}
One approach towards generalising $\mathsf{LCL}$ problems with finite descriptions is Locally
Finite Labeling problems, defined by Schmid in \cite{schmidLFL2026}. $\mathsf{LFL}$ problems are
described similarly to $\mathsf{LCL}$ problems, as a finite set of $r$-hop neighborhoods
such that a solution is correct if every node can find a \emph{morphism} between its own $r$-hop neighborhood and one of these configurations;
however, $\mathsf{LFL}$ configurations mark some edges as ``optional'': each ``optional''
edge can be used by the morphism arbitrarily many times (including $0$). This
class of problems has been shown to have decidable polynomial gaps.

Many important problems on unbounded degree graphs can be expressed in this formalism, 
such as finding a maximal independent set, $k$-coloring, and sinkless orientation; on the other hand, problems whose
description depends on arithmetic cannot always be described as $\mathsf{LFL}$s. One
example is \emph{degree splitting}, where each node is required to have as many blue
neighbors as red neighbors \cite[Section 1.2]{schmidLFL2026}.
\paragraph*{Meta-results on the theory of graphs}
Could we find some other way to finitely describe a problem on unbounded degree trees? In
this paper, we focus on problems whose solutions can be described by a single formula in the
\emph{monadic second-order} ($\mathsf{MSO}$) logic on trees, which are assumed to be rooted, unordered and
unranked. The relationship between trees and monadic second-order logic is
well-documented by Courcelle's theorem \cite{courcelleMSO2012}, which states that every
$\mathsf{MSO}$-definable property can be verified by a sequential algorithm on graphs of
bounded treewidth in linear time.

Monadic second-order properties on trees and tree-like graphs have been studied before in
the context of distributed algorithms: Feuilloley, Bousquet and Pierron \cite{feuilloleyMSOcert2022} showed that
$\mathsf{MSO}$ properties for trees can be certified with constant-size certificates and
verifiability radius $1$, while Jauregui, Li, Montealegre and Todinca \cite{jaureguiCourcelle2018} showed that there is
a $\mathsf{CONGEST}$ implementation of Courcelle's algorithm that can verify an
$\mathsf{MSO}$ property on graphs of bounded treewidth in $\tilde{O}(D)$ rounds, where $D$
is the diameter of the graph. Both of these results deal with \emph{verification}
problems, meaning problems where we need to check whether a (possibly labeled) graph has a
defined property -- such as whether a labeling is a correct solution to a problem, or a
graph has a certain structure. We will show that the relationship between $\mathsf{MSO}$ on trees and tree automata is strong
enough to define meta-theorems on \emph{construction} problems; that is, problem where you
are required to \emph{construct} a solution (usually a labeling) satisfying some property
described by the problem -- in our case, through a logical formula.
\paragraph*{Automata theory and Presburger definitions}
Automata theory appears in the proofs of $\mathsf{LOCAL}$ complexity gaps on trees by
Chang and Pettie (\cite{changHierarchy2019,treesChang2020}) as a way to formalize the
replacement of one tree with another which has a locally similar structure. Those results
use the bound on the maximum degree and constant checkability radius to construct the finite states
of the automaton, and so do not trivially generalize to the unbounded degree setting.
However, properties describable by $\mathsf{MSO}$ logic -- and by some of its extensions -- can be
characterized as properties recognizable by selected types of bottom-up tree automata (see
\cite{bonevaAutomata2005} for a review).  We focus on the
largest of these recognizable extensions: Presburger monadic second-order logic
($\mathsf{PMSO}$). Presburger formulas are able to describe simple arithmetic constraints
on integer variables; Presburger monadic second-order logic on rooted trees allows us to check those
arithmetic constraints on the sets of \emph{children} of a node -- they can be checked by
a class of automata called \emph{rational multiset tree automata}, where rational
multiset languages are roughly the multiset languages recognized by Presburger formulas.

In this paper, we show how we can use the local verifiability induced by the automata to
obtain gaps for $\mathsf{PMSO}$ \emph{certification} problems -- meaning problems where a solution is given
and we want to compute a certificate of its correctness that is locally verifiable. For
automata-recognizable properties, this can be done by returning a run of the verifying
automaton on the rooted tree. Additionally, we show that for \emph{construction} problems
which are both $\mathsf{PMSO}$-definable and \emph{locally verifiable}, we can find a
certification problem which also encodes a valid solution -- and solving this
certification problem is asymptotically exactly as hard as computing a solution.
\subsection{Organization}
We start by giving a high-level overview of the results and proof techniques in
Section~\ref{sec:results}, concluding with some related open questions; then in Section~\ref{sec:def} we give precise definitions for the $\mathsf{LOCAL}$ model (deterministic and
randomised), local verifiers, and the language of graphs. The bulk of the paper, where we
prove the main result, is Section~\ref{sec:main}; the proof overview contains more
specific references to its parts. Appendix~\ref{app:conversion} discusses the conversions between problems on
rooted and unrooted trees and between node-labeling and edge-labeling problems; Appendix~\ref{app:decomposition}
contains some technical results about graph decompositions, and Appendix~\ref{app:myhill} contains an
adaptation of the proof of the Myhill-Nerode theorem to the class of rational multiset
tree automata.
\subsection{A tour of the results}\label{sec:results}
\paragraph*{Presburger monadic second-order construction problems}
We focus on problems described by Presburger monadic second-order ($\mathsf{PMSO}$) logic
on rooted trees. $\mathsf{PMSO}$ is an extension of monadic
second-order logic which allows for ``linear'' constraints described in
terms of basic arithmetic (only using $\mathbb{N},+$ and $\leq$) on the children of a node\footnote{This can be
extended further to refer to all neighbors of a node on undirected general graphs;
however, the rooted tree formulation is
required for the automata-theoretic characterization to hold.}; this allows us
to define problems that first-order or even monadic second-order formulas cannot, such as
degree splitting \cite{halldorssonVertexsplit2022} and its variants. As an example,
consider a graph that is $\{B,R\}$-labeled (colored with blue or red) -- this is encoded
by two unary relations (i.e. sets) $B$ and $R$, so that ``$x$ is blue'' can be
expressed as $Bx$. We can convert this unary relation to a second-order variable by the monadic second-order formula
$$\exists S_B\,\forall x\, (x\in S_B)\leftrightarrow(Bx)$$
and we can do the same for the relation $Rx$ (``$x$ is red'') and the variable $S_R$. Then
the degree splitting problem ``the difference between the number of blue children and red
children of each node must be at most 1'' (to account for nodes with an odd number of
children) can be expressed by the formula
$$\forall x\, \left(\left((x\in S_B)\leftrightarrow(Bx)\right)\land \left((x\in S_R)\leftrightarrow(Rx)\right)\land\, x/\varphi(S_B,S_R)\right),$$
where $\varphi$ is the Presburger formula
$$\varphi(A,B)=\lvert A\rvert\leq\lvert B\rvert\leq \lvert A\rvert+1\,.$$

We define the class of construction problems $\Pi$ on rooted trees for which the property ``the graph $G$ is
correctly labeled with a solution of $\Pi$'' is both describable by a $\mathsf{PMSO}$ formula and
locally decidable; we call these Local $\mathsf{PMSO}$, or $\mathsf{LPMSO}$ construction
problems. We observe that the solutions to $\mathsf{LCL}$ problems can be described by
first-order formulas -- this is because the statement ``the $r$-hop neighborhood of vertex
$x$ is isomorphic to the finite graph $A$'' can be expressed by a first-order formula with
free variable $x$, and we can simply require each $x$ to satisfy at least one of the
finitely many isomorphism formulas to the allowed configurations. Then, any $\mathsf{LCL}$
problem can be defined by a $\mathsf{PMSO}$ formula over bounded-degree rooted trees;
conversely, many unbounded-degree generalizations of natural $\mathsf{LCL}$s can be
described by $\mathsf{PMSO}$ formulas, such as finding a MIS, a maximal matching, a $k$-coloring or
degree splitting. However, this class does not include the generalizations of problems for which the number of labels depends on
the maximum degree (i.e. is potentially unbounded), such as $(\Delta+1)$-coloring.

The main result of this paper is that the Local $\mathsf{PMSO}$ class exhibits the same
so-called \emph{polynomial} complexity gaps as the class of $\mathsf{LCL}$ problems on
rooted trees.
\begin{restatable}{theorem}{gapthm}\label{thm:gap}
    Let $\Pi$ be a Local $\mathsf{PMSO}$ construction problem on rooted trees; then one of the following holds:
    \begin{itemize}
        \item $\Pi$ is unsolvable, or
        \item $\Pi$ has deterministic and randomized $\mathsf{LOCAL}$ complexity $O(\log
        n)$, or
        \item $\Pi$ has deterministic and randomized $\mathsf{LOCAL}$ complexity
        $\Theta(n^{1/k})$ for some $k\in\mathbb{N}_{\geq 1}$.
    \end{itemize}
    Additionally, there is a sequential algorithm that can decide based only on the
    $\mathsf{PMSO}$ description of $\Pi$ which of those cases applies, and which $k$
    applies in the last case.
\end{restatable}
The proof of this theorem takes up the majority of this work. The result is achieved by solving what at first glance seems to be a harder problem:
computing \emph{both} a solution of $\Pi$ and a one-round verifiable certificate of the
correctness of the solution (Theorem~\ref{thm:autogap} + Lemma~\ref{lem:labelauto}); then showing that this new problem is not actually harder, as
by choosing our certificates carefully, we can compute them from a valid solution with a
constant-round algorithm (Theorem~\ref{thm:gap_jump}).

It is known that a $\mathsf{MSO}$-definable property on trees can be certified with constant-size
certificates \cite{feuilloleyMSOcert2022}, and this can easily be extended to
$\mathsf{PMSO}$-definable properties on rooted trees; the certificates are the states of a specific
type of bottom-up automata called rational multiset tree automata (RatMTA). These automata traverse a rooted tree from the leaves up to the
root -- a \emph{run} is a labeling which assigns to each node the state the automaton was
in when it reached that note. Seidl, Schwentick, and Muscholl \cite{seidlPresburger2003}
proved that a property defined by a $\mathsf{PMSO}$ formula $\varphi$ is also recognized by a
rational multiset tree automaton $\mathcal{A}$ -- meaning for every tree $T$ which
satisfies $\varphi$, there is a run of $\mathcal{A}$ on $T$ which gives an \emph{accepting}
state at the root.
\subparagraph*{Automata certification problems}
As the first result, we show that \emph{automata certification} problems on rooted trees
exhibit polynomial gaps: for any RatMTA $\mathcal{A}$ on $\Sigma$-labeled rooted trees we define the associated automata
certification problem of computing a run of $\mathcal{A}$ on any instance $\Sigma$-labeled
rooted tree.
\begin{restatable}{theorem}{autogapthm}\label{thm:autogap}
    Let $\mathcal{A}=(\Sigma,Q,\delta,F)$ be a rational multiset
    tree automaton, and let $\Pi_{\mathcal{A}}$ be the problem of
    computing a run of $\mathcal{A}$ on $\Sigma$-labeled rooted trees; then one of the following holds:
    \begin{itemize}
        \item $\Pi_{\mathcal{A}}$ is unsolvable, or
        \item $\Pi_{\mathcal{A}}$ has deterministic and randomized $\mathsf{LOCAL}$ complexity $O(\log
        n)$, or
        \item $\Pi_{\mathcal{A}}$ has deterministic and randomized $\mathsf{LOCAL}$ complexity
        $\Theta(n^{1/k})$ for some $k\in\mathbb{N}_{\geq 1}$.
    \end{itemize}
    Additionally, there is a sequential algorithm that can decide based only on the
    description of $\mathcal{A}$ which of those cases applies, and which $k$ applies in the last
    case.
\end{restatable}
We give an overview of the proof of Theorem~\ref{thm:autogap}, which takes up all of
Section~\ref{sec:auto_gap}. This proof is an adaptation of the $\omega(\log n)-n^{o(1)}$
and $\omega(n^{1/(k+1)})-o(n^{1/k})$ gap results for undirected trees from Chang and Pettie
\cite{changHierarchy2019,treesChang2020}; the central difference induced by the
unbounded-degree setting is the definition of \emph{classes} and \emph{types}. We quickly review the crucial
definitions from these proofs.
\begin{itemize}
    \item A \emph{partial} $Q$-labeling of a graph $G=(V,E)$ is a labeling with the symbols of
    $Q\uplus\{*\}$; nodes labeled with $*$ are considered unlabeled. A partial
    labeling $l':V\to Q\uplus\{*\}$ can be completed (or extended to a complete labeling) if there is a
    labeling $l:V\to Q$ which agrees with $l'$ on all nodes that are unlabeled;
    formally, if $l'(v)\neq *$ then $l'(v)=l(v)$.
    \item A \emph{pole} of an undirected tree $T$ is simply a special node; we refer to
    trees with one or two poles as \emph{unipolar} and \emph{bipolar} respectively.
    \item The \emph{class} of a partially labeled unipolar or bipolar tree with respect to a $\mathsf{LCL}$ problem $\Pi$ is the set of all labelings of the $r$-hop
    neighborhood(s) (where $r$ is the verifiability radius of $\Pi$) of the pole(s) which can be
    extended to valid labelings extended to full labelings of the unipolar tree which are
    valid according to $\Pi$ and agree with the partial labeling.
    \item A bipolar tree $H$ can be interpreted as a string of unipolar trees with a path
    connecting their poles, where the first and last node of this path are the poles of
    $H$; the string of types of those unipolar trees is the \emph{type} of $H$. Note that
    two trees with the same type have the same class. 
\end{itemize}
The two central properties of these definitions are: one, we can \emph{replace} a unipolar
or bipolar graph that is connected to a larger graph only by the poles by another of the
same \emph{class} without impacting whether the partial labeling can be completed. Two,
there are \emph{finitely many} classes of both unipolar and bipolar trees: this is because
they are sets of $Q$-labelings of trees of height $r$ and bounded maximum degree.
These two properties can be used to define an \emph{automaton} on the classes of unipolar
trees, but the second one clearly does not hold in the unbounded degree setting.

However, in this proof, we are given a rational multiset tree automaton from the start: we
show how we can use it and adapt it to give a new automata-theoretic definition of \emph{types}, which act as
a generalization of the concept of both types and classes defined above -- satisfying both
the replacement and finiteness properties. 

In fact, we show another strong property of automata-theoretic types: they can be used to
reduce the degree of our instance, allowing us to use results from the $\mathsf{LCL}$
version of the proof directly. Specifically, we can determine a value $\Delta$ (which only
depends on the description of the rational multiset tree automaton we start with) such
that every \emph{type} of unipolar and bipolar trees -- which in this context are called
\emph{tree types} and \emph{context types} respectively -- contains a tree of maximum
degree $\Delta$; even stronger, for each unipolar and bipolar tree, we can find a
\emph{subtree} of maximum degree $\Delta$ of the same type.

We will use this degree reduction procedure to argue that our certification problem is not
harder than the corresponding problem on graphs of maximum degree at most $\Delta$ --
which is an $\mathsf{LCL}$, since our certification can be verified in one round by the
transition function. To do so, we apply the classification results from the Chang-Pettie proofs
indirectly: if we have a $O(\log n)$ or $O(n^{1/k})$ round algorithms for the restricted
$\mathsf{LCL}$ problem, then we have a so-called a \emph{feasible function} of
suitable strength, which
allows us to precompute labels for some parts of our tree. We use this feasible function
to describe a $O(\log n)$ or $O(n^{1/k})$ algorithm for the unbounded degree certification
problem. This structure is robust: both the Chang-Pettie classification and the algorithms
we describe can be implemented in both the deterministic and randomized $\mathsf{LOCAL}$
model.

We give further details on this proof in the paragraphs below.
\subparagraph*{Automata-theoretic types}
Throughout this paper, we refer to the \emph{subtree rooted at a node} $v$ of a rooted
tree; this is the subtree induced by $v$ and all of its descendants. We can interpret a
subtree rooted at $v$ as a unipolar tree with a pole at the root $v$ -- in general, in
this rooted framework, we require unipolar trees to have their pole at the root. 

Conceptually, we define the \emph{type} of a node $v$ of a rooted tree as the set of all states of
$v$ that can be extended to a valid (but not necessarily accepting) run of the subtree
rooted at $v$ -- in practice, we define the type of a node $v$ of an unlabeled tree as the state it is assigned
by the deterministic and complete version of the automaton $\mathcal{A}$ we start from, called
$\mathcal{A}_{det}$. If the set of states of $\mathcal{A}$ is $Q$, the set of states of
$\mathcal{A}_{det}$ is $2^Q$, the set of all subsets of $Q$. Note that when we talk about
``unlabeled'' and ``partially labeled'' trees here, we are referring to output labels, i.e.
automaton states: the trees might still be labeled by input labels in $\Sigma$.

We extend the definition of $\mathcal{A}_{det}$ to work on rooted trees which are
\emph{partially labeled} with states of $\mathcal{A}$, such that a node that already has a
label $q$ is either assigned state $\{q\}$ (if there is a valid labeling of its rooted
subtree which agrees with the partial labeling) or state $\emptyset$ (if there is no such
labeling) -- we call this new automaton $\mathcal{A}_{lab}$, and define the \emph{type} of
a partially labeled rooted tree as the state that $\mathcal{A}_{lab}$ assigns to its root
(which is unique, as $\mathcal{A}_{lab}$ is deterministic and complete). The full
definitions of $\mathcal{A}_{det}$ and $\mathcal{A}_{lab}$ can be found in
Section~\ref{par:types}; the relation between this definition and the conceptual
definition we provide at the start can be seen in Lemma~\ref{lem:typeauto}.

What about bipolar trees? In this rooted framework, we require a bipolar tree to have a
\emph{consistently directed} path between the poles -- we can then relate bipolar trees to the automata-theoretic concept of a
\emph{context}: a rooted tree with one specifically marked \emph{leaf}, called a
\emph{variable node}. We denote a context as $C[x]$, where $x$ represents the variable
node -- given a rooted tree $T$, we can substitute $T$ into $C[x]$ (denoted as $C[T]$) by
identifying the variable node of $C[x]$ with the root of the tree $T$. We can interpret a
context as a bipolar tree with one pole at the root and one at the variable node.

Contexts allow us to define the \emph{replacement} properties of types, which are
discussed extensively in Section~\ref{sec:autypes} -- roughly, thanks to the
$\mathcal{A}_{lab}$ automaton for partially labeled rooted trees we constructed, we can replace a rooted (unipolar) tree with
a rooted (unipolar) tree of the same type in every context (overall graph) without
affecting whether the partial labeling can be completed. For each context $C[x]$ we can find a \emph{function} between types
$\mathrm{Type}(T)\mapsto\mathrm{Type}(C[T])$ -- the \emph{type} of a context (bipolar tree) can be
identified with this function, of which there are finitely many. We show that we can
replace a context (bipolar tree) by one of the same type in every context composition
(overall graph) without affecting whether the partial labeling can be completed.

Large contexts can be hard to efficiently analyze, so we restrict to \emph{direct contexts}, meaning
contexts where the variable node is an immediate child of the root; every context $C$ can be
expressed as a string of direct contexts $C_1,\ldots,C_k$, each corresponding to a node
along the directed path from root to variable node, such that $C=C_1\circ\ldots\circ C_k$.
This roughly corresponds to the string of \emph{classes} of unipolar trees which defines a
bipolar trees in the Chang-Pettie definitions.
\subparagraph*{Degree reductions}
Rational multiset tree automata have transitions defined by \emph{rational} sets of multisets. Rational
sets of multisets can be equivalently defined as sets of multisets definable by a
Presburger formula, or as \emph{semilinear} sets of multisets -- finite unions of linear
sets. A linear set of multisets is a set of the form
$$L=M_0+\mathbb{N}M_1+\ldots+\mathbb{N}M_k,$$
where $M_0,M_1,\ldots,M_k$ are multisets and $+$
represents multiset union. $M_0$ is the \emph{base multiset} or \emph{base vector}.

We focus on tree types and direct context types; the type of a partially labeled tree or direct context is fully defined by the
multiset of tree types of its (non-variable) children, the input label of its root, and
the partial output label of its root. We can then identify types
with sets of multisets (containing all multisets that appear in some tree or context of
that type) -- for tree types, the set corresponding to the root input-output label pair $(i,s)$ state $q$ is the
semilinear set $\delta_{lab}((i,s),q)$. We prove in Lemma~\ref{lem:contypes} that the set
representing a context type is also semilinear.

Semilinear sets have a useful property (Lemma~\ref{lem:semibase}): there is a value $d$ such that for
every multiset $M$ in the semilinear set there is an element $M'\subseteq M$ also in the
semilinear set of size $\leq d$ -- $d$ is simply the length of the largest base multiset of
a linear component, as the base multiset of a linear set is contained in all multisets of
a linear set.

By using this property a node $v$ of a rooted tree which knows the types of all its
children can compute its own tree type, and find a subset of its children that would
define the same tree type of size $\leq d$. We can set $\Delta_R$ as the largest $d$ over
all possible types, and define a procedure $\mathrm{ForgetChildren}_R(v)$ in which $v$ ``forgets'' about the children
which are not required to define its type, reducing its degree to $\leq\Delta_R+1$ (as the
possible edge between $v$ and its parent is not affected) without changing its type.

Similarly, if a node $v$ knows the types of all its children but one, it can treat itself
as a context and the unknown child as a variable node; we can then repeat the procedure
above for context types, define $\Delta_C$ as the maximum $d$ over all finite context
types, and define a procedure $\mathrm{ForgetChildren}_C(v)$ in which
$v$ ``forgets'' about some of its children and reduces its degree to $\leq\Delta_C+2$ (as
the edges between $v$'s parent, $v$ and the unknown child are unaffected) without changing
its type.

The finite representation of types in general is discussed in Section~\ref{sec:fintypes}; while the
degree reduction procedure is detailed later in Section~\ref{par:degred}.
\subparagraph*{Feasible functions}
Thanks to the degree reduction procedures we defined, we are able to utilize tools that
were defined to solve $\mathsf{LCL}$ problems on bounded-degree graphs -- in particular,
\emph{feasible functions}. Roughly speaking, a function is feasible on \emph{unlabeled}
(with regards to the output labeling) bipolar trees if it allows us to choose labels for the middle nodes of ``long enough'' bipolar
trees in a way that can always be extended to a full labeling of any \emph{unlabeled} tree
containing it. However, if we want to apply this feasible function multiple times, we need
to define it on \emph{partially labeled} bipolar trees -- this is not always possible for
all partially labeled bipolar trees. Intuitively, every application of the feasible
function might create new types of bipolar tree -- a function is $k$-feasible roughly if
it can be applied $k$ times in a row, and $\infty$-feasible if it is feasible over all
partially labeled bipolar trees. A more detailed definition can be found in Section~\ref{par:feasfunc}.

Now, consider the problem $\Pi_{\mathcal{A},\Delta}$, which is the restriction of
$\Pi_{\mathcal{A}}$ to rooted trees of maximum degree at most
$\Delta=\max\{\Delta_R+1,\Delta_C+2\}$. Clearly, $\Pi_{\mathcal{A},\Delta}$ cannot be harder than
$\Pi_{\mathcal{A}}$, as any algorithm for $\Pi_{\mathcal{A}}$ needs to solve all instances
of $\Pi_{\mathcal{A},\Delta}$.

We convert $\Pi_{\mathcal{A},\Delta}$ to a problem of the same complexity 
on undirected trees (see Appendix~\ref{app:conversion} for details) and use the following results from the
Chang-Pettie proofs: $\Pi_{\mathcal{A},\Delta}$ can be solved in $O(\log n)$ rounds if and
only if a $\infty$-feasible function exists for $\Pi_{\mathcal{A},\Delta}$, and similarly
$\Pi_{\mathcal{A},\Delta}$ can be solved in $O(n^{1/k})$ rounds if and only if a
$k$-feasible function exists for $\Pi_{\mathcal{A},\Delta}$. We argue that we can use
these feasible functions to obtain $O(\log n)$ and $O(n^{1/k})$ algorithms for
$\Pi_{\mathcal{A}}$ as well -- this is roughly equivalent to claiming that we can use an
algorithm for $\Pi_{\mathcal{A},\Delta}$ to solve $\Pi_{\mathcal{A}}$, but without being
able to treat the algorithm as a black-box.
\subparagraph*{Algorithm}
We describe the algorithm used to solve $\Pi_{\mathcal{A}}$; the technical description, as
well as correctness and complexity analysis, can be found in Section~\ref{par:algorithm}.

During the preprocessing step, we can compute a description of $\mathcal{A}_{lab}$, as
well as the value of $\Delta$ by brute-force enumeration of the
finitely many tree types and context types. Then we can compute the exact $\mathsf{LCL}$
description of $\Pi_{\mathcal{A},\Delta}$ over undirected trees, and use the results from
Chang and Chang-Pettie to decide its complexity and explicitly compute a suitable feasible
function $f$ -- including the value $\ell$, which depends only on the description of
$\Pi_{\mathcal{A},\Delta}$, such that $f$ is feasible for all bipolar trees of length $\in[\ell,2\ell]$.

We compute a parametrized decomposition of our unlabeled rooted tree instance $T=(V,E)$
with parameters $\gamma,\ell,L$, which is an ordered partition of the vertex set $V$ in
$$R_1<C_1<R_2<\ldots<R_L<C_L<R_{L+1}$$
with the following properties:
\begin{enumerate}
    \item for every $i\in[1,\ldots,L+1]$, the subgraph induced by $R_i$ is a forest of rooted tree
    of height at most $\gamma$ such that only the parent of the roots might be in a higher
    layer, and
    \item for every $i\in[1,\ldots,L]$, the subgraph induced by $C_i$ is a forest such
    that every connected component is a directed path of length $\in[\ell,2\ell]$, and
    only the starts/ends might have a parent/child in a higher layer (respectively).
\end{enumerate}
Aside from the parameters of this decomposition, the algorithm is the same for both the
$O(\log n)$ case and $O(n^{1/k})$ case -- we apply the feasible function on the layers
$C_i$ iteratively, so the number of layers cannot be higher than the ``strength'' of the
feasible function. If $f$ is a $\infty$-feasible function (equivalently,
$\Pi_{\mathcal{A},\Delta}$ has complexity $O(\log n)$) we compute a $(1,\ell,O(\log n))$
decomposition in $O(\log n)$ rounds; if $f$ is a $k$-feasible function (equivalently,
$\Pi_{\mathcal{A},\Delta}$ has complexity $O(n^{1/k})$) we compute a
$(O(n^{1/k}),\ell,k)$-decomposition in $O(n^{1/k})$ rounds -- details on how we compute
the decompositions can be found in Appendix~\ref{app:decomposition}.

Then, we can compute the types (according to $\mathcal{A}_{lab}$) of all nodes going
``up'' the layers: we can compute the types of a layer $R_i$ in $O(\gamma)$ rounds by
brute force, as long as we know the types of all nodes in the previous layers -- as we do,
each node applies the procedure $\mathrm{ForgetChildren}_R$ to virtually reduce its own
degree to $\leq\Delta$. When we reach a $C_i$ layer, assuming we know the types of all nodes of previous
layer, we can compute the types of nodes of $C_i$ as direct contexts -- we apply the
degree reduction procedure to them as well, virtually reducing their degree to $\leq
\Delta$. Now each connected component of $C_i$ represents a bipolar tree of length $\in[\ell,2\ell]$ and maximum
degree $\leq\Delta$: we can apply the feasible function $f$ for $\Pi_{\mathcal{A},\Delta}$
to assign labels to the middle nodes of the bipolar tree in a way that will not result in
an unextendable labeling -- which tells us their type according to $\mathcal{A}_{lab}$. Then we
can compute the types of all nodes of the bipolar tree which are \emph{ancestors} of the labeled
nodes, while nodes which are \emph{descendants} can wait to compute their type until the
last layer, since their parents do not need to wait for them anymore. This takes $O(\ell)$
rounds, and we repeat the whole process for each layer, taking $O((\gamma+\ell)L)$ rounds.

Once we have computed the types of \emph{all} nodes in the (now partially labeled) graph
according to $\mathcal{A}_{lab}$, we can choose an accepting state for the root from its
type (this must exist due to the feasibility of the function -- see
Lemma~\ref{lem:correct}). Then all nodes that have selected a state (through the
prelabeling or by being the root) can choose a label for their children, \emph{including}
the children that were temporarily ``forgotten'' to reduce the degree, which then choose
a label for their children, and so on -- we can show that every node is the
$O((\gamma+\ell)L)$-hop descendant of a prelabeled node, so this takes $O((\gamma+\ell)L)$
rounds and achieves a valid labeling. The $O((\gamma+\ell)L)$ complexity corresponds to
$O(n^{1/k})$ for $\gamma\in O(n^{1/k}),\ell\in O(1), L=k$ and to $O(\log n)$ for
$\gamma,\ell\in O(1), L\in O(\log n)$.

This completes the proof of Theorem~\ref{thm:autogap}.
\subparagraph*{Encoding labels as states}
Because every $\mathsf{PMSO}$ property can be recognized by a rational multiset tree
automaton, we could use Theorem~\ref{thm:autogap} to compute a certificate for an existing
solution -- we show that we can instead use it to compute \emph{both} a solution and a
certificate of its correctness. Specifically, we will find an automaton which does \emph{not}
certify the property ``is this labeling a valid
solution for $\Pi$?'' but instead ``does there \emph{exist} a valid labeling of this graph
according to $\Pi$?''. The details of this construction can be found in
Section~\ref{sec:labelauto}; the main result from that section that we will use to prove
Theorem~\ref{thm:gap} is Lemma~\ref{lem:labelauto}.
\begin{restatable}{lemma}{labelautolem}\label{lem:labelauto}
    Let $\mathcal{A}=(\Sigma_{in}\times\Sigma_{out},Q,\delta,F)$ be a rational multiset tree automaton
    recognising a $\Sigma_{in}\times\Sigma_{out}$-labeled rooted tree property $\mathcal{G}$; then there is a RatMTA
    $\mathcal{A}'=(\Sigma_{in},\Sigma_{out}\times Q,\delta',\Sigma_{out}\times F)$ which recognises the
    language of all $\Sigma_{in}$-labeled rooted trees $T$ such that there exists a
    $\Sigma_{out}$-labeling of $T$ such that:
    \begin{itemize}
        \item $T$ labeled with the resulting $\Sigma_{in}\times\Sigma_{out}$-labeling is
        in $\mathcal{G}$ (is recognized by $\mathcal{A}$), and
        \item a valid run of $\mathcal{A}'$ consists of a valid $\Sigma_{out}$-labeling
        and a run of $\mathcal{A}$ verifying the correctness of this labeling.
    \end{itemize}
\end{restatable}
Together, Theorem~\ref{thm:autogap} and this result show that the polynomial gaps hold for problems
requiring us to compute both a solution and a certificate of the correctness of this
solution. This problem is at least as hard as computing \emph{only} a solution -- the last
remaining step is showing that they actually have the same complexity.
\subparagraph*{Minimal automaton}
Note that up to this point, we have required the correctness of our solution 
to be described by a $\mathsf{PMSO}$ sentence (equivalently, certifiable by a rational multiset tree
automaton), but not locally verifiable; let us now assume we have a property $\mathcal{G}$
which is \emph{both} $\mathsf{PMSO}$-definable and locally decidable with radius $r$.

By the result of Seidl, Schwentick, and Muscholl \cite{seidlPresburger2003} we are
guaranteed that \emph{some} rational multiset tree automaton exists which can recognize
$\mathcal{G}$. However, not \emph{every} RatMTA recognizing $\mathcal{G}$ will allow us to
efficiently compute certificates -- we provide an example below.

Let $\Pi$ be a very simple problem: we want to label a rooted tree without inputs with the label
$T$. No other conditions are given, and there is a $0$-round $\mathsf{LOCAL}$
algorithm that outputs $T$ on each node -- which is trivially $O(\log n)$. We define a
rational multiset tree automaton that recognises the language $\mathcal{T}$ of $T$-labeled
rooted trees: $\mathcal{A}_T=(\{T\},\{0,1\},\delta_T,\{0,1\})$ where
$\delta(T,0)=\mathcal{M}(\{1\})$ (all multisets composed only of $1$s) and
$\delta(T,1)=\mathcal{M}(\{0\})$ (all multisets composed only of $0$s). Those sets are
semilinear (in fact, they are linear) and so this is a valid RatMTA; a run of
$\mathcal{A}_T$ is a $2$-coloring of the rooted tree, which always exists for every rooted
tree -- since the automaton accepts every state, this automaton accepts every tree.
However, computing a $2$-coloring of a rooted tree takes $\Omega(n)$ rounds in the
$\mathsf{LOCAL}$ model \cite{balliuRooted2022}. The result we were trying to achieve --
computing a solution \emph{and} a certificate is not harder than computing \emph{only} a
solution -- is then false for this specific choice of automaton.

Our example purposefully constructed a wasteful automaton, and we can see that the
language $\mathcal{T}$ could be more easily recognized by a one-state automaton -- then,
we could compute a run of it in $0$ rounds. To avoid such wasteful constructions, we need to work on the
\emph{simplest} possible automaton; this corresponds to the automata-theoretic concept of 
\emph{minimal} (by number of states) \emph{deterministic} and \emph{complete} automaton. For classical
(ranked, ordered) bottom-up tree automata, the existence and uniqueness of the minimal automaton is
guaranteed by the Myhill-Nerode theorem: we prove a version of it for rational multiset
tree automata on unranked, unordered trees in Appendix~\ref{app:myhill}.

The states of the minimal automaton are in bijection with the equivalence classes of a relation depending only on the property,
denoted $\equiv_{\mathcal{G}}$, which is the coarsest \emph{context-free} equivalence; two
trees $T_1$ and $T_2$ are equivalent iff for every context $C[x]$ we have
$C[T_1]\in\mathcal{G}\Leftrightarrow C[T_2]\in\mathcal{G}$. In simpler words,
$T_1\equiv_{\mathcal{G}} T_2$ if we can replace a rooted subtree $T_1$ by $T_2$ while
preserving the correctness of the labeling; as we have observed during the previous
proofs, this is in general true for trees of the same \emph{type}, meaning the same state
of the root, so the coarsest (minimal number of equivalence classes)
such equivalence relation corresponds to the automata with the minimal number of states.
\subparagraph*{Locally computing states}
This replacement property is in fact the same that underlies the proof for the bounded
degree gaps; however, in that case, two partially labeled unipolar trees are equivalent roughly speaking if the $2r$-hop
neighborhoods of their poles are isomorphic \emph{and} any partial labeling of this
isomorphic neighborhood can be extended in one tree if and only if it can be extended in
the other. In the unbounded degree setting, there are infinitely many equivalence classes
for this relation -- however, \emph{because} it satisfies the replacement property, this
neighborhood-based equivalence is a \emph{refinement} of the relation
$\equiv_{\mathcal{G}}$ defined above; by restricting to fully labeled trees, this implies that \emph{two correctly labeled trees
with isomorphic $2r$-hop neighborhoods of the root have the same type} in the minimal
automaton.

We can then build a very simple reduction from our original problem (``find a correct
solution for $\Pi$'') to the following problem: find \emph{both} a correct solution for
$\Pi$ and a run of the \emph{minimal} automaton certifying that solution. By
Lemma~\ref{lem:labelauto} we know this can be expressed as a single automaton
certification problem on unlabeled trees, and by Theorem~\ref{thm:autogap} we know these
problems have polynomial gaps.
\begin{restatable}{theorem}{gapjumpthm}\label{thm:gap_jump}
    Let $\Pi=(\Sigma_{in},\Sigma_{out},\mathcal{G})$ be a Local $\mathsf{PMSO}$ problem with checkability radius $k$ and let
    $\mathcal{A}=(\Sigma_{in}\times \Sigma_{out},Q,\delta,F)$ be the \emph{minimal} deterministic and complete RatMTA
    recognising $\mathcal{G}$. Denote $\mathcal{G}_Q$ the graph language of $(\Sigma_{in}\times \Sigma_{out}\times
    Q)$-labeled rooted trees such that for each $G\in\mathcal{G}_Q$:
    \begin{enumerate}
        \item the underlying $\Sigma_{in}\times \Sigma_{out}$-labeled graph of $G$ is in $\mathcal{G}$, and
        \item the underlying $Q$-labeled graph of $G$ is the run of $\mathcal{A}$ on its
        underlying $\Sigma_{in}\times \Sigma_{out}$-labeled graph.
    \end{enumerate}
    Then given a correct solution to $\Pi$, we can compute the run of $\mathcal{A}$ which
    accepts it in $2k+2$ rounds; equivalently,
    $\Pi$ and $\Pi'=(\Sigma_{in},\Sigma_{out}\times Q,\mathcal{G}_Q)$ have the same asymptotic complexity.
\end{restatable}
Together Theorem~\ref{thm:autogap}, Lemma~\ref{lem:labelauto} and
Theorem~\ref{thm:gap_jump} form a full proof of Theorem~\ref{thm:gap}.
\subsection{Open questions}
The central result of this paper is a step towards a possible generalization for $\mathsf{LCL}$ problems based
on finite model theory and automata theory; however, the use of bottom-up tree automata
forces us to restrict the results to rooted trees. In Section~\ref{sec:axioms} we propose
a generalization of the $\mathsf{PMSO}$ logic -- usually defined only on rooted trees -- to
general graphs; however, without the directional input of the rooting, we cannot easily generalize
the automata-theoretic characterization as well. We conjecture that by using a suitable
generalization of bottom-up tree automata -- possibly inspired by \emph{alternating
distributed graph automata} defined by Reiter in \cite{reiterAutomata2015} -- we could
extend the gap results to undirected trees. We leave this for future work.

Similarly, it is possible that our choice of assumptions on local decidability could be
improved, giving us a larger class of problems exhibiting the same gaps;
$\mathsf{LD}^*$ is the weakest notion of local decidability, which does not allow the verifier to
know the total size of the graph or the values of the unique IDs. The graph size
especially needs to be heavily taken into account when proving the \emph{replacement}
properties of types, which are fundamental to the proof of Theorem~\ref{thm:gap_jump},
especially the degree reduction procedures. 

We speculate that stronger notions of local
decidability might actually falsify the gap result -- that is, that there is a
$\mathsf{PMSO}$-definable problem on rooted trees which can only be locally verified by an
algorithm which is aware of the size of the graph, and which has complexity $\Theta(n^r)$
for some $r\in (0,1)$ which is not of the form $r=1/k$.

Additionally, this paper focuses on the so-called ``polynomial'' gaps -- many questions
about the full complexity landscape of this class are still open. Are there any Local
$\mathsf{PMSO}$ problems of complexity between $\omega(1)$ and $O(\log^*n)$? Is there a
difference between the randomized and deterministic complexities of Local
$\mathsf{PMSO}$ problems on unbounded-degree rooted trees? Are other common techniques
used to prove results on $\mathsf{LCL}$ problems -- such as round elimination --
applicable to Local $\mathsf{PMSO}$ problems? We speculate that many such techniques can
be proven to work in suitably large classes of problems described by logical formulas,
establishing a class of \emph{meta-theorems} for the $\mathsf{LOCAL}$ model.
\section{Definitions, notation, and assumptions}\label{sec:def}
In this section we review some preliminaries on the $\mathsf{LOCAL}$ model and finite model theory. 
\subsection{Local computation}\label{sec:localmodel}
We refer to the standard $\mathsf{LOCAL}$ model as defined in \cite{linialLocality1992}.
Specifically, we refer to two variants of it: the first is the \emph{deterministic} $\mathsf{LOCAL}$ model, where
every node of the input graph $G$ is provided with a globally unique ID from the interval $[1,\ldots,n^c]$ where
$n$ is the number of nodes of $G$ and $c\geq 1$ is a known constant; additionally, each
node is aware of the exact value of $n$.

The second variant is the \emph{randomized} $\mathsf{LOCAL}$ mode,
where nodes are \emph{not} provided unique IDs, but instead access to a perfectly random
coin toss, where each coin is independent of the others; equivalently, each node is
provided an endless string of random bits. Nodes are still aware of the exact value of $n$, and the solution of the algorithm is required to be correct with probability
at least $1-\frac{1}{n}$ over the probability space of all possible coin results.

For most of this paper, we focus on \emph{node labeling} problems: we assume our graph is
simple trees, and that we only care about whether an edge exists or not; we do not get inputs
or assign outputs to them, except for a possible orientation. This is because -- up to some
graph operations that can be simulated internally in the nodes -- problems with inputs and
outputs on edges and half-edges can be transformed into problems on nodes only. See
Appendix~\ref{app:conversion} for a more detailed description of this construction.

Additionally, we discuss the issue of \emph{solvability}: throughout this paper, we assume a
problem is \emph{solvable} if it can be solved on \emph{every} instance of the considered
graph class. We denote a problem as \emph{weakly solvable} if it can be solved on all but
finitely many graphs, or equivalently if there exists a constant $n_0$ such that it can be
solved on all graphs with more than $n_0$ vertices -- this is the definition of
solvability that is often used in the literature. Observe however that a weakly solvable
problem can be converted into a solvable problem at the cost of increasing both the
complexity and verifiability radius (i.e. the complexity of a distributed verifier for the
solution; see definition below) by $n_0$, such that an algorithm can always recognize if
the instance given is unsolvable and output some error. If a $\mathsf{LOCAL}$ problem is
unsolvable for infinitely many instances -- equivalently, for at least one instance of
infinitely many graph sizes -- we say it is \emph{unsolvable}.
\subsection{Local verifiers}
We discuss local verification following the notation from
\cite{fraignaudIdentifiers2012,fraigniaudIdentifiers2013}, with one major difference: we
do not assume our instances are connected graphs. This is not relevant in the context of
undirected or rooted trees, as all trees are assumed to be connected; but it is relevant
to defining the class of local first-order problems. This is because connectivity is not
axiomatizable by a first-order sentence; when
working with monadic second-order or higher logic, connectivity can be required in the
problem description itself.
\begin{definition}
    An \emph{instance} is a pair $(G,\mathbf{x})$ where $G=(V,E)$ is a simple graph and $\mathbf{x}$ is
    a function $V\to\{0,1\}^*$ assigning to every node a binary input
    string, called a \emph{label}.
\end{definition}
\begin{definition}
    A \emph{graph property} (or \emph{graph language}) is a \emph{decidable} class of instances.
\end{definition}
\begin{definition}
    We define a \emph{generalised graph construction problem} as a tuple
    $\Pi=(\Sigma_{in},\Sigma_{out},\mathcal{G})$ where:
    \begin{itemize}
        \item $\Sigma_{in}$ is a finite set of input labels,
        \item $\Sigma_{out}$ is a finite set of output labels, and
        \item $\mathcal{G}$ is a property of $\Sigma_{in}\times\Sigma_{out}$-labeled graphs.
    \end{itemize}
    An \emph{instance} for $\Pi$ is a $\Sigma_{in}$-labeled graph $(G,\lambda_{in})$, and a
    \emph{solution} to $\Pi$ is a $\Sigma_{out}$-labeling $\lambda_{out}$ of $G$ such that
    $(G,\lambda_{in}\times\lambda_{out})$ is in $\mathcal{G}$.
\end{definition}
Each construction problem is naturally associated with a decision problem, that is
the one that checks whether $(G,\lambda_{in}\times\lambda_{out})\in\mathcal{G}$ (decides $\mathcal{G}$) as a
language; however, multiple construction problems can be associated with the same decision
problem.
\begin{definition}\label{def:distributedverifier}
    A \emph{distributed verifier} for a graph property $\mathcal{P}$ is a distributed algorithm
    $\mathcal{A}$ with output space $\{\text{``yes'',``no''}\}$ such that,
    for all graphs $G$ we have
    $$G \in \mathcal{P} \Leftrightarrow \forall v\in V\ \mathcal{A}(G,\mathbf{x},v)=\text{``yes''}$$
    meaning, $\mathcal{A}$ outputs ``yes'' on all nodes if and only if $G \in\mathcal{L}$.
\end{definition}
There has been a large volume of discussion about the role of identifiers and knowledge of
size$(G)$ have on the class of locally decidable problems; see
\cite{fraignaudIdentifiers2012,fraigniaudIdentifiers2013}. We choose to restrict local
verifiers in a way that is analogous to $\mathsf{LCL}$ problems: the allowed configurations, and so
the local verifier, do not depend on the ID assignment or on the knowledge of size$(G)$.
However, the algorithm can use the uniqueness of IDs (and not their value) to precisely
determine the topology of the graph; this is the deterministic \emph{anonymous} $\mathsf{LOCAL}$ model defined in \cite{fraignaudIdentifiers2012}.

Let $t:\mathbb{N}\to\mathbb{N}$ be a function; we define the class of \emph{Locally
Decidable} problems in time $t$ (denoted $\mathsf{LD}^*(t)$) as the class of all graph
properties that can be decided by a distributed verifier that takes at most $t(n)$
communication rounds, where $n$ is the size of the gaph. We define
$\mathsf{LD}^*(O(t)):=\bigcup_{c\in\mathbb{N}}\mathsf{LD}^*(c\cdot t)$; finally, we define
$\mathsf{LD}^*:=\mathsf{LD}^*(O(1))$ as the class of all problems that can be decided 
by a distributed verifier in constant rounds.
\subsection{Axiomatizable graph properties}\label{sec:axioms}
To work within a logic formalism, we must first establish a vocabulary and axioms for the
theory of graphs. A (finite, simple, unlabeled) graph is commonly represented as a pair $(V,E)$
where $V$ is a finite set of vertices/nodes and $E\subseteq V\times V$ is a set of edges;
we can interpret it as a \emph{structure} on an vocabulary containing a single binary symbol $E$,
where $E(x,y)$ means there is a (directed) edge from vertex $x$ to vertex $y$ (often denoted as
$xEy$ or $\mathrm{child}(x,y)$ when discussing rooted trees).

\emph{First-order logic} ($\mathsf{FO}$) allows us to build formulas using a chosen 
vocabulary, equality $(=)$, boolean symbols $(\lor,\land,\lnot)$ and quantifiers
$(\exists,\forall)$ only over the elements of the ``universe'' $V$ (i.e. the vertices).
Conversely, \emph{second-order logic} allows us to define second-order variables, which
are sets of tuples of elements of $V$, and quantify over them.

Second-order logic is more powerful in the sense that it can describe many more properties
of graphs; however, this expressive power comes at a cost, as full second-order logic can
express enough arithmetic to fall under G\"odel's famous incompleteness theorems.
Throughout this paper we instead work in \emph{monadic second-order logic}
($\mathsf{MSO}$), which is generated by the symbols of first-order logic plus quantifiers
only on \emph{sets} (i.e. 1-tuples) of variables. Importantly, this
does not allow us to quantify over pairs of vertices, i.e. to define sets of edges or
functions between vertices.

We say a graph property is \emph{definable} by a type of logic if there exists a single
\emph{sentence} (a formula where every variable is defined by a quantifier) in that logic
such that the property is the set of all graphs that satisfy this sentence; those are
called the \emph{models} of $\varphi$ and denoted $Mod(\varphi)$. Note that as most logic formalisms are
closed under conjunction, being defined by one or finitely many sentences is equivalent.

Some examples of first-order axioms are the graph being loopless, undirected, or of
bounded constant degree; some more complex properties, like being connected or being a
tree, can be described by monadic second-order logic but not first-order logic.

Additionally, we often work with graphs labeled by a finite alphabet $\Sigma$, which are
structures on a vocabulary containing the binary symbol $E$ and unary symbols
$\{S_s\}_{s\in\Sigma}$ such that $S_s(x)$ (denoted $x\in s$) represents the vertex $x$ being labeled $s$. To
make this a proper labeling, a monadic second-order axiom\footnote{As the second-order
variables here are unary relations in the vocabulary and not really variables, $\mathrm{Part}$ can also be written as a
first-order formula.} is required:
$\mathrm{Part}(\{S_s\}_{s\in\Sigma})$, where
\begin{align*}
    \mathrm{Part}(S_1,\ldots,S_k):=\forall x:\, \bigvee_{i\in[1,\ldots,k]} x\in S_1\ \ \land\qquad&(x \text{ is part of at least one set})\\
    \bigwedge_{i,j\in[1,\ldots,k],i\neq j} (x\in S_i)\to \lnot(x\in S_j)\qquad&(x\text{ is not in two sets at once})
\end{align*}
Note that this is a valid (finite) formula only if $\Sigma$ is finite. 
\begin{definition}
    We call \emph{Presburger} formulas over a set of integer variables $\mathcal{V}$ those
    generated by $\mathcal{V}$, $\mathbb{N}$, $+$, $\leq$ and $\mathrm{div}_k(p)$ for all
    $k\in\mathbb{Z}^+$.
\end{definition}
Crucially, Presburger formulas can encode \emph{linear} multiplication (i.e.
multiplication by a constant) but do not allow \emph{general} multiplication (i.e.
multiplying two variables) -- Presburger arithmetic is not expressive enough to fall under G\"odel's incompleteness theorems.
\begin{definition}
    We call \emph{Presburger $\mathsf{MSO}$} ($\mathsf{PMSO}$) formulas on the language of
    rooted trees those generated by $\mathsf{MSO}$ syntax augmented
    by propositions of the type $x/\varphi$, where $\varphi$ is a Presburger formula on the set
    of integer variables $\mathcal{V}=\{\lvert X\rvert\,: X\text{ is a second-order
    variable}\}$. The formula $x/\varphi$ holds in some directed graph under a valuation $\rho$ if the
    valuation $\mu$ mapping each variable $\lvert X\rvert$ from $\varphi$ to the cardinality
    of $\rho(X)\cap \mathrm{children}(\rho(x))$ is a solution for $\varphi$.
\end{definition}
In other words, $\mathsf{PMSO}$ can capture complex relationship between sets (defined by Presburger
formulas), but those sets are restricted to the children of a given node. Presburger
monadic second-order formulas are usually defined only over \emph{rooted trees} -- we
propose a natural extension of their definition to general undirected graphs.
\begin{definition}
    We call \emph{undirected Presburger $\mathsf{MSO}$} ($\mathsf{uPMSO}$) formulas on the
    language of undirected graphs those generated by $\mathsf{MSO}$ syntax augmented
    by propositions of the type $x/\varphi$, where $\varphi$ is a Presburger formula on the set
    of integer variables $\mathcal{V}=\{\lvert X\rvert\,: X\text{ is a second-order
    variable}\}$. The formula $x/\varphi$ holds in some graph under a valuation $\rho$ if the
    valuation $\mu$ mapping each variable $\lvert X\rvert$ from $\varphi$ to the cardinality
    of $\rho(X)\cap \mathrm{neighbors}(\rho(x))$ is a solution for $\varphi$.
\end{definition}
We will specifically look at the intersection of those classes with the class of locally
verifiable problems defined above.
\begin{definition}
    We define the class of \emph{Local First-Order} $(\mathsf{LFO})$ problems as the set of all graph
    languages which are axiomatizable by a first-order sentence \emph{and} verifiable by a
    $\mathsf{LOCAL}$ algorithm in constant-time; formally,
    $\mathsf{LFO}=\mathsf{FO}\cap\mathsf{LD}^*$.
    Similarly, we define the classes of
    \begin{itemize}
        \item \emph{Local Monadic Second-Order} problems $(\mathsf{LMSO})$,
        \item \emph{Local Presburger-Monadic Second-Order} problems $(\mathsf{LPMSO})$, and
        \item \emph{Local undirected Presburger-Monadic Second-Order} problems $(\mathsf{LuPMSO})$.
    \end{itemize}
\end{definition}
\section{Local monadic second-order problems}\label{sec:main}
The objective of this section is to prove that the class of Local Presburger Monadic Second-Order
problems on rooted trees has a complexity landscape with polynomial gaps, similar to the bounded
degree case (see \cite{balliuRooted2022}).
\gapthm*
We will prove this in two steps: first, in Section~\ref{sec:auto_gap} we show that the related problem $\Pi_Q$ of computing \emph{both} a valid
output labeling and an automata-based certificate of correctness has one of the given
complexities; $\Pi_Q$ is trivially harder than $\Pi$, but in
Section~\ref{sec:auto_redu} we show that a constant-round reduction from $\Pi$ to $\Pi_Q$ exists,
and so they have the same asymptotic complexity. 
\subsection{Background: rational sets}
We follow the notations from \cite{bonevaAutomata2005} and
\cite{csuhajMultiset2000}; to compare their results, we use semilinear sets. We use
$\mathcal{M}(Q)$ to denote the set of multisets of elements of $Q$.
\begin{definition}
    A subset $L$ of $\mathbb{N}^k$ is said to be \emph{linear} if it is of the form
    $$L=u_0+\mathbb{N}u_1+\ldots+\mathbb{N}u_n=\left\{u_0+t_1u_1+\ldots+t_nu_n\,\middle|\,t_1,\ldots,t_n\in\mathbb{N}\right\}$$
    for some $u_0,\ldots ,u_n\in\mathbb{N}^k$. In particular, $u_0$ is called the
    \emph{base vector} of $L$. A subset of $\mathbb {N}^k$ is said
    to be \emph{semilinear} if it is a finite union of linear subsets. 
\end{definition}
We can identify a multiset $M$ over a finite alphabet $\Sigma$ (with a fixed order
$\Sigma=\{s_1,\ldots,s_k\}$) with a vector in $\mathbb{N}^k$ where the $i$-th element
represents the multiplicity of $s_i$. This is called the \emph{Parikh vector} of $M$; we
do not distinguish between $M$ and its Parikh vector, and refer to a multiset language as
semilinear if and only if its Parikh image is semilinear.
\begin{definition}
    The family $\mathrm{Rat}(\mathcal{M}(A))$ of \emph{rational multiset languages} is the smallest
    subset of $\mathcal{P}(\mathcal{M}(A))$ which:
    \begin{itemize}
        \item \textbf{contains all finite languages;} every finite subset of $\mathcal{M}(A)$ is in $\mathrm{Rat}(\mathcal{M}(A))$,
        \item \textbf{is closed under finite union;} if $L_1,L_2\in\mathrm{Rat}(\mathcal{M}(A))$ then $L_1\cup L_2\in
        \mathrm{Rat}(\mathcal{M}(A))$,
        \item \textbf{is closed under finite multiset union;} if $L_1,L_2\in\mathrm{Rat}(\mathcal{M}(A))$ then $L_1+L_2=\{m_1\cup
        m_2\,|\,m_1\in L_1,m_2\in L_2\}\in\mathrm{Rat}(\mathcal{M}(A))$, and
        \item \textbf{is closed under Kleene star;} for a multiset language $L$ denote
        $L^0=\{\emptyset\}$, $L^1=L$, and $L^{i+1}=L^i+L^1$; then if
        $L\in\mathrm{Rat}(\mathcal{M}(A))$ then
        $L^*=\bigcup_{i\in\mathbb{N}}L^i\in\mathrm{Rat}(\mathcal{M}(A))$.
    \end{itemize}
\end{definition}
\begin{definition}
    A (nondeterministic) \emph{finite multiset automaton (FMA)} is a structure $\mathcal{A}=(\Sigma,Q,q_0,f,F)$,
    where:
    \begin{itemize}
        \item $\Sigma$ is the input alphabet,
        \item $Q$ is a finite set of states,
        \item $q_0$ is the initial state,
        \item $f$ is a transition function $\Sigma\times Q\to 2^Q$,
        \item $F\subseteq Q$ is the set of final states.
    \end{itemize}
    A configuration of $\mathcal{A}$ is a pair $(q,\tau)$, where $q$ is a state (the
    current state) and $\tau$ is a multiset (the current content of the bag). A
    configuration \emph{transitions} into another (denoted $(q,\tau)\vdash (q',\tau')$) if
    there is an element $s\in\Sigma$ such that 
    \begin{itemize}
        \item $\tau=\tau'\uplus \{s\}$, and
        \item $q'\in f(s,q)$.
    \end{itemize}
    The multisets $M$ accepted by this automaton are the ones for which there exists a chain
    of transitions starting from $(q_0,M)$ and ending at $(q_f,\emptyset)$ where $q_f\in F$.
\end{definition}
Note that we only defined \emph{nondeterministic} automata; contrary to word or even tree
automata, deterministic FMA are strictly weaker than their nondeterministic counterparts.
\begin{theorem}\label{thm:semilinear}
    Let $\mathcal{L}$ be a multiset language. The following are equivalent:
    \begin{enumerate}[(1)]
        \item $\mathcal{L}$ is semilinear.
        \item $\mathcal{L}$ is rational.
        \item $\mathcal{L}$ is definable by a Presburger sentence.
        \item There exists a finite multiset automaton recognising $\mathcal{L}$.
    \end{enumerate}
\end{theorem}
\begin{proof}
    $(1)\Leftrightarrow (2)$ is a common result on rational sets of commutative monoids,
    found explicitly in \cite[p.175]{eilenbergRational1969}.

    $(1)\Leftrightarrow(3)$ is found in \cite[Theorem 1.3]{ginsburgPresburger1966}.

    $(1)\Leftrightarrow (4)$ can be found in \cite{csuhajMultiset2000}.
\end{proof}
\begin{corollary}
    Rational multiset languages are closed under finite union, finite intersection and complement, corresponding to the
boolean disjunction, conjunction and negation operations on Presburger formulas.
\end{corollary}
\begin{lemma}\label{lem:semibase}
    Let $\mathcal{L}$ be a semilinear multiset language. Then we can compute a value $d(\mathcal{L})\in\mathbb{N}$
    such that for all $M\in\mathcal{L}$ there is $M'\in\mathcal{L}$ such that $M'\subseteq
    M$ and $\lvert M'\rvert \leq d(\mathcal{L})$.
\end{lemma}
\begin{proof}
    We can pick $l_0$ as the maximum length of the base vectors of the linear components
    of $\mathcal{L}$, since every element of a linear set contains its base
    vector.
\end{proof}
\subsection{Background: bottom-up tree automata}
\begin{definition}
    A \emph{rational-multiset tree automaton (RatMTA)} is a tuple $\mathcal{A}:=(\Sigma,Q,\delta,F)$ where:
    \begin{itemize}
        \item $\Sigma$ is a finite set of labels,
        \item $Q$ is a finite set of states,
        \item $\delta:\Sigma\times Q\to\mathrm{Rat}(\mathcal{M}(Q))$ is the transition function,
        which assigns to each pair of a label and a state a rational set of multisets
        of $Q$, and
        \item $F\subseteq Q$ is the accepting set of states.
    \end{itemize}
    If for each $s\in\Sigma$ the set $\{\delta(s,q)\}_{q\in Q}$ forms a partition of
    $\mathcal{M}(Q)$, the automaton is called \emph{deterministic}.

    A \emph{valid run} of $\mathcal{A}$ on an unranked, unordered rooted tree $T=((V,E),r,\lambda)$ is a labeling $\mu:V\to Q$
    such that for each node $v\in V$ the multiset
    $\{\!\{\mu(v_c):v_c\in C_v\}\!\}$ is in $\delta(\lambda(v),\mu(v))$. A rooted tree $T$ is
    \emph{accepted} by $\mathcal{A}$ iff there is an \emph{accepting run} on $T$, meaning
    a run on $T$ such that the label state of its root is in $F$.
    
    We say that a RatMTA $\mathcal{A}:=(\Sigma,Q,\delta,F)$ is \emph{deterministic} if for
    every $s\in\Sigma,\ q_1,q_2\in Q$ we have $\delta(s,q_1)\cap \delta(s,q_2)=\emptyset$;
    we say that a RatMTA is \emph{complete} if for all $s\in\Sigma$ we have $\bigcup_{q\in
    Q}\delta(s,q) = \mathcal{M}(Q)$. In a deterministic and complete RatMTA the family
    $\{\delta(s,q)\}_{q\in Q}$ is a partition of $\mathcal{M}(Q)$ for all $s\in\Sigma$.
\end{definition}
The relationship between Presburger arithmetic and rational multisets shown in
Theorem~\ref{thm:semilinear} extends to Presburger MSO formulas on trees and rational multiset tree
automata: we restate \cite[Theorem 5]{seidlPresburger2003} below. Note that while the authors
formulation uses Presburger u-tree automata, they are simply rational multiset tree automata
defined using Presburger formulas instead of rational sets, and the two are equivalent by
Theorem~\ref{thm:semilinear} (also mentioned explicitly in Section 2 of their paper).
\begin{theorem}\label{thm:PMSO_auto}
    Let $\mathcal{L}$ be a language of unranked, unordered rooted trees. $\mathcal{L}$ is
    PMSO-definable iff there exists a RatMTA $\mathcal{A}$ that accepts $\mathcal{L}$.
\end{theorem}
\subsection{Warmup: rational automata properties}\label{sec:autypes}
Through this section, we extend well-known results for \emph{finite, ordered}
bottom-up tree automata (see \cite{treeautomataTA}) to rational multiset tree automata,
and show the connection between rooted tree types (from \cite{changHierarchy2019}) and
states of an automaton. The results below hold for \emph{deterministic, complete} RatMTA:
luckily, in \cite[Lemma 2.2-2.3]{colcombetRational2003} Colcombet proved that this is not
a restriction on the languages we can consider.
\begin{theorem}\label{thm:deterministic}
    Every labeled, rooted tree language which can be recognised by a RatMTA can be
    recognised by a deterministic, complete RatMTA.
\end{theorem}
\begin{definition}
    A \emph{context} in the theory of $\Sigma$-labeled rooted trees is a
    $(\Sigma\uplus\{x\})$-labeled rooted tree which has exactly one \emph{leaf} node
    labeled as $x$; this is called the variable node. 
\end{definition}
We denote contexts as $C[x]$; given a rooted tree $T$, we can \emph{substitute} $T$ into
$C[x]$ by identifying the root of $T$ with the variable node of $C[x]$ -- we denote this
as $C[T]$. We can imagine contexts having a pole in place of the variable node, and any 
rooted tree having a pole at the root; the operation of \emph{replacing} a unipolar tree $T_1$
by another $T_2$ is then expressed through contexts as going from $C[T_1]$ to $C[T_2]$,
where $C$ is the rest of the graph.
\begin{figure}[th]
\centering
\includegraphics[width=0.8\textwidth]{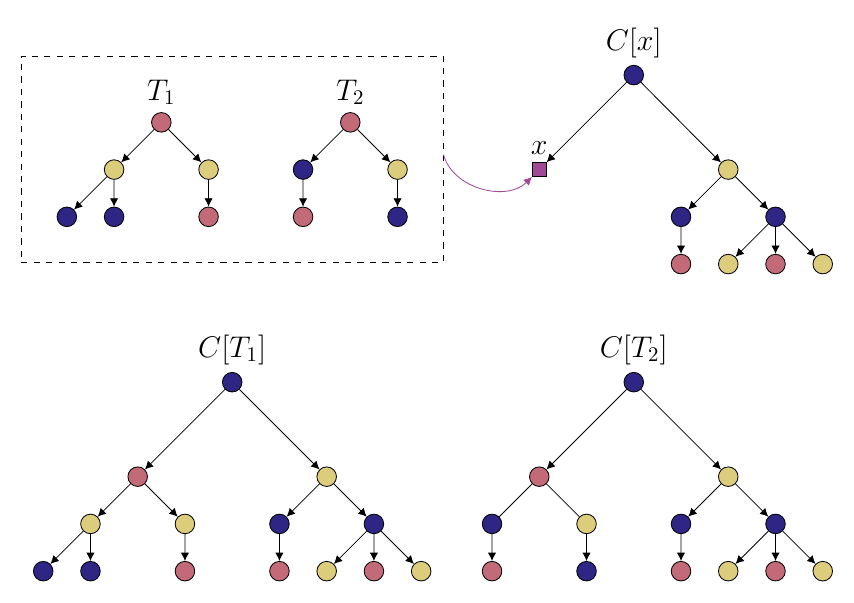}
\caption{Replacing $T_1$ by $T_2$ in the context $C[x]$.}
\label{fig:treereplace}
\end{figure}
\begin{definition}
    Given a graph language $\mathcal{G}$, we define an equivalence relation
    $\equiv_{\mathcal{G}}$ as follows: $T_1\equiv_{\mathcal{G}} T_2$ if and only if for every
    context $C$ we have $C[T_1]\in\mathcal{G}$ if and only if $C[T_2]\in\mathcal{G}$.

    In other words, $T_1\equiv_{\mathcal{G}} T_2$ if replacing $T_1$ by $T_2$ in a
    tree does not affect membership in $\mathcal{G}$.
\end{definition}
Given a \emph{deterministic, complete} rational multiset tree automaton $\mathcal{A}$, for
every rooted tree $T$ we have exactly one run of $\mathcal{A}$ (accepting or not); we
define the \emph{type} of $T$ (denoted $\mathrm{Type}(T)$) as the state of its root.
Crucially, since the transition function only looks at the states and not any input
labels, if we have $\mathrm{Type}(T_1)=\mathrm{Type}(T_2)$ we can replace $T_1$ by $T_2$
in any context $C[x]$ without changing the run of $\mathcal{A}$ on the nodes of $C[x]$.
\begin{lemma}\label{lem:treereplace}
    Let $\mathcal{A}=(\Sigma,Q,\delta,F)$ be a deterministic, complete rational multiset tree automaton, $C[x]$ a context, and
    $T_1,T_2$ $\Sigma$-labeled rooted trees: then
    $$\mathrm{Type}(T_1)=\mathrm{Type}(T_2)\ \ \Rightarrow\ \ \mathrm{Type}(C[T_1])=\mathrm{Type}(C[T_2]).$$
\end{lemma}
\begin{proof}
    Let $r_1$ be the run of $\mathcal{A}$ on $T_1$, $r_2$ be the run of $\mathcal{A}$ on
    $T_2$, and $r$ be the run of $\mathcal{A}$ on $C[T_1]$ restricted to the nodes of
    $C[T_1]\smallsetminus T_1$; since $\mathrm{Type}(T_1)=\mathrm{Type}(T_2)$ and the
    transition function depends only on the states, $r\cup r_2$ is the run of
    $\mathcal{A}$ on $C[T_2]$. If $r$ is empty ($C$ is the trivial context where the root
    is the variable node) then the types of $T_1,T_2,C[T_1],C[T_2]$ are all the same;
    else, $r$ assigns the same state to both roots and $\mathrm{Type}(C[T_1])=\mathrm{Type}(C[T_2])$.
\end{proof}
\begin{corollary}\label{cor:typereplace}
        Let $\mathcal{A}=(\Sigma,Q,\delta,F)$ be a rational multiset tree automaton
        recognising a graph language $\mathcal{G}$, and let $T_1,T_2$ be $\Sigma$-labeled rooted trees: then
    $$\mathrm{Type}(T_1)=\mathrm{Type}(T_2)\ \ \Rightarrow\ \ T_1\equiv_{\mathcal{G}} T_2$$
\end{corollary}
We prove something stronger, a generalization of the Myhill-Nerode theorem (see
\cite[p36]{treeautomataTA}) for rational multiset tree automata: if a language is
recognizable by any deterministic, complete RatMTA then $\equiv_{\mathcal{G}}$ has
finitely many equivalence classes, and those classes form the states of a deterministic,
complete RatMTA which is \emph{minimal} (by number of states) and \emph{unique} up to
relabeling the states.
\begin{restatable}{theorem}{minratmtathm}\label{thm:minimal-ratmta}
    \textbf{\emph{(Myhill-Nerode)}} Let $\mathcal{G}$ be a tree language which can be recognised by a deterministic,
    complete RatMTA $\mathcal{A}$. Then:
    \begin{enumerate}[(1)]
        \item $\equiv_{\mathcal{G}}$ has finite index, upper bounded by the number of states
        of $\mathcal{A}$, and
        \item there is a deterministic, complete RatMTA which recognises $\mathcal{G}$
        whose states are in bijection with the equivalence classes of $\equiv_\mathcal{G}$.
    \end{enumerate}
\end{restatable}
\begin{proof}
    Found in Appendix~\ref{app:myhill}.
\end{proof}
But what about bipolar trees? Contexts are themselves
rooted trees, so we can imagine them as having one pole at the variable node (going out)
and one pole at the root (going in); specifically, we are interested in \emph{direct
contexts}, which have the variable node as a direct child of the root node -- in other
words, bipolar trees which have both poles at the root. After accounting for the
direction, these function similarly to \emph{rooted tree types} by the definition given by
Chang and Pettie in \cite{changHierarchy2019}; we can encode any context (bipolar
tree) as a string of direct contexts, each corresponding to one node on the path
between the root and variable node -- we define a notion of \emph{type} for them
corresponding to \emph{bipolar tree types} from \cite{changHierarchy2019}.

We can interpret contexts, and more specifically direct contexts, as functions
between types; by the previous lemma if $\mathrm{Type}(T_1)=\mathrm{Type}(T_2)$ then
$\mathrm{Type}(C[T_1])=\mathrm{Type}(C[T_2])$. As we have finite types, we have finitely
many functions between types; these functions form our \emph{context types}.
\begin{definition}
    The \emph{type} of a context $C$ with respect to a deterministic, complete automaton
    $\mathcal{A}$ is the function
    $$\mathrm{Type}(C):\mathrm{Type}(T)\mapsto \mathrm{Type}(C[T]).$$
\end{definition}
Context types have the same replacement properties as types: we can replace a context $C_1$
(directed bipolar tree) with one of the same type $C_2$ such that the run of $\mathcal{A}$ is
the same on all nodes outside of $C_1,C_2$.
\begin{lemma}\label{lem:contextreplace}
    Let $\mathcal{A}=(\Sigma,Q,\delta,F)$ be a deterministic, complete rational multiset
    tree automaton, let $C[x],C_1[x],C_2[x]$ be contexts, and let $T$ be a
    $\Sigma$-labeled rooted tree: then
    $$\mathrm{Type}(C_1)=\mathrm{Type}(C_2)\ \ \Rightarrow\ \ \mathrm{Type}(C[C_1[T]])=\mathrm{Type}(C[C_2[T]]).$$
\end{lemma}
\begin{proof}
    Trivially $\mathrm{Type}(T)=\mathrm{Type}(T)$; since
    $\mathrm{Type}(C_1)=\mathrm{Type}(C_2)$, we have that 
    $$\mathrm{Type}(C_1[T])=\mathrm{Type}(C_1)(\mathrm{Type}(T))=\mathrm{Type}(C_2)(\mathrm{Type}(T))=\mathrm{Type}(C_2[T]),$$
    and by Lemma~\ref{lem:treereplace} $\mathrm{Type}(C[C_1[T]])=\mathrm{Type}(C[C_2[T]])$.
\end{proof}
\subsubsection{Finitely representing types}\label{sec:fintypes}
Each tree is associated to a type by a complete, deterministic rational multiset tree automaton $\mathcal{A}=(\Sigma,Q,\delta,F)$;
this type is the state assigned to the root by the automaton. If we focus on the
transition function, we see that there are two factors determining the possible states of
children of the root: the root's \emph{label} and its \emph{state} -- to each pair of a
label and state, we naturally associate a semilinear set $\delta(s,q)$. By
Lemma~\ref{lem:semibase} we know there is a value $d(\delta(s,q))$, which we simply denote
$d(s,q)$, such that every element of $\delta(s,q)$ has a subset of size $\leq d(s,q)$ that
is also in $\delta(s,q)$. We define $\Delta_R$ to be the maximum value of $d(s,q)$ over
all choices of $s\in\Sigma$ and $q\in Q$; we will use this property during the proof to
\emph{reduce} the number of children a node has while maintaining its state. Formally:
\begin{lemma}\label{lem:redtree}
    Let $T$ be a labeled rooted tree, let $s\in\Sigma$ be the label of its root, and let $\mathcal{C}(T)$ be the set of subtrees rooted at the
    children of the root of $T$. Then there is a labeled rooted tree $T'$ such that:
    \begin{enumerate}[(1)]
        \item the root of $T'$ is labeled $s$,
        \item $\mathrm{Type}(T')=\mathrm{Type}(T)$,
        \item $\mathcal{C}(T')\subseteq\mathcal{C}(T)$, and
        \item $\lvert\mathcal{C}(T')\rvert\leq\Delta_R$.
    \end{enumerate}
\end{lemma}
\begin{proof}
Let $M(\mathcal{C}(T))$ be the multiset of states of the children of the root of $T$. Apply
Lemma~\ref{lem:semibase} to $M(\mathcal{C}(T))\in\delta(s,\mathrm{Type}(T))$; by definition of
$\Delta_R$ there is a $M'\subseteq M$ such that $M'\in\delta(s,\mathrm{Type}(T))$ and $\lvert M'\rvert\leq
d(s,\mathrm{Type}(T))\leq\Delta_R$. Observe that $M$ is a bijection between
$\mathcal{C}(T)$ and $M(\mathcal{C}(T))$, and define $C:=M^{-1}(M')$.
Construct a tree $T'$ with a $s$-labeled root and children the roots of elements of $C$;
then (1) and (3) hold by definition, (2) holds since $M'\in\delta(s,\mathrm{Type}(T))$ and
(4) holds since $\lvert C\rvert=\lvert M'\rvert\leq\Delta_R$.
\end{proof}
We would like to define a similar procedure for context types: however, they cannot be
easily represented by a single semilinear set. Instead, we restrict to contexts where the
variable node is an immediate child of the root, which we call \emph{direct contexts}. We can identify a
direct context with a pair of the label of its root and the multiset of types of the other children of the
root; then a direct context \emph{type} can be represented by the set of all pairs
representing its elements or as a function of types (see previous section).
\begin{lemma}\label{lem:contypes}
    Let $\mathcal{A}=(\Sigma,Q,\delta,F)$ be a deterministic, complete RatMTA. Then there
    is a value $\Delta_C\in\mathbb{N}$, computable from the description of $\mathcal{A}$,
    such that for every pair $(s,M)$ representing a direct context type $f$ we can compute
    a multiset $M'$ such that: 
    \begin{itemize}
        \item $(s,M')$ represents $f$,
        \item $\lvert M'\rvert\leq \Delta_C$, and
        \item $M'\subseteq M$.
    \end{itemize}
\end{lemma}
\begin{proof}
Consider a pair of states $q,q_f\in Q$; we want to show that the set $\mathrm{contexts}(s,q,q_f)$
of all direct contexts $C$ (represented as multisets of states) such that $\mathrm{Type}(t)=
q\Rightarrow \mathrm{Type}(C[t])=q_f$ and the root of $C$ is labeled $s$ is a semilinear set.
The set $\delta_{det}(s,q_f)$ is semilinear: we iterate over its linear components
$L_{s,q_f,1}\cup\ldots\cup L_{s,q_f,m}$ and remove one instance of $q$ where possible.

Formally, let 
$$L=\mathbf{b}+\mathbb{N}\mathbf{a_1}+\ldots+\mathbb{N}\mathbf{a_w}$$
where $\mathbf{b},\mathbf{a_1},\ldots,\mathbf{a_w}$ are $Q$-indexed vectors of
$\mathbb{N}$; let $\pi_q$ be the projection to the index $q$. Let $\mathbf{q}$ be the
$q$-th base vector, which is $0$ everywhere except at $q$ where it is $1$. Then 
$$\mathrm{remove}_q(L)=\begin{dcases}
    (\mathbf{b}-\mathbf{q})+\mathbb{N}\mathbf{a_1}+\ldots+\mathbb{N}\mathbf{a_w} & \text{if
        }\pi_q(\mathbf{b})\neq 0\\
    \emptyset & \text{if }\pi_q(\mathbf{b})=0,\,\forall i\leq w\,\pi_q(\mathbf{a_i})=0\\
    \bigcup_{\substack{1\leq i\leq l\\\pi_q(\mathbf{a_i})\neq 0}}(\mathbf{b}+\mathbf{a_i}-\mathbf{q})+\mathbb{N}\mathbf{a_1}+\ldots+\mathbb{N}\mathbf{a_w} &\text{otherwise}
\end{dcases} $$
Then we can define $\mathrm{contexts}(s,q,q_f)$ as the union of all the remove$_q(L_{s,q_f,i})$; this
is a finite union of semilinear sets, so it is semilinear.

Given a context type as a function of types $f:Q\to Q$, we can explicitly represent it as
a set of pairs $\{(q,f(q))\}_{q\in Q}$; then the set of all direct contexts of type $f$
which have their root labeled $s$ is the set
$$\mathrm{contexts}(f,s):=\bigcap_{q\in Q}\mathrm{contexts}(s,q,f(q))$$
which is a finite intersection of semilinear sets, so it is semilinear. Then by
Lemma~\ref{lem:semibase} there is a value $D(f):=D(\mathrm{contexts}(f,s))$ such that for
all $M\in\mathrm{contexts}(f,s)$ we can find a $M'\subseteq M$ such that
$M'\in\mathrm{contexts}(f,s)$ and $\lvert M'\rvert \leq D(f)$. Finally, we
have finitely many possible functions $f$ and labels $s$, so we explicitly
define $\Delta_C$ as the maximum value of $D(f)$ over all $\mathrm{contexts}(f,s)$.
\end{proof}
\begin{corollary}\label{cor:replacecontext}
    Let $C$ be a labeled direct context, let $s\in\Sigma$ be the label of its root, and
    let $\mathcal{C}(C)$ be the set of subtrees rooted at the non-variable 
    children of the root of $C$. Then there is a labeled direct context $C'$ such that:
    \begin{enumerate}[(1)]
        \item the root of $C'$ is labeled $s$,
        \item $\mathrm{Type}(C')=\mathrm{Type}(C)$,
        \item $\mathcal{C}(C')\subseteq\mathcal{C}(C)$, and
        \item $\lvert\mathcal{C}(C')\rvert\leq\Delta_C$.
    \end{enumerate}
\end{corollary}
\subsection{Constructing automata certificates}\label{sec:auto_gap}
Let $\mathcal{A}=(\Sigma,Q,\delta,F)$ be a rational multiset tree automaton, and
let $\mathcal{G}_\mathcal{A}$ be the class of all rooted trees which are labeled with an \emph{accepting} run
for $\mathcal{A}$. We define the \emph{automaton certification} problem associated with $\mathcal{A}$ as
$\Pi_\mathcal{A}:=(\Sigma,Q,\mathcal{G}_\mathcal{A})$, and denote the class of all
problems associated to RatMTA as \emph{automata certification problems}.
\autogapthm*
The proof of Theorem~\ref{thm:autogap} takes up all of Section~\ref{sec:auto_gap}; a high
level overview of it, including references to the specific sections, can be found in Section~\ref{sec:results}. 
\subsubsection{Preprocessing: partially labeled tree types}\label{par:types}
We denote the \emph{type} of an unlabeled rooted tree as the set of all labelings of the root that can
appear in a valid labeling of the rooted tree. Our choice of notation seems to conflict
with the definition of \emph{type} given in Section~\ref{sec:autypes}; however, by
constructing a suitable deterministic and complete automaton $\mathcal{A}_{det}$ which
uses sets of states of $\mathcal{A}$ as states, the two notations coincide. This
construction can also be used as a proof of Theorem~\ref{thm:deterministic}, as the two
automata recognize the same language of rooted trees.

Let $\mathcal{A}=(\Sigma,Q,\delta,F)$; we define $Q_{det}=2^Q$, and define an operation
$h:\mathcal{M}(Q)\to 2^{\mathcal{M}(2^Q)}$ which
sends a multiset $M$ of states $Q$ to the set of all multisets $M'$ of sets of states $2^Q$ for
which there is a \emph{choice function}, a bijection $c:M'\to M$ which sends a set to
one of its element. Then $F_{det}:=h(F)=\bigcup_{Q\in F}h(Q)$ and for all $s\in\Sigma$ we define
$$\delta_{det}(s,H)=\bigcap_{q\in H}h(\delta(s,q))\cap \bigcap_{q\in Q\smallsetminus H}h(\delta(s,q))^C$$
where the complement is taken with respect to the set of all multisets of $2^{Q}$. In simpler words, instead of labeling nodes with the states of a run, we label them
with exactly the set of states they might receive in \emph{any} run. 

Observe that $h$ sends semilinear sets to semilinear sets; the simplest way to observe that is
taking the multiset automaton representation of $M$ (see Theorem~\ref{thm:semilinear}) and
replacing every transition corresponding to $q$ with transitions corresponding to all sets
of states containing $q$ -- since additionally semilinear sets are closed under finite
intersection and complement, $\mathcal{A}_{det}=(\Sigma,2^Q,\delta_{det},F_{det})$ is a
rational multiset tree automaton.Formally, observe that for any $s\in\Sigma$ the family $\{\delta_{det}(s,H)\}_{H\subseteq Q}$ are the atoms of the
$\sigma$-algebra generated by the family $\{h(\delta(s,q))\}_{q\in Q}$ over all multisets
of $2^Q$; then they form a partition of $\mathcal{M}(2^Q)$ and $\mathcal{A}_{det}$ is
deterministic and complete.

We call a tree which is $\Sigma$-labeled \emph{unlabeled}, in the sense that it has not
been assigned an output label from $Q$. Then the \emph{type} of an \emph{unlabeled} rooted tree $T$ by the definition above corresponds to the
\emph{type} of $T$ according to $\mathcal{A}_{det}$. Recall
Corollary~\ref{cor:typereplace}: we can replace a rooted subtree by one of the same type
without affecting the overall labeling. During the run of our algorithm, however, we want
to compute the type of \emph{partially labeled} rooted trees, for which some states have
already been assigned: to do so, we further modify $\mathcal{A}_{det}$.

Define an automaton
$\mathcal{A}_{lab}=(\Sigma\times(Q\uplus\{*\}),2^Q,\delta_{lab},F_{lab})$ over
$\Sigma\times(Q\uplus\{*\})$-labeled rooted trees,
where $*$ represents unlabeled nodes. Explicitly, we define
\begin{gather*}
    F_{lab}=\left\{Q_F\in 2^Q\,\middle|\,Q_F\cap F\neq\emptyset\right\}\\
    \forall s\in \Sigma:\,\delta_{lab}((s,*),Q_0)=\delta_{det}(s,Q_0)\\
    \forall s\in \Sigma:\,\delta_{lab}((s,q),Q_0)=\begin{cases}
        \bigcup_{Q_1\ni q}\delta_{det}(s,Q_1)&\text{ if }Q_0=\{q\}\\
        \bigcup_{Q_1\not\ni q}\delta_{det}(s,Q_1)&\text{ if }Q_0=\emptyset\\
        \emptyset&\text{ otherwise}
    \end{cases}
\end{gather*}
Since semilinear sets are closed under finite union, this is a rational multiset tree
automaton; additionally, for any $(s,q)\in \Sigma\times Q$ the sets $\delta_{lab}((s,q),\{q\})$ and
$\delta_{lab}((s,q),\emptyset)$ form a partition of $\mathcal{M}(2^Q)$ while the others are
empty -- then this is still a deterministic, complete automaton. We will define the
\emph{type} of a partially $Q$-labeled rooted tree $T$ as the
type of $T$ according to the automaton $\mathcal{A}_{lab}$ -- we could equivalently define
it as the set of all output labels of its root that can appear in a full labeling which is
coherent with the given partial labeling, but we will leave the proof of this equivalence
for Lemma~\ref{lem:typeauto}. Similarly, we define context
types and direct context types according to $\mathcal{A}_{lab}$ as in Section~\ref{sec:autypes}.
\subsubsection{Degree reduction by forgetting some children}\label{par:degred}
We make a few remarks about the similarities and differences between the types described
in the previous section and concepts used in the bounded-degree versions of this proof.
\begin{remark}
    We can roughly equate \emph{direct context types} with \emph{rooted tree classes} from
    \cite{treesChang2020} and \emph{context types} with \emph{bipolar tree types} by
    replacing the variable node with a pole and adding another pole at the root. In the
    same way as we can represent a bipolar tree as a string of rooted trees $H=(T_1,\ldots,T_c)$,
    we can represent a context as a composition of direct contexts $C=C_1\circ\ldots\circ C_c$,
    and both of those operation distribute over types.
\end{remark}
In fact, the equivalence relation induced by two contexts having the same type is a
\emph{coarsening} of the one obtained by two bipolar trees having the same type; this is
necessary, as without a bound on the degree there are infinitely many bipolar tree type
equivalence classes. However, this coarsening still satisfies one important property: we
can replace a bipolar tree with another of the same context type without changing the
types of the rest of the graph (see Corollary~\ref{cor:typereplace} and
Lemma~\ref{lem:contextreplace}).

One result of the coarsening of the equivalence relation is that two trees can be
equivalent even if they have a different topology (including input labels), specifically
different maximum degrees.
In fact, Lemmas~\ref{lem:redtree} and~\ref{lem:contypes} show that we can use the
properties of semilinear sets to replace high-degree partially labeled rooted trees and direct
contexts with low-degree rooted trees and direct contexts of the same types. We make this
explicit by defining two versions of a procedure called $\mathrm{ForgetChildren}$.
\begin{itemize}
    \item $\mathrm{ForgetChildren}_R(v)$: consider a node $v$ which is aware of the types of all of its
children, and let $T_v$ be the subtree rooted at $v$. By Lemma~\ref{lem:redtree} we can
find a tree $T'_v$ of the same type, sharing the same input and partial output labels for the root, and whose children are a subset of the
children of $v$; note that because we are already aware of the types of all children, we
can find this subset of children without any extra communication. Procedure $\mathrm{ForgetChildren}_R(v)$ causes the node to ignore the
edges corresponding to the children which are \emph{not} in this subset, virtually
replacing itself with $T'_v$, until procedure $\mathrm{Remember}$ is called.
    \item $\mathrm{ForgetChildren}_C(v)$: consider a node $v$ which is aware of the types of all but one of its
children. Treat the unknown child as a variable node and $v$ as the root of a direct
context $C_v[x]$: by Corollary~\ref{cor:replacecontext} there is a direct context $C'_v[x]$ with the
same type and sharing the same input and partial output labels for the root, such that the nonvariable children of $C'_v$ are a subset of
those of $C_v$. Procedure $\mathrm{ForgetChildren}_C(v)$ causes the node to ignore the
edges corresponding to the children which are \emph{not} in this subset, virtually
replacing itself with $C'_v$, until procedure $\mathrm{Remember}$ is called.
\end{itemize}
These procedures reduce the degree of a node $v$ to either $\leq\Delta_R+1$ (if
$\mathrm{ForgetChildren}_R(v)$ is called -- the possible incoming edge of $v$ is not
considered) or $\leq\Delta_C+2$ (if
$\mathrm{ForgetChildren}_C(v)$ is called -- the edges going into $v$ and to the variable
node are not considered). However, we need to compute the types of almost all children of a
node to apply them, which is not trivial for trees of high height; to help, we will show
that we can \emph{precompute} certain labels to ``split'' our tree into pieces of more
manageable size.
\subsubsection{Feasible functions}\label{par:feasfunc}
We define $\Pi_{\mathcal{A},\Delta}$ as the natural restriction of $\Pi_\mathcal{A}$ to the
family of rooted trees of maximum degree $\leq \Delta := \max\{\Delta_R+1,\Delta_C+2\}$
(which we call the \emph{reduction degree} of $\Pi_{\mathcal{A}}$). This is a $\mathsf{LCL}$ with checkability
radius $1$ on rooted trees, which we can convert into a problem on unrooted trees with
inputs by encoding the direction in the input labels\footnote{To make sure the input
labels are on nodes, the conversion splits every edge in four parts; so
$\Pi_{\mathcal{A}}$ has checkability radius 1 on rooted trees, but 4 on unrooted trees.
Note, however, that a feasible function labeling the $3$-hop neighborhood of an edge in
the unrooted tree version will label the $0$-hop neighborhood of the related edge in the
rooted tree version.}; see Appendix~\ref{app:conversion} for details.

Ideally, we would like to apply the $\mathrm{ForgetChildren}$ procedures we just described
to the whole tree instance, so that we can apply the algorithm for
$\Pi_{\mathcal{A},\Delta}$ directly. However, as observed above, computing the types of
every node in the graph is not efficient, as in general we cannot find the type of a
parent node before its child.

Instead of treating the algorithm for $\Pi_{\mathcal{A},\Delta}$ as a black-box, we use to
extract a \emph{feasible function}, which allows us to precompute some labels for nodes that know the types of all
their children but one -- by the way the partially labeled type automaton
$\mathcal{A}_{lab}$ is defined, assigning a state (output label) $q$ to a node fully determines its
type as $\{q\}$ as long as the label was chosen in a way that can be extended to a valid labeling; otherwise its type
will be $\emptyset$. Computing the type of a parent node before its child in this way is
necessary to break up long directed paths, where we cannot afford to wait to label the 
start of the path until all the other nodes have chosen a label one by one. 
\begin{definition}
Let $H$ be a partially labeled bipolar tree with core path $(v_1,\ldots,v_x)$
of length $x \geq \ell$ (where $\ell$ is a constant value depending only on the description of
$\Pi_{\mathcal{A},\Delta}$). Let $e$ be the middle edge between $v_{\lfloor x/2\rfloor}$
and $v_{\lfloor x/2\rfloor+1}$, and assume that $N^{r-1}(e)$ is unlabeled; let $f$ be a
function assigning labels to the elements of $N^{r-1}(e)$. Then $f$ is a
\emph{feasible} function for $H$ if for every partially labeled graph $(G,l)$ which
contains $H$ as a subgraph and for $l'$ the labeling obtained by extending $l$ with the
labels of $f$, if $l$ can be extended to a full solution then $l'$ can be extended to a
full solution.
\end{definition}
A function is $\infty$-feasible if it is feasible for every partially labeled bipolar tree $H$.
\begin{lemma}\label{lem:decidelog}
    For a $\mathsf{LCL}$ problem $\Pi$ on undirected trees, the following are equivalent:
    \begin{enumerate}[(1)]
        \item $\Pi$ is solvable in $n^{o(1)}$ rounds in the randomised $\mathsf{LOCAL}$ model,
        \item $\Pi$ is solvable in $O(\log n)$ rounds in the deterministic $\mathsf{LOCAL}$ model, and
        \item there is a $\infty$-feasible function for $\Pi$.
    \end{enumerate}
    Additionally, whether $(3)$ holds is decidable.
\end{lemma}
\begin{proof}
    This statement is a summary of Lemmas 14,15,16 from \cite{changHierarchy2019}.
\end{proof}
For problems of higher complexities, we require a weaker concept of feasibility. We say
that a function $f$ is $1$-feasible if it is feasible for any \emph{unlabeled} bipolar
tree; observe that applying $f$ to a bipolar tree can (and often will) change its type,
and even create trees of types that are never realized by unlabeled trees; but those types can be
accounted for and computed by the partially labeled tree automaton $\mathcal{A}_{lab}$.

We say a partially labeled tree has \emph{labeling layer} $k$ if it can be obtained by
applying $f$ in parallel to arbitrarily many bipolar subtrees of labeling layer $k-1$,
where unlabeled trees have labeling layer $1$. Then a function is $k$-feasible if it is
feasible on all trees of labeling layer $k$. Note that up to replacing subtrees with some
of the same type, trees of labeling layer $k$ correspond to trees in $\mathscr{H}_k$ from
\cite{treesChang2020}.
\begin{lemma}\label{lem:decidek}
    For a LCL problem $\Pi$ on undirected trees, the following are equivalent:
    \begin{enumerate}[(1)]
        \item $\Pi$ is solvable in $o(n^{1/(k-1)})$ rounds in the randomised $\mathsf{LOCAL}$ model,
        \item $\Pi$ is solvable in $O(n^{1/k})$ rounds in the deterministic $\mathsf{LOCAL}$ model, and
        \item there is a $k$-feasible function for $\Pi$.
    \end{enumerate}
    Additionally, whether $(3)$ holds is decidable.
\end{lemma}
\begin{proof}
    This statement is a summary of Lemmas 15, 17, and 18 from \cite{treesChang2020}.
\end{proof}
\subsubsection{Algorithm}\label{par:algorithm}
We state two theorems formalising the relation between the complexities of
$\Pi_{\mathcal{A}}$ and $\Pi_{\mathcal{A},\Delta}$; note that the theorems treat
$\Pi_{\mathcal{A},\Delta}$ as both a problem on rooted trees and on undirected trees
interchangeably -- this is justified in Appendix~\ref{app:conversion}. 
\begin{theorem}\label{thm:logred}
    Let $\Pi_{\mathcal{A}}$ be an automaton certification problem, and let $\Pi_{\mathcal{A},\Delta}$ be its $\mathsf{LCL}$
    restriction to the reduction degree. Then the following are equivalent:
    \begin{enumerate}[(1)]
        \item $\Pi_{\mathcal{A},\Delta}$ is solvable in $O(\log n)$ rounds in the
        deterministic or randomised $\mathsf{LOCAL}$ model,
        \item there is a $\infty$-feasible function for $\Pi_{\mathcal{A},\Delta}$;
        \item $\Pi_{\mathcal{A}}$ is solvable in $O(\log n)$ rounds in the
        deterministic and randomised $\mathsf{LOCAL}$ model.
    \end{enumerate}
\end{theorem}
\begin{theorem}\label{thm:polyred}
    Let $\Pi_{\mathcal{A}}$ be an automaton certification problem, and let $\Pi_{\mathcal{A},\Delta}$ be its $\mathsf{LCL}$
    restriction to the reduction degree. Then the following are equivalent:
    \begin{enumerate}[(1)]
        \item $\Pi_{\mathcal{A},\Delta}$ is solvable in $O(n^{1/k})$ rounds in the
        deterministic or randomised $\mathsf{LOCAL}$ model,
        \item there is a $k$-feasible function for $\Pi_{\mathcal{A},\Delta}$;
        \item $\Pi_{\mathcal{A}}$ is solvable in $O(n^{1/k})$ rounds in the
        deterministic and randomised $\mathsf{LOCAL}$ model.
    \end{enumerate}
\end{theorem}
We prove Theorems~\ref{thm:logred} and~\ref{thm:polyred} at the same time; in both cases,
$(1)\Leftrightarrow(2)$ by the previous section and $(3)\Rightarrow(1)$ since any algorithm
for $\Pi_{\mathcal{A}}$ also solves $\Pi_{\mathcal{A},\Delta}$. We need to prove
$(2)\Rightarrow(3)$: the algorithms we present are identical save for the choice of graph
decomposition.

We work on $(\gamma,\ell,L)$-decompositions; their definition for unbounded degree undirected trees
are provided in \cite{schmidLFL2026}, but we give a slightly stronger definition which
requires the orientation of the layers to be consistent with the direction of the rooted
tree.
\begin{definition}
    A $(\gamma,\ell,L)$-decomposition of a rooted tree is a partition of the nodes into
    $2k+1$ ordered layers $R_1<C_1<R_2<C_2<\ldots<R_k<C_k<R_{k+1}$ which satisfies:
    \begin{enumerate}
        \item every $R_i$ induces a subgraph where every connected component is a rooted
        tree of height $\leq\gamma$, and only the root can have a neighbor in a
        greater layer (with respect to $<$),
        \item every $C_i$ induces a subgraph where every connected component is a
        \emph{consistently directed} path of length $\in[\ell,2\ell]$, and only the
        endpoints can have a neighbor in a greater layer.
    \end{enumerate}
\end{definition}
The important parameter here is the number of layers $L$, as each new layer will require a
stronger feasible function. If we have a $\infty$-feasible function this is not a problem:
we can compute a $(1,\ell,L)$ decomposition with $L\in O(\log n)$ for any $n$ (where
$\ell$ is the same parameter that appears in the feasibility definition).
\begin{restatable}{lemma}{logdec}\label{lem:logdec}
    For any positive integer $\ell$ there is a $L(n)\in O(\log n)$ such that we
    can compute a $(1,\ell,L)$-decomposition in $O(\log n)$ rounds.
\end{restatable}
If instead we have a $k$-feasible function, we compute a
$(\gamma,\ell,k)$-decomposition where $\gamma\in O(n^{1/k})$.
\begin{restatable}{lemma}{polydec}\label{lem:polydec}
    For any positive integers $k,\ell$ and $\gamma=n^{1/k}(2\ell)^{1-1/k}$ we can
    compute a $(\gamma,\ell,k)$-decomposition in $O(n^{1/k})$ rounds.
\end{restatable}
The proofs of Lemmas~\ref{lem:logdec} and~\ref{lem:polydec} are simple adaptations of the
existing ones to the unbounded degree rooted setting and can be found in Appendix~\ref{app:decomposition}.

From here, the steps are the same; we iterate through the layers, starting from $R_1$, 
compute the type of all nodes according to $\mathcal{A}_{lab}$, and apply $\mathrm{ForgetChildren}$$_R$ to all vertices of
$R_1$. Then we reach $C_1$: each connected component is a directed path, and each node
knows the types of all its children but one (which may be in $C_1$ itself or, for the last
vertex, in a higher layer). We apply $\mathrm{ForgetChildren}$$_C$ to all of its vertices;
now $C_1\cup R_1$ is a bipolar tree of degree $\leq\Delta$ and core path length
$\in[\ell,2\ell]$ -- we apply the feasible function $f$ for $\Pi_{\mathcal{A},\Delta}$ to assign
labels to two of the central nodes. Because $f$ is $1$-feasible (and equivalently,
$\Pi_{\mathcal{A},\Delta}$ is solvable) then the labeling computed can be extended to a
valid labeling; then the type of nodes that were assigned a label state $q\in Q$ are of
type $\{q\}$.

Nodes that are \emph{ancestors} of the prelabeled nodes can now
compute their own types one after the other (this takes at most $\ell$ rounds) while nodes
which are \emph{descendants} of the prelabeled nodes join $R_2$, delaying 
their type computation to the next step. We repeat the steps above for each layer: compute
all types and reduce degree for $R_i$; compute context types, reduce degree, precommit
some labels, compute types for $C_i$, until we have computed types for all nodes of the
input graph.

To complete the algorithm, every node applies $\mathrm{Remember}$ to return to the full
tree, then the root chooses an accepting state from its type and labels
itself with it; then every node that already has an assigned state (root and prelabeled nodes)
chooses labels for its children. We repeat this until the tree is fully labeled, which we
show also takes $O((\gamma+\ell)L)$ rounds in Lemma~\ref{lem:complex}. We prove that this choice is possible and
leads to a correct labeling in Lemmas~\ref{lem:typeauto} and~\ref{lem:correct}. 
\begin{figure}[th]
\centering
\includegraphics[width=0.8\textwidth]{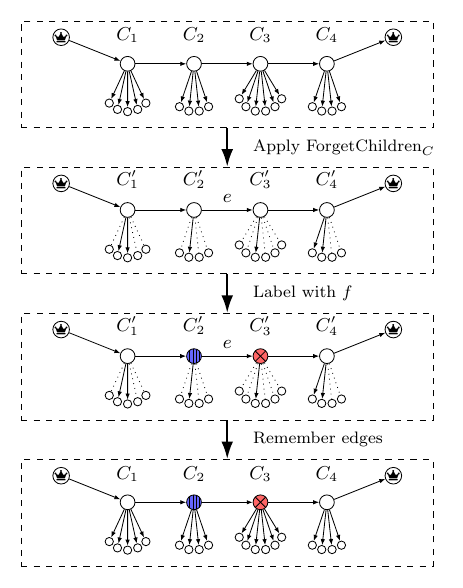}
\caption{Using $\mathsf{ForgetChildren}$ and the feasible function to label bipolar trees.}
\label{fig:forgetchildren}
\end{figure}
\begin{lemma}
    Let $\Pi_{\mathcal{A}}$ be the problem associated to the rational multiset tree
    automaton $\mathcal{A}=(\Sigma,Q,\delta,F)$. Let
    $\mathcal{A}_{lab}=(\Sigma\times(Q\uplus\{*\}),2^Q,\delta_{lab},F_{lab})$ be the partially labeled
    type automaton of $\mathcal{A}$. Given a run $r$ of $\mathcal{A}_{lab}$ on a
    $\Sigma$-labeled and partially
    $Q$-labeled rooted tree $(G,l)$ \emph{which does not label the root with $\emptyset$} we can find a run $r'$\footnote{Note that those runs need to be valid, ie respect the transition
    function, but not necessarily accepting, ie have an accepting state at the root.} of $\mathcal{A}$ on the underlying
    unlabeled rooted tree $G$ such that for all nodes $v$ of $G$, $r'(v)\in
    r(v)$ and $r(v)$ extends the partial $Q$-labeling of $G$.
\end{lemma}
\begin{proof}
    First of all, observe that by definition of
    $\mathcal{A}_{lab}=(\Sigma\times(Q\uplus\{*\}),2^Q,\delta_{lab},F_{lab})$ if any node receives state
    $\emptyset$, so does its parent (since no set $h(\delta(s,q))$ can contain $\emptyset$); by transitivity, if any node receives state
    $\emptyset$, so does the root -- we can assume no node receives state $\emptyset$, and
    in particular every node that has an assigned label $q$ is of state $\{q\}$.

    Arbitrarily choose an element from the state assigned to the root. We show that for
    every node that has chosen a state of $\mathcal{A}$ from its assigned state of
    $\mathcal{A}_{lab}$, there is a way to choose an element for each of its children from
    their assigned state of $\mathcal{A}_{lab}$ such that the resulting transition is
    valid for $\mathcal{A}$.

    Let $v$ be a node with $\Sigma$-label $s$, let $H$ be its assigned state of $\mathcal{A}_{lab}$ and $q\in H$
    be its assigned state of $\mathcal{A}$. We distinguish two cases:
    \begin{itemize}
        \item $v$ was labeled in $G$ with a state $q'$: then $H=\{q'\}$ and since $q\in
        H$, $q=q'$. Then the multiset $M$ of states of $\mathcal{A}_{lab}$ of its children is in
        $\delta_{lab}((s,q),\{q\})$; observe that it is the union of $\delta_{det}(s,Q_i)$
        where the $Q_i$s are all states containing $q$, so in particular
        $\delta_{det}(s,Q_i)\subseteq h(\delta(s,q))$ and then
        $M\in\delta_{lab}((s,q),\{q\})\subseteq h(\delta(s,q))$.
        \item $v$ is unlabeled in $G$: then the multiset $M$ of types (according to $\mathcal{A}_{lab}$) of its children is
        in $\delta_{lab}((s,*),H)=\delta_{det}(s,H)\subseteq h(\delta(s,q))$.
    \end{itemize}
    In either case, $M\in h(\delta(s,q))$; by definition there is a choice function $c$
    from $M$ to an element $M'\in \delta(s,q)$ such that for every state $Q_0\in 2^Q$
    in $M$ we have $c(Q_0)\in Q_0$: we assign states of $\mathcal{A}$ to the children of
    $v$ according to $c$.

    We can then use top-down tree induction -- the base case is the root, where we can
    choose arbitrarily -- to prove that every node $v$ can choose the labels of its children
    such that the transition function of $\mathcal{A}$ at $v$ is respected; then 
    the result is a valid run of $\mathcal{A}$. Additionally, since every node labeled $q$
    is of type $\{q\}$, the run we obtain is an extension of the $Q$-labeling of $G$.
\end{proof}
\begin{lemma}\label{lem:typeauto}
    Let $\Pi_{\mathcal{A}}$ be the problem associated to the rational multiset tree
    automaton $\mathcal{A}=(\Sigma,Q,\delta,F)$. Let
    $\mathcal{A}_{lab}=(\Sigma\times(Q\uplus\{*\}),2^Q,\delta_{lab},F_{lab})$ be the partially labeled
    type automaton of $\mathcal{A}$. Then for any $\Sigma$-labeled and partially $Q$-labeled rooted tree $G$ with
    root $r$ the following statements are equivalent:
    \begin{enumerate}[(1)]
        \item the type assigned by $\mathcal{A}_{lab}$ to $r$ contains an accepting state
        of $\mathcal{A}$,
        \item $G$ is accepted by $\mathcal{A}_{lab}$,
        \item the partial $Q$-labeling of $G$ can be extended to an accepting run of
        $\mathcal{A}$ on the underlying unlabeled graph.
    \end{enumerate}
\end{lemma}
\begin{proof}
    This statement and proof is just a formalization of the statement ``$\mathcal{A}_{lab}$ assigns to each
    node a state corresponding to the set of states of $\mathcal{A}$ that can appear in a
    run of $\mathcal{A}$ which extends the partial labeling''.

    $(1)\Leftrightarrow(2)$ since $F_{lab}$ is defined as any set of states $2^Q$ that
    contains at least one state of $F$.

    $(1)\Rightarrow(3)$: by the previous lemma, we can select an accepting state $q$ of $\mathcal{A}$ to the
    root that is contained in the type assigned by $\mathcal{A}_{lab}$ and then compute a
    valid run of $\mathcal{A}$ that assigns $q$ to the root and extends the $Q$-labeling
    of $G$.

    $(3)\Rightarrow(1)$: let $r$ be the accepting run of $\mathcal{A}$ that extends the
    $Q$-labeling of $G$, and let $r'$ be the run of $\mathcal{A}_{lab}$ on $G$; we show
    that for every vertex $v$ we have $r(v)\in r'(v)$. We work by bottom-up tree
    induction: for the base case, let $l$ be a leaf, and observe that
    $\emptyset\in\delta_{lab}((s,*),r'(l))=\delta_{det}(s,r'(l))$ ($r'$ is valid in $l$) implies
    $\emptyset\in\delta(s,r(l))\Leftrightarrow r(l)\in r'(l)$; then since $r(l)$ is a valid label of
    $l$ according to $\mathcal{A}$, $r(l)\in r'(l)$. If $l$ was labeled by $r(l)$ in $G$,
    this still holds; let $H$ be the state that $\mathcal{A}_{lab}$ would have assigned it
    if it was unlabeled, then $r(l)\in H$ by the same reasoning and
    $\emptyset\in\delta_{lab}((s,*),H)\subseteq\delta_{lab}(r(l),\{r(l)\})$. Then
    $\mathcal{A}_{lab}$ assigns $\{r(l)\}$ as a state to $l$, which obviously contains
    $r(l)$.
    
    For the inductive step, let $M$ be the set of states of children of $v$ according to
    $r$ and let $M'$ be the set of states of children of $v$ according to $r'$; by
    induction we have a bijective choice function $c:M'\to M$ which sends a set of states to one of
    its elements. Then $M'\in h(M)$ and since $M\in\delta(s,r(v))$ we have $M'\in
    h(\delta(s,r(v)))$; assume $v$ is unlabeled, then $M'\in\delta_{lab}((s,*),r'(v))$; since
    $\delta_{lab}((s,*),r'(v))=\delta_{det}(s,r'(v))$ has nonempty intersection with
    $h(\delta(s,r(v)))$ we have $r(v)\in r'(v)$. Assume $v$ is labeled with $r(v)$, and
    let $H$ be the state that $\mathcal{A}_{lab}$ would have assigned it
    if it was unlabeled, then $r(v)\in H$ and $M'\in
    \delta_{lab}((s,*),H)\subseteq\delta_{lab}((s,r(v)),\{r(v)\})$; then $\mathcal{A}_{lab}$ assigns
    $\{r(v)\}$ as a state to $v$, which obviously contains $r(v)$.
\end{proof}
\begin{lemma}\label{lem:correct}
    Let $G$ be a $\Sigma$-labeled (not labeled with $Q$) rooted tree with root $r$. Then the
    type that the above algorithm assigns to the root $r$ contains at least
    one accepting state.
\end{lemma}
\begin{proof}
    Recall that our assumption is that there is a feasible function (equivalently, that
    $\Pi_{\mathcal{A},\Delta}$ is solvable) -- we first need to show that $G$
    is solvable by $\Pi_{\mathcal{A}}$, or equivalently (by Lemma~\ref{lem:typeauto}) that
    $G$ is accepted by $\mathcal{A}_{lab}$.
    
    Assume we had computed the types of all nodes of
    $G$ according to $\mathcal{A}_{lab}$ without ever prelabeling nodes through a feasible
    function; then applied the procedure $\mathrm{ForgetChildren}_R$ to every node and
    completely removed the ``forgotten'' edges. This
    procedure leaves the types of all nodes unchanged by definition, in particular it will
    leave the type of the root unchanged -- it also disconnects the graph, so we only
    consider the connected component $G'$ containing the root, which is a rooted tree of
    maximum degree $\leq\Delta$. Because $\Pi_{\mathcal{A},\Delta}$ is solvable, $G'$ is
    solvable as an instance of $\Pi_{\mathcal{A},\Delta}$; which means $G$ is accepted by
    $\mathcal{A}_{lab}$, as the two problems are described by the same automaton
    $\mathcal{A}$. Then the type of the root of $G$ and $G'$ (without any prelabeling) contains an
    accepting state.

    Let us now look at the algorithm; we show that every step of the algorithm preserves
    at least one of the equivalent definitions in Lemma~\ref{lem:typeauto}. We only need
    to focus on the steps that change the topology of the graph: the virtual removal of
    edges in the $\mathrm{ForgetChildren}$ procedures and the partial labeling by the
    feasible function. As we have already observed, either $\mathrm{ForgetChildren}$
    preserves the types of all nodes; then it preserves (1). Additionally, by definition
    the feasible function always extends the labeling $l$ of $G$ to $l'$ such that if
    $l$ could be extended to a valid labeling then so can $l'$; then it preserves (3). 
\end{proof}
\begin{lemma}\label{lem:height}
    Every node of $G$ is the $\gamma\cdot (L+1)+\ell$-hop descendant of either the root or a node
    that was preassigned a label.
\end{lemma}
\begin{proof}
    Assume this node $v$ is in layer $R_i$. Then it is at most distance $\gamma$ from the root
    of its layer. The root is the child of a node in a higher layer; if this layer is
    $C_j$ for some $j\geq i$, then the middle node of it has been prelabeled and $v$ is
    the $\gamma+\ell$-hop descendant of it. If this layer is $R_j$ for some $j>i$, then
    $v$ is the $2\gamma$-hop descendant of its root; we can at most repeat this once for every layer
    $R_1,\ldots,R_{L+1}$ until we reach the root of $G$.
\end{proof}
\begin{lemma}\label{lem:complex}
    Applying the above algorithm using a $(\gamma,\ell,L)$-decomposition takes
    $O((\gamma+\ell)L)$ rounds to terminate.
\end{lemma}
\begin{proof}
    For the $\mathcal{A}_{lab}$ type labeling step: each $i\in[1,\ldots,L]$, $R_i$ is composed of rooted trees of depth at most
    $\gamma+\ell$ (as at most $\ell$ nodes can be promoted from $C_{i-1}$ to $R_i$ during
    the algorithm) so computing all types by brute-force takes $O(\gamma+\ell)$ rounds. 
    Similarly, $C_i$ is composed of paths of length at most $2\ell$, and labeling them takes $O(\ell)$ time; 
    at the end, labeling $R_{L+1}$ takes a further $O(\gamma+\ell)$ rounds. In total, this
    takes $O((\gamma+\ell)L)$ rounds.

    For the $\mathcal{A}$ state labeling step: at each round, each node that has already
    decided on a label chooses labels for its children and communicates them. Since by
    Lemma~\ref{lem:height} every node is the $\gamma\cdot (L+1)+\ell$-hop descendant of a
    node that has chosen its label during the first part, this takes
    $O((\gamma)L+\ell)\subseteq O((\gamma+\ell)L)$ rounds.
\end{proof}

Then the overall algorithm takes $O((\gamma+\ell)L)$ rounds. By applying it with a
$\infty$-feasible function and a $(1,\ell,O(\log n))$-decomposition (which we can compute
in $O(\log n)$ rounds by Lemma~\ref{lem:logdec}), we prove
Theorem~\ref{thm:logred}; by applying it with a $k$-feasible function and a
$(O(n^{1/k}),\ell,k)$-decomposition (which we can compute
in $O(n^{1/k})$ rounds by Lemma~\ref{lem:polydec}), we prove Theorem~\ref{thm:polyred}. A
simple corollary to Theorem~\ref{thm:polyred} covers the case where both problems are
unsolvable.
\begin{corollary}
    Let $\Pi_{\mathcal{A}}$ be an automaton certification problem, and let $\Pi_{\mathcal{A},\Delta}$ be its $\mathsf{LCL}$
    restriction to the reduction degree. Then $\Pi_{\mathcal{A},\Delta}$ is unsolvable if
    and only if $\Pi_{\mathcal{A}}$ is unsolvable.
\end{corollary}
\begin{proof}
    Because $\Pi_{\mathcal{A},\Delta}$ is a restriction of $\Pi_{\mathcal{A}}$, if
    $\Pi_{\mathcal{A},\Delta}$ is unsolvable so is $\Pi_{\mathcal{A}}$.

    For the converse, assume $\Pi_{\mathcal{A},\Delta}$ is solvable: then there is a
    $O(n)$ algorithm to solve it; by Theorem~\ref{thm:polyred}, there is a $O(n)$
    algorithm to solve $\Pi_{\mathcal{A}}$, which is not unsolvable.
\end{proof}
Together with the
classical classification results for $\mathsf{LCL}$ problems from
\cite{changHierarchy2019,treesChang2020} this proves Theorem~\ref{thm:autogap}.
\subsection{Interlude: labels and certificates}\label{sec:labelauto}
Throughout the proof of Theorem~\ref{thm:autogap}, we worked under the assumption that our
graph was $\Sigma$-labeled and computed a certificate (a $Q$-labeling) of the correctness
of the $\Sigma$-labeling. When talking about a Local $\mathsf{PMSO}$ labeling problem
$\Pi=(\Sigma_{in},\Sigma_{out},\mathcal{G})$ we could use the algorithms from
Theorem~\ref{thm:autogap} to \emph{certify} a solution ($\Sigma_{out}$-labeling) we have already computed: we instead
want to use it to compute \emph{both} a solution and a certificate of its correctness.

We can do this with a little trick: instead of using the automaton which verifies the
correctness of a solution directly -- by checking whether a
$\Sigma_{in}\times\Sigma_{out}$-labeled rooted tree has the property $\mathcal{G}$ -- we
modify it into an automaton that verifies if a $\Sigma_{in}$-labeled rooted tree \emph{can} be $\Sigma_{out}$-labeled
according to $\mathcal{G}$; additionally, this automaton encodes valid $\Sigma_{out}$
labeling in its accepting runs.
\labelautolem*
\begin{proof}
    The only thing that needs to be defined in $\mathcal{A}'$ is $\delta'$; 
    given a multiset $M$ of states we define as $X_\Sigma(M)$
    the set of multisets of $\Sigma\times Q$ such that projecting all their elements to
    the second component gives $M$. Essentially, we allow every possible combination of
    labels as long as the states respect $M$, since the transition function does not
    depend on the labels of the children nodes; then for every state $q\in Q$ and pair $(s_i,s_o)\in\Sigma_{in}\times\Sigma_{out}$ we define
    $$\delta'(s_i,(s_o,q)):= X_\Sigma(\delta((s_i,s_o),q)).$$
    Its correctness is trivial, as the transition function depends only on the state, but
    we need to prove that it is a rational multiset tree automata;
    the easiest way to see that $\delta'(s_i,(s_o,q))$ is a rational set is using the
    finite multiset automaton characterization. Given a multiset automaton $\mathcal{A}_M:=(Q,Q_A,q_0,f,F_A)$ for $\delta((s_i,s_o),q)$
    (note that its alphabet is the states of $\mathcal{A}$) we can
    replace the alphabet with $\Sigma_{out}\times Q$ and the transition function by
    $$f_{\Sigma_{out}}:(\Sigma_{out}\times Q)\times Q_A\to 2^{Q_A}\quad f_{\Sigma_{out}}((s_o,q),q_A):= f(q,q_A)
    $$
    ignoring $s\in\Sigma$ completely; visually, we replaced every transition arrow
    associated to a symbol $q$ with multiple arrows corresponding to all elements of $\Sigma\times\{q\}$.
\end{proof}
Note that even if we start from a deterministic, complete automaton $\mathcal{A}$ (as we do in the next section), the 
resulting automaton $\mathcal{A}'$ might be nondeterministic; this does not matter for applying Theorem~\ref{sec:auto_gap}.

Conceptually, this step is equivalent to taking a $\mathsf{PMSO}$ sentence $\varphi$ on $\Sigma_{in}\times\Sigma_{out}$-labeled graphs with out-label sets 
$\{S_{s_o}\}_{s_o\in\Sigma_{out}}$ to a $\mathsf{PMSO}$ sentence of the form
$$\exists_{s_o\in \Sigma_{out}} S_{s_o}\, \mathsf{Part}(\{S_{s_o}\}_{s_o\in\Sigma_{out}})\land \varphi $$
where $\mathsf{Part}(\{S_s\}_{s\in\Sigma})$ is the first order sentence expressing ``$\{S_s\}_{s\in\Sigma}$ is a
partition'' (see Section~\ref{sec:axioms} for details) and $\varphi$ uses the second-order variables
$S_{s_o}$ in place of the unary relations representing the $\Sigma_{out}$-labeling -- so it is defined on $\Sigma_{in}$-labeled graphs.

In other words, we accept every $\Sigma_{in}$-labeled tree for which a valid $\Sigma_{out}$-labeling \emph{exists} instead of every graph
which is already correctly labeled. This is only a valid formula if $\Sigma_{out}$ is a finite
set, for the same reason as to why $\mathcal{A}'$ is only a valid finite state automaton
if $\Sigma_{out}$ is a finite set.
\subsection{Bridging the verification gap}\label{sec:auto_redu}
To conclude, we show that \emph{any $\mathsf{LPMSO}$ problem is asymptotically as hard as
an automaton certification problem}. Theorem~\ref{thm:PMSO_auto} guarantees that for any
$\mathsf{PMSO}$ property $\mathcal{G}$ describing a $\mathsf{LPMSO}$ problem $\Pi$, there is \emph{some}
rational multiset tree automaton certifying $\mathcal{G}$; however, not \emph{every} such
RatMTA will allow us to prove the desired result. We give an example: let $\Pi=(\{*\},\{T\},\mathcal{G}_T)$ be the
trivial problem of labeling an unlabeled rooted tree with $T$ everywhere. Clearly,
$\mathcal{G}_T$ is locally verifiable in $0$ rounds; additionally it can be described by
the simple $\mathsf{PMSO}$ formula $\forall x\, x\in S_T$ (``every node is labeled $T$''). We define a
rational multiset tree automaton that recognises the language $\mathcal{G}_T$ of $T$-labeled
rooted trees: $\mathcal{A}_T=(\{T\},\{0,1\},\delta_T,\{0,1\})$ where
$\delta(T,0)=\mathcal{M}(\{1\})$ (all multisets composed only of $1$s) and
$\delta(T,1)=\mathcal{M}(\{0\})$ (all multisets composed only of $0$s). A valid run of
$\mathcal{A}$ is a $2$-coloring of the rooted tree with $\{0,1\}$; since all states are
accepting, every valid run is an accepting run. Solving the automaton certification
problem associated with $\mathcal{A}_T$ is equivalent to $2$-coloring a rooted tree, which
is known to take $\Omega(n)$ rounds \cite{balliuRooted2022}.

We need to choose our automata carefully so that we are not
forced to do extra work: for this purpose we choose the deterministic and complete RatMTA which
is \emph{minimal} with respect to the number of states, which is given to us by
Theorem~\ref{thm:minimal-ratmta} (restated below).
\minratmtathm*
Recall that $\equiv_{\mathcal{G}}$ is the coarsest possible ``replacement'' relation: we
can replace a rooted subtree with an equivalent one without impacting the overall
correctness of the labeling. In particular, the equivalence relationship
between multipolar trees based on isomorphisms of the $k$-hop neighborhood of the poles (where $k$ is the
local checkability radius of our problem) which is defined in \cite{changHierarchy2019}
has the same property, and so is a refinement of $\equiv_{\mathcal{G}}$; we will use this
fact to compute a certificate for a valid solution in constant rounds.
\gapjumpthm*
We say that a $(\Sigma_{in}\times\Sigma_{out})$-labeled rooted tree is \emph{$k$-away from happy} with respect to a
locally verifiable (in $k$ rounds) problem $\Pi=(\Sigma_{in},\Sigma_{out},\mathcal{G})$ if, by running the $k$-round
verification algorithm for $\Pi$, every node at distance $\geq k+1$ from the root (which
is too far away to communicate with the root) returns ``yes''.
\begin{definition}
    Let $\Pi=(\Sigma_{in},\Sigma_{out},\mathcal{G})$ be a problem whose solution can be locally
    checked in $k$ rounds. We define an equivalence relation on $(\Sigma_{in}\times\Sigma_{out})$-labeled rooted
    trees $\equiv_{\Pi}$ as follows: $t_1\equiv_{\Pi}t_2$ if and only if 
    \begin{enumerate}
    \item $N^{2k+2}(r(t_1))\cong N^{2k+2}(r(t_2))$ (the $2k+2$-radius neighborhoods of their
    roots are isomorphic), and
    \item $t_1$ is $k$-away from happy iff $t_2$ is $k$-away from happy (with respect to $\Pi$).
\end{enumerate}
\end{definition}
\begin{lemma}\label{lem:refine}
    For any $\Pi=(\Sigma_{in},\Sigma_{out},\mathcal{G})$ we have $t_1\equiv_{\Pi} t_2\Rightarrow
    t_1\equiv_{\mathcal{G}} t_2$. 
\end{lemma}
\begin{proof}
Let $t_1\equiv_{\Pi} t_2$ and let $C$ be any context; we want to show that
$C[t_1]\in\mathcal{G}\Leftrightarrow C[t_2]\in\mathcal{G}$, since the order is arbitrary
it suffices to show $C[t_1]\notin\mathcal{G}\Rightarrow C[t_2]\notin\mathcal{G}$.

Assume $C[t_1]\notin\mathcal{G}$, and run the $r$-rounds verification algorithm on $C[t_1]$ and
$C[t_2]$. The algorithm returns ``no'' on at least one node $v$ of $C[t_1]$; this can happen
in three cases:
\begin{enumerate}
    \item $v$ is not a member of $t_1$: all nodes of $C[t_2]\smallsetminus t_2$ have the
    same view to the corresponding nodes of $C[t_1]\smallsetminus t_1$, so the
    corresponding node in $C[t_2]$ will also return ``no'' and $C[t_2]\notin\mathcal{G}$;
    \item $v$ is a member of $t_1$ which is $\leq k$ steps away from the root: since the
    $2k+2$-hop neighborhoods of the roots are isomorphic, the corresponding nodes have
    isomorphic $k$-hop views, so the corresponding node in $C[t_2]$ will also return ``no'' and $C[t_2]\notin\mathcal{G}$;
    \item $v$ is a member of $t_1$ which is $>k$ steps away from the root: then $t_1$ is
    not $k$-away from happy, so neither is $t_2$, then there is a node in $t_2$ and
    $C[t_2]$ which will return ``no'' and $C[t_2]\notin\mathcal{G}$.
\end{enumerate}
\end{proof}
In essence, we have proven that $\equiv_{\Pi}$ (which has infinitely many equivalence
classes) is a refinement of $\equiv_{\mathcal{G}}$, which for $\mathsf{PMSO}$-definable
languages has finitely many equivalence classes (which are up to renaming the states of the minimal
automaton). In other words, knowing the $2k+2$-hop neighborhood of a node and whether it is
$k$-away from happy is sufficient to determine its exact state in the minimal automaton;
we show this correspondence is computable.
\begin{lemma}
    Let $\Pi=(\Sigma_{in},\Sigma_{out},\mathcal{G})$ be a Local $\mathsf{PMSO}$ problem on rooted trees whose solution can be locally
    checked in $k$ rounds. Then there exists a \emph{computable} function
    $$f:\mathcal{T}(\Sigma_{in}\times\Sigma_{out})/\equiv_{\Pi}\to\mathcal{T}(\Sigma_{in}\times\Sigma_{out})/\equiv_{\mathcal{G}}$$
    which for every $(\Sigma_{in}\times\Sigma_{out})$-labeled rooted tree $t$ maps $[t]_{\Pi}$ to the equivalence class $[t]_{\mathcal{G}}$.
\end{lemma}
\begin{proof}
    $f$ exists and is well-defined by Lemma~\ref{lem:refine}; we just need to show that it is computable.

    We represent an equivalence class $T_{\Pi}\in\mathcal{T}(\Sigma)/\equiv_{\Pi}$ as a
    pair $(N^{2k+2}(r),b)$ where $N^{2k+2}(r)$ is a $2k+2$-height rooted tree (considered
    as a $2k+2$-hop neighborhood centered at the root) and $b$ is a bit representing
    whether the trees of this class are $k$-away from happy or not; note that this
    representation might lead to empty classes, but the function is defined on nonempty
    classes (alternatively, it can be defined as sending empty class to empty class).

    We do a preprocessing step on $\Pi$: for every equivalence class $q$ of $\mathcal{G}$
    (corresponding to the states of the minimal automaton) and every label pair $(s_i,s_o)$ in $\Sigma_{in}\times\Sigma_{out}$ we want to compute a finite
    representative $t$ whose root is labeled $(s_i,s_o)$ and which is of class $q$. We show that
    every pair has a finite representative of maximum
    degree $\Delta$ and maximum height $\lvert\Sigma_{in}\rvert\cdot\lvert\Sigma_{out}\rvert\cdot h$, where $h$ is the number of equivalence classes
    of $\mathcal{G}$. Let $t$ be any element of class $T_{\mathcal{G}}$; we can compute
    the type of all its nodes and apply the procedure $\mathsf{ForgetChildren}$ to reduce
    its maximum degree to $\Delta$ without affecting the type of the root. 
    
    Now let us only
    consider trees of maximum degree $\Delta$ and initialize a set $\mathcal{Q}$
    containing all possible $\Sigma_{in}\times\Sigma_{out}$-labelings of a leaf. We define a recursive
    function $g$ which constructs all trees of maximum degree $\Delta$ of height $1$ where
    the root is $\Sigma_{in}\times\Sigma_{out}$-labeled and the leaves are labeled with trees of $\mathcal{Q}$ --
    then, for each leaf node labeled by a tree $t_1$ we replace the leaf with the tree
    $t_1$. We iterate over the finitely many trees that are created this way, and
    compute a run of the minimal automaton on it; if the tuple of the root state $q$ and
    root labels $(s_i,s_o)$ of $t$ does not appear
    in any element of $\mathcal{Q}$, we add the tree $t$ to $\mathcal{Q}'$ -- once we have
    iterated over all trees constructed this way, we set
    $g(\mathcal{Q}):=\mathcal{Q}\cup\mathcal{Q}'$. We remark on the properties of $g$:
    \begin{itemize}
        \item $g$ is deterministic (as long as we process the trees in a deterministic order),
        \item $\mathcal{Q}$ starts with one element and is weakly increasing
        ($\mathcal{Q}\subseteq g(\mathcal{Q})$),
        \item $\mathcal{Q}$ does not contain two trees with the same state and root label pair: in other
        words,
        $$\lvert\mathcal{Q}\rvert\leq\lvert\Sigma_{in}\times\Sigma_{out}\rvert\cdot\lvert\mathcal{T}(\Sigma)/\equiv_{\mathcal{G}}\rvert=\lvert\Sigma_{in}\rvert\cdot\lvert\Sigma_{out}\rvert\cdot
        h,$$
        \item if $h_0$ is the maximum height of a tree in $\mathcal{Q}$, then $h_0+1$ is
        the maximum height of a tree in $g(\mathcal{Q})$.
    \end{itemize}
    Since $g$ is a deterministic monotone (in the set inclusion sense) function with an
    upper bound $\lvert\Sigma_{in}\rvert\cdot\lvert\Sigma_{out}\rvert\cdot h$ on size, it reaches a fixed point after at
    most $\lvert\Sigma_{in}\rvert\cdot\lvert\Sigma_{out}\rvert\cdot h$ iterations; at
    that point, we have found a representative of maximum degree $\Delta$ and maximum height $h$
    for every state and label pair which \emph{can} have a representative.

    Finally, we describe the process for computing $f$: given a pair $(N^{2k+2}(r),b)$, we
    consider every possible assignment of states to the leaves of $N^{2k+2}(r)$, which are
    $\Sigma_{in}\times\Sigma_{out}$-labeled. Then for every state and label pair assigned
    to a leaf, we replace the leaf with the finite
    representative that we precomputed; then we can check if this tree is $k$-away from
    happy. If the pair we are given represents a nonempty class, we show that we will find at least one
    tree of this form that represents it: simply take any tree of this class, compute a run of
    the minimal automaton, and then replace the subtrees rooted at nodes $2k+2$-away from
    the root with the corresponding representatives.

    We have finitely many assignments of states to the leaves of $N^{2k+2}(r)$; we can
    iterate through all of them, and if we find a valid representative (a tree which is
    $k$-away from happy if $b=1$, one who is not if $b=0$) we compute its equivalence
    class; we set this class as the result of $f((N^{2k+2}(r),b))$ -- if we cannot find
    any tree of this class, the class is empty and we can set $f(\emptyset)=\emptyset$.
\end{proof}
We can now prove Theorem~\ref{thm:gap_jump}:
\begin{proof}
    We describe an algorithm that starts from a correct solution of $\Pi$ and computes a
    correct solution of $\Pi'$ in $2k+2$ rounds. Since the solution is correct, every
    rooted subtree is in particular $k$-away from happy; then every node can gather its $2k+2$-hop
    descendants and determine its own state by computing $f(N^{2k+2}(r),1)$, then output the pair of its original label
    and state.
\end{proof}
\printbibliography

@inproceedings{fraigniaudIdentifiers2013,
  author       = {Pierre Fraigniaud and
                  Mika G{\"{o}}{\"{o}}s and
                  Amos Korman and
                  Jukka Suomela},
  editor       = {Panagiota Fatourou and
                  Gadi Taubenfeld},
  title        = {What can be decided locally without identifiers?},
  booktitle    = {{ACM} Symposium on Principles of Distributed Computing, {PODC} '13,
                  Montreal, QC, Canada, July 22-24, 2013},
  pages        = {157--165},
  publisher    = {{ACM}},
  year         = {2013},
  url          = {https://doi.org/10.1145/2484239.2484264},
  doi          = {10.1145/2484239.2484264},
  bibsource    = {dblp computer science bibliography, https://dblp.org}
}

@INPROCEEDINGS{linialDGA1987,
  author       = {Nathan Linial},
  title        = {Distributive Graph Algorithms-Global Solutions from Local Data},
  booktitle    = {28th Annual Symposium on Foundations of Computer Science, Los Angeles,
                  California, USA, 27-29 October 1987},
  pages        = {331--335},
  publisher    = {{IEEE} Computer Society},
  year         = {1987},
  url          = {https://doi.org/10.1109/SFCS.1987.20},
  doi          = {10.1109/SFCS.1987.20},
  bibsource    = {dblp computer science bibliography, https://dblp.org}
}

@article{linialLocality1992,
  author       = {Nathan Linial},
  title        = {Locality in Distributed Graph Algorithms},
  journal      = {{SIAM} J. Comput.},
  volume       = {21},
  number       = {1},
  pages        = {193--201},
  year         = {1992},
  url          = {https://doi.org/10.1137/0221015},
  doi          = {10.1137/0221015},
  bibsource    = {dblp computer science bibliography, https://dblp.org}
}

@inproceedings{balliuPaths2019,
  author       = {Alkida Balliu and
                  Sebastian Brandt and
                  Yi{-}Jun Chang and
                  Dennis Olivetti and
                  Mika{\"{e}}l Rabie and
                  Jukka Suomela},
  editor       = {Peter Robinson and
                  Faith Ellen},
  title        = {The Distributed Complexity of Locally Checkable Problems on Paths
                  is Decidable},
  booktitle    = {Proceedings of the 2019 {ACM} Symposium on Principles of Distributed
                  Computing, {PODC} 2019, Toronto, ON, Canada, July 29 - August 2, 2019},
  pages        = {262--271},
  publisher    = {{ACM}},
  year         = {2019},
  url          = {https://doi.org/10.1145/3293611.3331606},
  doi          = {10.1145/3293611.3331606},
  bibsource    = {dblp computer science bibliography, https://dblp.org}
}

@inproceedings{brandtGrids2017,
  author       = {Sebastian Brandt and
                  Juho Hirvonen and
                  Janne H. Korhonen and
                  Tuomo Lempi{\"{a}}inen and
                  Patric R. J. {\"{O}}sterg{\aa}rd and
                  Christopher Purcell and
                  Joel Rybicki and
                  Jukka Suomela and
                  Przemyslaw Uznanski},
  editor       = {Elad Michael Schiller and
                  Alexander A. Schwarzmann},
  title        = {{LCL} Problems on Grids},
  booktitle    = {Proceedings of the {ACM} Symposium on Principles of Distributed Computing,
                  {PODC} 2017, Washington, DC, USA, July 25-27, 2017},
  pages        = {101--110},
  publisher    = {{ACM}},
  year         = {2017},
  url          = {https://doi.org/10.1145/3087801.3087833},
  doi          = {10.1145/3087801.3087833},
  bibsource    = {dblp computer science bibliography, https://dblp.org}
}

@misc{schmidLFL2026,
      title={LCLs Beyond Bounded Degrees}, 
      author={Gustav Schmid},
      year={2026},
      eprint={2602.02340},
      archivePrefix={arXiv},
      primaryClass={cs.DC},
      url={https://arxiv.org/abs/2602.02340}, 
}

@inproceedings{feuilloleyMSOcert2022,
  author       = {Laurent Feuilloley and
                  Nicolas Bousquet and
                  Th{\'{e}}o Pierron},
  editor       = {Alessia Milani and
                  Philipp Woelfel},
  title        = {What Can Be Certified Compactly? Compact local certification of {MSO}
                  properties in tree-like graphs},
  booktitle    = {{PODC} '22: {ACM} Symposium on Principles of Distributed Computing,
                  Salerno, Italy, July 25 - 29, 2022},
  pages        = {131--140},
  publisher    = {{ACM}},
  year         = {2022},
  url          = {https://doi.org/10.1145/3519270.3538416},
  doi          = {10.1145/3519270.3538416},
  bibsource    = {dblp computer science bibliography, https://dblp.org}
}

@inproceedings{balliuRooted2022,
  author       = {Alkida Balliu and
                  Sebastian Brandt and
                  Dennis Olivetti and
                  Jan Studen{\'{y}} and
                  Jukka Suomela and
                  Aleksandr Tereshchenko},
  editor       = {Avery Miller and
                  Keren Censor{-}Hillel and
                  Janne H. Korhonen},
  title        = {Locally Checkable Problems in Rooted Trees},
  booktitle    = {{PODC} '21: {ACM} Symposium on Principles of Distributed Computing,
                  Virtual Event, Italy, July 26-30, 2021},
  pages        = {263--272},
  publisher    = {{ACM}},
  year         = {2021},
  url          = {https://doi.org/10.1145/3465084.3467934},
  doi          = {10.1145/3465084.3467934},
  bibsource    = {dblp computer science bibliography, https://dblp.org}
}

@inproceedings{changHierarchy2019,
  author       = {Yi{-}Jun Chang and
                  Seth Pettie},
  editor       = {Chris Umans},
  title        = {A Time Hierarchy Theorem for the {LOCAL} Model},
  booktitle    = {58th {IEEE} Annual Symposium on Foundations of Computer Science, {FOCS}
                  2017, Berkeley, CA, USA, October 15-17, 2017},
  pages        = {156--167},
  publisher    = {{IEEE} Computer Society},
  year         = {2017},
  url          = {https://doi.org/10.1109/FOCS.2017.23},
  doi          = {10.1109/FOCS.2017.23},
  bibsource    = {dblp computer science bibliography, https://dblp.org}
}

@book{treeautomataTA,
  TITLE = {{Tree Automata Techniques and Applications}},
  AUTHOR = {Comon, Hubert and Dauchet, Max and Gilleron, R{\'e}mi and Jacquemard, Florent and Lugiez, Denis and L{\"o}ding, Christof and Tison, Sophie and Tommasi, Marc},
  URL = {https://inria.hal.science/hal-03367725},
  PAGES = {262},
  YEAR = {2008},
  PUBLISHER = {HAL},
  HAL_ID = {hal-03367725},
  HAL_VERSION = {v1},
}

@inproceedings{csuhajMultiset2000,
  author       = {Erzs{\'{e}}bet Csuhaj{-}Varj{\'{u}} and
                  Carlos Mart{\'{\i}}n{-}Vide and
                  Victor Mitrana},
  editor       = {Cristian Calude and
                  Gheorghe Paun and
                  Grzegorz Rozenberg and
                  Arto Salomaa},
  title        = {Multiset Automata},
  booktitle    = {Multiset Processing, Mathematical, Computer Science, and Molecular
                  Computing Points of View [Workshop on Multiset Processing, {WMP} 2000,
                  Curtea de Arges, Romania, August 21-25, 2000]},
  series       = {Lecture Notes in Computer Science},
  pages        = {69--84},
  publisher    = {Springer},
  year         = {2000},
  url          = {https://doi.org/10.1007/3-540-45523-X\_4},
  doi          = {10.1007/3-540-45523-X\_4},
  bibsource    = {dblp computer science bibliography, https://dblp.org}
}

@inproceedings{bonevaAutomata2005,
  author       = {Iovka Boneva and
                  Jean{-}Marc Talbot},
  editor       = {J{\"{u}}rgen Giesl},
  title        = {Automata and Logics for Unranked and Unordered Trees},
  booktitle    = {Term Rewriting and Applications, 16th International Conference, {RTA}
                  2005, Nara, Japan, April 19-21, 2005, Proceedings},
  series       = {Lecture Notes in Computer Science},
  pages        = {500--515},
  publisher    = {Springer},
  year         = {2005},
  url          = {https://doi.org/10.1007/978-3-540-32033-3\_36},
  doi          = {10.1007/978-3-540-32033-3\_36},
  bibsource    = {dblp computer science bibliography, https://dblp.org}
}

@article{eilenbergRational1969,
  title = {Rational sets in commutative monoids},
  volume = {13},
  ISSN = {0021-8693},
  url = {http://dx.doi.org/10.1016/0021-8693(69)90070-2},
  DOI = {10.1016/0021-8693(69)90070-2},
  number = {2},
  journal = {Journal of Algebra},
  publisher = {Elsevier BV},
  author = {Eilenberg,  Samuel and Sch\"{u}tzenberger,  M.P},
  year = {1969},
  month = oct,
  pages = {173–191}
}

@inproceedings{colcombetRational2003,
  author       = {Thomas Colcombet},
  editor       = {Anton{\'{\i}}n Kucera and
                  Richard Mayr},
  title        = {Rewriting in the partial algebra of typed terms modulo {AC}},
  booktitle    = {4th International Workshop on Verification of Infinite-State Systems
                  {(CONCUR} 2002 Satellite Workshop), Infinity 2002, Brno, Czech Republic,
                  August 24, 2002},
  series       = {Electronic Notes in Theoretical Computer Science},
  number       = {6},
  pages        = {40--54},
  publisher    = {Elsevier},
  year         = {2002},
  url          = {https://doi.org/10.1016/S1571-0661(04)80532-2},
  doi          = {10.1016/S1571-0661(04)80532-2},
  bibsource    = {dblp computer science bibliography, https://dblp.org}
}

@article{ginsburgPresburger1966,
  title = {Semigroups,  Presburger formulas,  and languages},
  volume = {16},
  ISSN = {0030-8730},
  url = {http://dx.doi.org/10.2140/pjm.1966.16.285},
  DOI = {10.2140/pjm.1966.16.285},
  number = {2},
  journal = {Pacific Journal of Mathematics},
  publisher = {Mathematical Sciences Publishers},
  author = {Ginsburg,  Seymour and Spanier,  Edwin},
  year = {1966},
  month = feb,
  pages = {285–296}
}

@inproceedings{treesChang2020,
  author       = {Yi{-}Jun Chang},
  editor       = {Hagit Attiya},
  title        = {The Complexity Landscape of Distributed Locally Checkable Problems
                  on Trees},
  booktitle    = {34th International Symposium on Distributed Computing, {DISC} 2020,
                  Virtual Conference, October 12-16, 2020},
  series       = {LIPIcs},
  pages        = {18:1--18:17},
  publisher    = {Schloss Dagstuhl - Leibniz-Zentrum f{\"{u}}r Informatik},
  year         = {2020},
  url          = {https://doi.org/10.4230/LIPIcs.DISC.2020.18},
  doi          = {10.4230/LIPICS.DISC.2020.18},
  bibsource    = {dblp computer science bibliography, https://dblp.org}
}

@inproceedings{fraignaudIdentifiers2012,
  author       = {Pierre Fraigniaud and
                  Magn{\'{u}}s M. Halld{\'{o}}rsson and
                  Amos Korman},
  editor       = {Roberto Baldoni and
                  Paola Flocchini and
                  Binoy Ravindran},
  title        = {On the Impact of Identifiers on Local Decision},
  booktitle    = {Principles of Distributed Systems, 16th International Conference,
                  {OPODIS} 2012, Rome, Italy, December 18-20, 2012. Proceedings},
  series       = {Lecture Notes in Computer Science},
  pages        = {224--238},
  publisher    = {Springer},
  year         = {2012},
  url          = {https://doi.org/10.1007/978-3-642-35476-2\_16},
  doi          = {10.1007/978-3-642-35476-2\_16},
  bibsource    = {dblp computer science bibliography, https://dblp.org}
}

@inproceedings{suomelaLandscapetalk2020,
  author       = {Jukka Suomela},
  editor       = {Susanne Albers},
  title        = {Landscape of Locality (Invited Talk)},
  booktitle    = {17th Scandinavian Symposium and Workshops on Algorithm Theory, {SWAT}
                  2020, T{\'{o}}rshavn, Faroe Islands, June 22-24, 2020},
  series       = {LIPIcs},
  pages        = {2:1--2:1},
  publisher    = {Schloss Dagstuhl - Leibniz-Zentrum f{\"{u}}r Informatik},
  year         = {2020},
  url          = {https://doi.org/10.4230/LIPIcs.SWAT.2020.2},
  doi          = {10.4230/LIPICS.SWAT.2020.2},
  bibsource    = {dblp computer science bibliography, https://dblp.org}
}

@book{courcelleMSO2012,
  author       = {Bruno Courcelle and
                  Joost Engelfriet},
  title        = {Graph Structure and Monadic Second-Order Logic - {A} Language-Theoretic
                  Approach},
  series       = {Encyclopedia of mathematics and its applications},
  volume       = {138},
  publisher    = {Cambridge University Press},
  year         = {2012},
  url          = {http://www.cambridge.org/fr/knowledge/isbn/item5758776/?site\_locale=fr\_FR},
  isbn         = {978-0-521-89833-1},
  bibsource    = {dblp computer science bibliography, https://dblp.org}
}

@inproceedings{naorLocally1995,
  author       = {Moni Naor and
                  Larry J. Stockmeyer},
  editor       = {S. Rao Kosaraju and
                  David S. Johnson and
                  Alok Aggarwal},
  title        = {What can be computed locally?},
  booktitle    = {Proceedings of the Twenty-Fifth Annual {ACM} Symposium on Theory of
                  Computing, May 16-18, 1993, San Diego, CA, {USA}},
  pages        = {184--193},
  publisher    = {{ACM}},
  year         = {1993},
  url          = {https://doi.org/10.1145/167088.167149},
  doi          = {10.1145/167088.167149},
  bibsource    = {dblp computer science bibliography, https://dblp.org}
}

@inproceedings{balliuNewclasses2018,
  author       = {Alkida Balliu and
                  Juho Hirvonen and
                  Janne H. Korhonen and
                  Tuomo Lempi{\"{a}}inen and
                  Dennis Olivetti and
                  Jukka Suomela},
  editor       = {Ilias Diakonikolas and
                  David Kempe and
                  Monika Henzinger},
  title        = {New classes of distributed time complexity},
  booktitle    = {Proceedings of the 50th Annual {ACM} {SIGACT} Symposium on Theory
                  of Computing, {STOC} 2018, Los Angeles, CA, USA, June 25-29, 2018},
  pages        = {1307--1318},
  publisher    = {{ACM}},
  year         = {2018},
  url          = {https://doi.org/10.1145/3188745.3188860},
  doi          = {10.1145/3188745.3188860},
  bibsource    = {dblp computer science bibliography, https://dblp.org}
}

@INPROCEEDINGS{balliuMIS2019,
  author       = {Alkida Balliu and
                  Sebastian Brandt and
                  Juho Hirvonen and
                  Dennis Olivetti and
                  Mika{\"{e}}l Rabie and
                  Jukka Suomela},
  editor       = {David Zuckerman},
  title        = {Lower Bounds for Maximal Matchings and Maximal Independent Sets},
  booktitle    = {60th {IEEE} Annual Symposium on Foundations of Computer Science, {FOCS}
                  2019, Baltimore, Maryland, USA, November 9-12, 2019},
  pages        = {481--497},
  publisher    = {{IEEE} Computer Society},
  year         = {2019},
  url          = {https://doi.org/10.1109/FOCS.2019.00037},
  doi          = {10.1109/FOCS.2019.00037},
  bibsource    = {dblp computer science bibliography, https://dblp.org}
}

@inproceedings{barenboimMIS2009,
  author       = {Leonid Barenboim and
                  Michael Elkin},
  editor       = {Rida A. Bazzi and
                  Boaz Patt{-}Shamir},
  title        = {Sublogarithmic distributed {MIS} algorithm for sparse graphs using
                  nash-williams decomposition},
  booktitle    = {Proceedings of the Twenty-Seventh Annual {ACM} Symposium on Principles
                  of Distributed Computing, {PODC} 2008, Toronto, Canada, August 18-21,
                  2008},
  pages        = {25--34},
  publisher    = {{ACM}},
  year         = {2008},
  url          = {https://doi.org/10.1145/1400751.1400757},
  doi          = {10.1145/1400751.1400757},
  bibsource    = {dblp computer science bibliography, https://dblp.org}
}

@inproceedings{seidlPresburger2003,
  series = {SIGMOD/PODS03},
  title = {Numerical document queries},
  url = {http://dx.doi.org/10.1145/773153.773169},
  DOI = {10.1145/773153.773169},
  booktitle = {Proceedings of the twenty-second ACM SIGMOD-SIGACT-SIGART symposium on Principles of database systems},
  publisher = {ACM},
  author = {Seidl,  Helmut and Schwentick,  Thomas and Muscholl,  Anca},
  year = {2003},
  pages = {155–166},
  collection = {SIGMOD/PODS03}
}

@inbook{balliuSinkless2023,
  title = {Sinkless Orientation Made Simple},
  ISBN = {9781611977585},
  url = {http://dx.doi.org/10.1137/1.9781611977585.ch17},
  DOI = {10.1137/1.9781611977585.ch17},
  booktitle = {Symposium on Simplicity in Algorithms (SOSA)},
  publisher = {Society for Industrial and Applied Mathematics},
  author = {Balliu,  Alkida and Korhonen,  Janne H. and Kuhn,  Fabian and Lievonen,  Henrik and Olivetti,  Dennis and Pai,  Shreyas and Paz,  Ami and Rybicki,  Joel and Schmid,  Stefan and Studený,  Jan and Suomela,  Jukka and Uitto,  Jara},
  year = {2023},
  pages = {175–191}
}

@misc{jaureguiCourcelle2018,
  doi = {10.48550/ARXIV.1805.10708},
  url = {https://arxiv.org/abs/1805.10708},
  author = {Jauregui,  Benjamin and Li,  Jason and Montealegre,  Pedro and Todinca,  Ioan},
  title = {Distributed Treewidth Computation and Courcelle's Theorem in the CONGEST Model},
  publisher = {arXiv},
  year = {2018},
  copyright = {Creative Commons Attribution 4.0 International}
}

@inproceedings{balliuBinary2020,
  author       = {Alkida Balliu and
                  Sebastian Brandt and
                  Yuval Efron and
                  Juho Hirvonen and
                  Yannic Maus and
                  Dennis Olivetti and
                  Jukka Suomela},
  editor       = {Hagit Attiya},
  title        = {Classification of Distributed Binary Labeling Problems},
  booktitle    = {34th International Symposium on Distributed Computing, {DISC} 2020,
                  Virtual Conference, October 12-16, 2020},
  series       = {LIPIcs},
  pages        = {17:1--17:17},
  publisher    = {Schloss Dagstuhl - Leibniz-Zentrum f{\"{u}}r Informatik},
  year         = {2020},
  url          = {https://doi.org/10.4230/LIPIcs.DISC.2020.17},
  doi          = {10.4230/LIPICS.DISC.2020.17},
  bibsource    = {dblp computer science bibliography, https://dblp.org}
}

@inproceedings{halldorssonVertexsplit2022,
  author       = {Magn{\'{u}}s M. Halld{\'{o}}rsson and
                  Yannic Maus and
                  Alexandre Nolin},
  editor       = {Christian Scheideler},
  title        = {Fast Distributed Vertex Splitting with Applications},
  booktitle    = {36th International Symposium on Distributed Computing, {DISC} 2022,
                  Augusta, Georgia, USA, October 25-27, 2022},
  series       = {LIPIcs},
  pages        = {26:1--26:24},
  publisher    = {Schloss Dagstuhl - Leibniz-Zentrum f{\"{u}}r Informatik},
  year         = {2022},
  url          = {https://doi.org/10.4230/LIPIcs.DISC.2022.26},
  doi          = {10.4230/LIPICS.DISC.2022.26},
  bibsource    = {dblp computer science bibliography, https://dblp.org}
}

@INPROCEEDINGS{reiterAutomata2015,
  author={Reiter, Fabian},
  booktitle={2015 30th Annual ACM/IEEE Symposium on Logic in Computer Science}, 
  title={Distributed Graph Automata}, 
  year={2015},
  volume={},
  number={},
  pages={192-201},
  doi={10.1109/LICS.2015.27}}
\appendix
\section{Converting between problem types}\label{app:conversion}
Throughout the proof of Theorem~\ref{thm:autogap}, we treat $\Pi_{\mathcal{A},\Delta}$
alternatively as a problem on bounded degree \emph{rooted} trees (with or without inputs) and as a
problem on bounded degree \emph{undirected} trees with inputs; this encoding is often mentioned
as a trivial result without detail \cite{naorLocally1995,balliuRooted2022}. In this section
we formalize this result as two lemmas: one showing the immersion (encoding) of rooted
trees in unrooted node-labeled trees (Lemma~\ref{lem:roottound}) and a more general result
showing that we can encode \emph{locally verifiable requirements} on the graph structure
or input labels of \emph{trees} as part of a problem description without changing the
overall complexity (Lemma~\ref{lem:errorLCL}).
\subsection{Edge labels to node labels}
\begin{lemma}\label{lem:roottound}
    Let $\mathcal{T}(\Sigma)$ be the class of $\Sigma$-labeled \emph{undirected}
    trees and let $\mathcal{T}_d(\Sigma)$ be the class of $\Sigma_{in}$-labeled \emph{directed}
    trees; let $\mathcal{E}$ be the set $\{h,t,e\}$ then for any $\Sigma$ there is a function
    $c:\mathcal{T}_d(\Sigma)\hookrightarrow\mathcal{T}(\Sigma\uplus \mathcal{E})$ with the
    following properties:
    \begin{enumerate}[(1)]
        \item $c$ is an immersion (it is injective, but not surjective),
        \item for each $T_d\in \mathcal{T}_d(\Sigma)$ there is an immersion $c_{T_d}$ from
        the vertices of $T_d$ to the vertices of $c(T_d)$ which preserves $\Sigma$-labels, and such
        that $d_{c(T_d)}(c(v_1),c(v_2))=4d_{T_d}(v_1,v_2)$,
        \item for any $T\in \mathcal{T}(\Sigma\uplus \mathcal{E})$ the statement ``does there exist a $\Sigma$-labeled \emph{directed} tree
        $T_d$ such that $c(T_d)=T$?'' is locally decidable, and
        \item for any $T\in \mathcal{T}(\Sigma\uplus \mathcal{E})$ the statement ``does there exist a $\Sigma$-labeled \emph{rooted} tree
        $T_r$ such that $c(T_r)=T$?'' is locally decidable.
    \end{enumerate}
\end{lemma}
\begin{proof}
    This proof is inspired by similar encodings found in distributed graph algorithms
    literature, for example in \cite{balliuBinary2020}.

    We begin by observing that we can encode an edge orientation as a labeling on
    half-edges, where one half of edge $e$ receives label $h$ (head) and the other half
    receives label $t$ (tail). This encoding can be checked in one round: each node sends
    the label of each half-edge over that edge, and receives the same message from each
    neighbor, then checks that over each edge the messages sent and received exactly form
    the set $\{h,t\}$. Additionally, each node can check that the number of $h$ labels it
    is adjacent to (its own indegree) is at most one; if this is true for every node of a
    graph \emph{which we know to be a tree}, then the tree is rooted.

    Let $T_d=(V,E)$ be a node and half-edge labeled graph as above. We define its \emph{incidence graph} (also known as the Levi
    graph of the associated incidence structure) as the bipartite graph $L(T_d)=(V\uplus
    E,inc)$ which has a node for each node and edge of the original graph, where $v\in V$
    and $e\in E$ are adjacent if they are incident in the original graph. Visually, we add
    a node splitting each edge of the original graph in half.

    We repeat this operation once more: let $hE$ be the set of half-edges of $T_d$, then 
    $c(T_d):=L(L(T_d))=(V\uplus E\uplus hE,inc)$ is a tree with one node for each node, edge and
    half-edge of $T$ such that each edge is adjacent to exactly its two halves, each
    half-edge is adjacent to its edge and incident node, and each node is adjacent to all
    of its incident half-edges. All nodes in $V$ keep their $\Sigma$ labels, all nodes in
    $E$ are labeled $e$, and nodes in $hE$ are labeled $h$ or $t$ according to their
    orientation. In one round, each node can check whether the
    orientation is locally consistent:
    \begin{itemize}
        \item nodes labeled $h$ or $t$ check that they have exactly one neighbor labeled
        $e$ (originally an edge) and exactly one labeled with a label in $\Sigma$
        (originally a vertex), 
        \item nodes labeled $e$ check that they have exactly one neighbor labeled
        $h$ and exactly one labeled $t$,
        \item nodes with a label in $\Sigma$ check that all their neighbors are labeled
        with $h$ or $t$; if we are checking that the original tree was rooted, they
        additionally check that they have at most one $h$-labeled neighbor.
    \end{itemize}
    Since each edge has been split in $4$, we have
    $d_{c(T_d)}(c(v_1),c(v_2))=4d_{T_d}(v_1,v_2)$.
    
    Finally, we observe that the operation is bijective over its codomain; we can contract
    the nodes with labels in $\mathcal{E}$ (which have degree $2$) back to edges which we
    orient towards the node nearest to the $h$ label.
\end{proof}
In the above statement, property $(1)$ ensures we can convert between the two types of
graphs freely; property $(2)$ ensures that if we have a $\mathsf{LOCAL}$ algorithm to
solve \emph{any} problem $\Pi$ on directed/rooted trees (even on directed trees, and outside the
$\mathsf{PMSO}$-verifiable classes) in $T(n)$ rounds, we can run the same algorithm on
their image through $c$ in $\leq 4T(n)+1$ rounds, where the $+1$ accounts for gathering
half-edge information at the boundary (note that this requires knowledge of $n$ or at
least some bound on $T(n)$). Finally, properties $(3)$ and $(4)$ ensure we can properly
define a problem on \emph{all} undirected trees which is only hard on the trees that are
encoding a valid directed/rooted tree structure.
\subsection{Expanding graphs}
We show how to expand \emph{any} problem defined on a subset of trees to a problem defined
on all trees, as long as membership in the subset is locally verifiable. Observe that a
local verifier for a problem can give ``false positives'': since decision problems are
made \emph{by consensus}, as long as one node in every false instance outputs ``no'', the
instance will be rejected regardless of the amount of ``no''s it outputs. We can require
our verifier to be \emph{minimal} -- meaning it will always reject a neighborhood that
does not appear in any valid instance. 
\begin{definition}
    Let $\mathcal{V}$ be a $k$-round local verifier for property $\mathcal{G}$. We say $\mathcal{V}$
    is \emph{minimal} if for every $k$-hop neighborhood $N$ that is accepted by $\mathcal{V}$ there
    is $G\in\mathcal{G}$ such that $N$ is isomorphic to the $k$-hop neighborhood of a node in $G$.
\end{definition}
\begin{remark}\label{rem:minimal}
    If we have a $k$-round local verifier $\mathcal{V}$ for a property $\mathcal{G}$ and we allow
our verification algorithm to compute uncomputable functions, we can always find a minimal
verifier $\mathcal{V}'$; we define $\mathcal{V}'$ as follows: return ``yes'' on neighborhood $N$ if and only if
$\mathcal{V}$ does \emph{and} there exists a $G\in\mathcal{G}$ such that $N$ appears in
$G$. 
\end{remark}
Minimal verifiers for tree languages have an \emph{extension} property, meaning we
can extend neighborhoods of \emph{any} size which are accepted by the verifier into valid
graphs.
\begin{lemma}\label{lem:verextend}
    Let $\mathcal{V}$ be a $k$-round minimal local verifier for a (directed or undirected)
    tree property $\mathcal{G}$. Let $G$ be a tree, and let $W$ be a set of nodes of $G$
    such that:
    \begin{itemize}
        \item $\mathcal{V}$ returns ``yes'' on every node of $W$, and
        \item the subgraph induced by $W$ of $G$ is connected.
    \end{itemize}
    Then the subgraph induced by $W$ and every node that is $k$-away from $W$ in $G$ is
    isomorphic to an
    induced subgraph of a tree $G'\in\mathcal{G}$.
\end{lemma}
\begin{proof}
    If $\lvert W\rvert=1$, this is trivially true since $\mathcal{V}$ is minimal.
    Conversely, if $W=V(G)$ this is trivial since every node of $G$ is accepted and
    $G\in\mathcal{G}$.
    
    Let $v$ be a node on the boundary of $W$, that is, a node which is in $W$ but has a
    neighbor in $G$ which is not in $W$. Let $N^k(v)$ be its $k$-hop neighborhood in $G$
    -- we classify the edges incident to $v$ into nonempty sets $E_W(v)$ (edges between
    $v$ and a node of $W$) and $E_N(v)$ (edges between $v$ and a node \emph{not} in $W$).
    Then by the minimality of $\mathcal{V}$, let $H_v$ be the tree $H_v\in\mathcal{G}$ such that $N^k(v)$ is isomorphic to the
    the $k$-hop neighborhood of a node $v'\in H$.
    
    We ``glue'' $H_v$ and $G$ together --
    recall that both $H_v$ and $G$ are trees, so for each node $x\in G$ there is a unique
    shortest path from $x$ to $v$. For every $x\notin N^k(v)$, if this shortest path passes through an edge of
    $E_N(v)$, we remove $x$ -- conceptually, we remove all nodes that are ``on the other
    side'' of $v$ from $W$ and more than $k$-away from $v$. Similarly in $H_v$, for every
    $x'\notin N^k(v')$ if the shortest path between $x'$ and $v'$ passes through the
    isomorphic image of an edge of $E_W(v)$ we remove $x'$; now, we join $H_v$ and $G$ by
    identifying the two isomorphic neighborhoods. Observe that this operation does not
    affect the $k$-hop neighborhood of any node in $W$.

    We repeat this operation on every node $v$ on the boundary of $W$, and claim the
    resulting graph $G'$ is in $\mathcal{G}$: we run $\mathcal{V}$ on every node $v$ of $G'$,
    which has to be either:
    \begin{itemize}
        \item a member of $W$ -- its $k$-hop neighborhood is isomorphic to its
        neighborhood in $G$ and we know that $\mathcal{V}$ outputs ``yes'' on all nodes of $W$,
        so it outputs ``yes'' for $v$ in $G'$ as well, or
        \item not a member of $W$ -- let $w$ be the node of $W$ which $v$ is closest to,
        which is unique since the subgraph induced by $W$ is connected, and let $H_w$ be
        the graph that was glued at $w$. Then the distance between $w$ and any node that
        does not correspond to a node in $H_w$ is at least the distance between $w$ and
        $v$ (at least $1$) plus the distance between $w$ and this node (which is outside
        $N^k(w)$, so at least $k$) -- in other words, $N^k(v)$ in $G'$ is isomorphic to
        $N^k(v)$ in $H_w$, and since $H_w\in\mathcal{G}$ the verifier $\mathcal{V}$
        outputs ``yes'' for $v$ in $G'$ as well.
    \end{itemize}
    Then $\mathcal{V}$ accepts all nodes of $G'$ so $G'\in\mathcal{G}$.
\end{proof}
Observe that this ``extension'' graph $G'$ might be much larger or smaller than the graph $G$
we start from; this makes applying a $\mathsf{LOCAL}$ algorithm for $G'$ (whose complexity depends
on graph size) on $G$ difficult. We can solve this by requiring an even stronger property
from our verifier: the extension $G'$ we find has to be at least as big as $G$, but bigger
by at most a constant multiplicative factor.
\begin{definition}
    Let $\mathcal{V}$ be a $k$-round local verifier for a tree property $\mathcal{G}$. We say
    $\mathcal{V}$ has \emph{bounded expansion} if there is a function $f:\mathbb{N}\to
    \mathbb{N}$ such that $n\leq f(n) \in O(n)$ (the
    \emph{expansion function}) for every tree $G$ and every $W\subseteq
    V(G)$ such that:
    \begin{itemize}
        \item $W$ is not empty,
        \item $N^1(W)\subsetneq G$,
        \item $\mathcal{V}$ returns ``yes'' on every node of $W$, and
        \item the subgraph induced by $W$ of $G$ is connected.
    \end{itemize}
    Then the subgraph induced by $W$ and every node that is $k$-away from $W$ in $G$ is
    isomorphic to an induced subgraph of a tree $G'\in\mathcal{G}$ of size $\lvert
    G'\rvert=f(\lvert G\rvert)$.
\end{definition}
Once again, we show that the verifier from Lemma~\ref{lem:roottound} has this property.
\begin{proposition}
    The $1$-round verifier $\mathcal{V}$ for encoded rooted trees defined in the proof of
    Lemma ~\ref{lem:roottound} has bounded expansion with expansion function $f(n)=4n-3$.
\end{proposition}
\begin{proof}
    Observe that the encoding of a rooted tree $T$ with $k$ vertices is a undirected tree with
    $k$ $\Sigma$-labeled vertices and $3$ vertices for each edge of $T$: then the encoding
    of $T$ is a tree with $3(k-1)+k=4k-3$ nodes.

    If $W=V(G)$ this is trivial since every node of $G$ is accepted and
    $G\in\mathcal{G}$.
    
    We show how to extend the graph induced by $N^1(W)$ to a valid encoded tree. Let $v$ be a node \emph{immediately outside} the boundary of $W$, meaning a node of
    $G\smallsetminus W$ which has a neighbor in $W$; since $W$ induces a connected subgraph,
    $v$ has exactly one such neighbor, which we call $w$. Let $l:\Sigma\uplus\{h,t,e\}$ be
    the labeling of $G$; one of the following cases has to be true:
    \begin{enumerate}[(1)]
        \item $l(w)\in\Sigma, l(v)=h$: we can extend $N^1(W)$ by appending a path of
        length $3$ to $v$, such that the path $w-v-v_1-v_2-v_3$ is labeled
        $l(w)-h-e-t-l(w)$; then every vertex in this path is accepted by $\mathcal{V}$.
        \item $l(w)\in\Sigma, l(v)=t$: the same as the previous case, but the path is
        instead labeled $l(w)-t-e-h-l(w)$.
        \item $l(w)=h,l(v)=e$: we can extend $N^1(W)$ by appending a path of
        length $2$ to $v$, such that the path $w-v-v_1-v_2$ is labeled
        $h-e-t-s$ and $s$ is an arbitrary element of $\Sigma$; 
        then every vertex in this path is accepted by $\mathcal{V}$.
        \item $l(w)=t,l(v)=e$: the same as the previous case, but the path is labeled
        $t-e-h-s$ instead.
        \item $l(w)=e,l(v)\in\{h,t\}$: we can extend $N^1(W)$ by appending a single
        vertex to $v$ and labeling it by an arbitrary element $s\in\Sigma$.
        \item $l(w)\in\{h,t\},l(v)\in\Sigma$: we do not need to extend $N^1(W)$, as $v$ is
        already accepted by $\mathcal{V}$.
    \end{enumerate}
    Let $n$ be the number of vertices of $G$. This procedure adds at most $3$ new vertices every time it is used, and has to be run once for every
    vertex on the boundary of $W$, of which there are at most $n$; then the
    new graph $G'$ has at most $4n$ vertices. Let $n'$ be its actual number of vertices:
    since $G'$ is a valid instance (every node of $G'$ is accepted by $\mathcal{V}$) there
    is a $k$ such that $n'=4k-3$, then $4k-3=n'\leq 4n\implies k<n$. 
    
    $G'$ is an encoding of a tree $T$ with $k$ vertices; since $W\subseteq N^1(W)\subsetneq G$,
    there is at least one leaf of $G$ which is not in $W$ -- then there is also a leaf $v$ of
    $G'$ which is not in the isomorphic image of $W$, either because it was in $N^1(W)$
    (case (6) above) or because it was added at the end of a path (cases (1)-(5)). We can
    append a directed path of $n-k$ vertices to the degree 1 node corresponding to $v$ in
    $T$; the resulting rooted tree $T'$ has $n$ vertices, and its encoding is an
    undirected tree with $4n-3$ vertices which contains $N^1(W)$ as an isomorphic subtree.
\end{proof}
\subsection{Extending problems}
Under the strong assumptions we proved in the previous section, we can extend a locally
verifiable problem on a locally decidable class of graphs to a larger class of graphs --
specifically, we can extend encodings of rooted tree problems to all undirected labeled trees.
\begin{lemma}\label{lem:errorLCL}
    Let $\Pi$ be a locally verifiable (with radius $r$) problem defined on a class $\mathcal{G}$ of
    $\Sigma_{in}$-labeled \emph{trees}. Let $\mathcal{G}'$ be a class of trees such that:
    \begin{itemize}
        \item $\mathcal{G}\subseteq\mathcal{G}'$,
        \item $\mathcal{G}'$ is $\Sigma_{in}'$-labeled for some finite $\Sigma_{in}'$
        (clearly $\Sigma_{in}\subseteq\Sigma_{in}'$),
        \item $\mathcal{G}\in \mathsf{LD}^*(k)$ for some $k\in\mathbb{N}$ over
        $\mathcal{G}'$; in other words, the statement ``$G\in\mathcal{G}$'' is locally
        decidable in $k$ rounds for any $G\in\mathcal{G}'$, and
        \item the $k$-round local verifier for $\mathcal{G}$ has \emph{bounded expansion}. 
    \end{itemize}
    Then we can define a locally verifiable problem $\Pi'$ on $\mathcal{G}'$ such that:
    \begin{enumerate}
        \item any valid solution of $\Pi'$ for a graph $G\in\mathcal{G}$ is also a valid
        solution of $\Pi$ for $G$, and
        \item $\Pi$ and $\Pi'$ have the same asymptotic complexity.
    \end{enumerate}
\end{lemma}
\begin{proof}
    We describe the locally verifiable problem $\Pi'$ by its verification algorithm: $\Pi$ is a
    tuple $(\Sigma_{in},\Sigma_{out},r,\mathcal{A})$ where $\Sigma_{in},\Sigma_{out}$ are
    finite sets of input and output labels respectively, $r$ is the checkability radius
    and $\mathcal{A}$ is a $r$-round anonymous-$\mathsf{LOCAL}$ algorithm on
    $\Sigma_{in}\times\Sigma_{out}$-labeled graphs such that a $\Sigma_{out}$ labeling of
    a $\Sigma_{in}$-labeled graph $G$ is a valid solution to $\Pi$ if and only if
    $\mathcal{A}$ returns ``yes'' on every node of $G$.

    We define
    $\Pi'=(\Sigma'_{in},\Sigma_{out}\uplus\{err,err_0,err_1,err_2\},\max\{r,k\},\mathcal{A}')$: we only need
    to define $\mathcal{A}'$. Let $\mathcal{V}$ be the minimal $k$-round local verifier of
    membership in $\mathcal{G}$; we define $\mathcal{A}'$ to accept if and only if:
    \begin{itemize}
        \item the center is out-labeled with $err$ if and only if $\mathcal{V}$ outputs
        ``no'' at the center (looking at its $k$-hop neighborhood), \emph{and}
        \item if the center is out-labeled with $err_1$ (resp: $err_2$) then it is
        accepted if and only if it has a neighbor labeled with $err_0$ (resp: $err_1$), \emph{and}
        \item if the center is out-labeled with $err_0$ then it is
        accepted if and only if it has a neighbor labeled with either $err_2$ or $err$, \emph{and}
        \item if the center is out-labeled with an element of $\Sigma_{out}$ and there is
        \emph{any} error label ($err$, $err_0$, $err_1$, or $err_2$) in its $r$-hop neighborhood, it
        is accepted, \emph{and}
        \item if the center is out-labeled with an element of $\Sigma_{out}$ and there is
        \emph{no} error label in its $r$-hop neighborhood, it is accepted if and only if
        $\mathcal{A}$ returns ``yes'' at the center (looking at its $r$-hop neighborhood).
    \end{itemize}
    We prove that $\Pi'$ has the required properties. First, let $G\in\mathcal{G}$: we
    solve $\Pi'$ on $G$, and show that all out-labels are in $\Sigma_{out}$. Assume by
    contradiction there is a $err$ out-label on some node: then the local verifier has
    output ``no'' on that node, which would imply $G\notin\mathcal{G}$, but this is false
    by hypothesis. Assume then that there is a $err_i$ for some $i\in\{0,1,2\}$: since
    there are no $err$ out-labels, every $err_i$ out-labeled node is adjacent to an
    $err_{i-1\mod 3}$ out-labeled node: we can follow a path in the tree which is
    out-labeled $err_2,err_1,err_0,err_2\ldots$ -- this path will never contain a cycle,
    as $T$ is by definition a tree, but it can also never end (as it can only end at an
    $err$ label), so it will eventually
    contain more nodes than $G$ itself, which is a contradiction. Finally, since all nodes
    are out-labeled from $\Sigma_{out}$, their $r$-hop neighborhood has to be accepted by
    $\mathcal{A}$, so the solution is correct for $\Pi$.

    Second, we show that $\Pi$ and $\Pi'$ have the same asymptotic complexity: let $f\in O(n)$ be
    the expansion function of $\mathcal{G}$ over $\mathcal{G}'$, and assume we have a
    $T(n)$ algorithm to solve $\Pi$; since $T(n)$ and $f(n)$ are at
    most linear there is $c>0$ (WLOG, we assume $c\geq 1$) and $n_0\in\mathbb{N}$ such that for every $n'\geq n_0$ we have
    $T(f(n'))\leq cf(T(n'))$. We can assume without loss of generality that $T(n)$ is
    weakly increasing, as if $T(n)>T(n+1)$ we can find a $T(n)$-round algorithm for graphs
    of size $n+1$ as a $T(n)$-round algorithm which does not use the extra information it gathers.
    
    Let $n$ be the number of nodes of our input instance $G$: if $n<n_0$, we gather
    the whole graph and run the verifier algorithm $\mathcal{V}$ on every node: if any
    nodes output ``no'', the algorithm out-labels them with $err$, and every other node
    out-labels itself with $err_{d-1\mod 3}$ where $d$ is the shortest distance between
    itself and an error.
    
    Assume $n\geq n_0$, and so $T(f(n))\leq cf(T(n))$: let $T'(n):=cf(T(n))\in
    O(T(n))$. Each node gathers its $(2T'(n)+k)$-hop
    neighborhood and runs the $k$-round verifier $\mathcal{V}$ for every node of its
    $2T'(n)$-hop neighborhood. Let $E$ the set of nodes for which $\mathcal{V}$ returns
    ``no''. If there are any such nodes in the $2T'(n)$-hop neighborhood of $v$, let $d$ be
    the minimum distance between $v$ and a node of $E$: if $d=0$, $v$ is out-labeled with
    $err$, else it is out-labeled with $err_{d-1 \mod 3}$.

    Assume no nodes in the $2T'(n)$-hop neighborhood of $v$ return ``no'', and there is a
    node in the tree that can see $n$ nodes (the whole graph) in its $T'(n)$-hop
    neighborhood: then any two nodes are at most $2T'(n)$ distance from each other, and
    every node has gathered the whole graph, which is a correct instance since no nodes
    have been rejected -- then each node can use the $T(n)\leq f(T(n))\leq cf(T(n))$
    algorithm for $\Pi$ to compute a full solution and choose the out-label for $v$ accordingly.

    Assume no nodes in the $2T'(n)$-hop neighborhood of $v$ return ``no'', and there are
    no nodes in the tree that can see $n$ nodes (the whole graph) in their $T'(n)$-hop
    neighborhood: then by the bounded expansion property of $\mathcal{V}$, there is a graph with at
    most $f(n)$ nodes which contains the $T'(n)+k$-hop neighborhood of $v$ as a isomorphic
    induced subgraph -- we then run the $T(n')$-round algorithm for $\Pi$ while pretending
    that the real number of nodes is $n'=f(n)$, since $T(n')=T(f(n))\leq cf(T(n))=T'(n)$,
    and choose the out-label for $v$ accordingly.

    We verify that this results in a correct solution: let $v$ be a node of $G'$. Then
    \begin{itemize}
        \item $v$ is out-labeled $err$ if and only if it is at distance $0$ from a node
        rejected by $\mathcal{V}$ -- meaning it is rejected itself;
        \item if $v$ is out-labeled $err_i$, then there is a $m\in\mathbb{N}$ such that the minimum
        distance between $v$ and a node $v_e$ labeled $err$ is $3m+i+1$; let $v'$ be the
        neighbor of $v$ along the minimum path between $v$ and $v_e$, then $v'$ is at
        distance $3m+i$ from $v_e$ and it is labeled either $err$ (if $m=i=0$) or
        $err_{i-1\mod 3}$;
        \item if every node in the $r$-hop neighborhood of $v$ is out-labeled with an
        element of $\Sigma_{out}$, then the $(T(n)+r)$-hop neighborhood of $v$ is accepted
        by $\mathcal{V}$ -- there is a graph $G\in\mathcal{G}$ which contains an induced
        subgraph isomorphic to the $(T'(n)+r+k)$-hop neighborhood of $v$, and the algorithm for $\Pi$ has to
        be correct on $G$, meaning the $r$-hop neighborhood of $v$ has to be accepted by
        $\mathcal{A}$, and
        \item if none of the above apply, $v$ is out-labeled with an element of
        $\Sigma_{out}$ and has a $r$-hop neighbor which is labeled with an error label,
        which is correct.
    \end{itemize} 
    The runtime of this algorithm for $\Pi'$ is $\max\{n_0,2T'(n)+k\}$; since $2T'(n)+k\in O(T(n))$, $\Pi$ and $\Pi'$ have the same asymptotic complexity.
\end{proof}
Since we have shown that the verifier for the set of undirected graphs encoding rooted
trees has a bounded expansion verifier, we can convert a problem on rooted trees to a
problem on unrooted trees; specifically, we can convert a $\mathsf{LCL}$ on rooted trees
(with or without input labels) to a $\mathsf{LCL}$ on undirected trees with input labels. 
Then, classification results which apply to all problems on
undirected trees with inputs (such as the classification results of
\cite{changHierarchy2019,treesChang2020}) apply to problems on rooted trees as well,
making the statements of Theorem~\ref{thm:logred} and~\ref{thm:polyred} well-defined.
\section{Graph decompositions}\label{app:decomposition}
We provide a more detailed proof of Lemmas~\ref{lem:logdec} and~\ref{lem:polydec}. Both
proofs are standard for bounded degree graphs, and their proofs do not strongly rely on
the bounded degree; they are included purely for completeness. We
define two shared steps for the two algorithms:

$\mathsf{Rake}$: a Rake step takes one round -- each node which is a leaf joins the current
layer, then removes itself from the graph by communicating its removal to its parent.

$\mathsf{Compress}(l)$: a Compress step takes a parameter $l$, and takes $l+1$ rounds --
each node that has indegree$=1$ and outdegree$=$1 checks whether it is part of a directed path of
nodes with the same degrees of length $\geq l$. If it is, it joins the current layer and removes itself from the graph
by communicating its removal to both its parent and child.
\begin{lemma}\label{lem:postprocess}
    Let $\ell$ be a constant and $P$ be a path of length $n\geq\ell$. Then in $O(\log^*n)$ rounds in the randomized
   or deterministic $\mathsf{LOCAL}$ model we can construct a
    $(\ell,\ell)$-ruling set $R$ such that $P\smallsetminus R$ has connected components of
    length $[\ell,2\ell]$.
\end{lemma}
\begin{proof}
    We start by computing a $(\ell+1)$-hop $(2\ell+1)^2$-coloring -- since $\ell$ is a constant, this is
    possible in $O(\log^* n)$ rounds using Linial's algorithm \cite{linialLocality1992}. Then, we use this coloring to greedily compute
    a $(\ell,\ell)$ ruling set of the graph: nodes can choose by looking at their $l$-hop
    neighborhood whether to join the ruling set, and no nodes that are $\leq\ell$ apart
    make the choice at the same time. The nodes of the ruling set split the path in parts
    of length $[\ell,2\ell]$, except for the first and last part, which have
    length $[0,\ell]$. In $2\ell$ rounds, the first and last nodes of the ruling set can
    reposition themselves: this results in first and last parts of length $x\in[\ell,3\ell]$.
    If $x\leq 2\ell$, we are happy; if $x\in(2\ell,3\ell]$ we add the node at position
    $\lceil x/2\rceil$ to the ruling set. Every part is now of size $[\ell,2\ell]$; the
    overall complexity is $O(\log^*n)$.
\end{proof}
We can now give complete proofs of Lemmas~\ref{lem:logdec} and~\ref{lem:polydec}.
\logdec*
\begin{proof}
    Assume $n\geq 2$; we do a $\mathsf{Rake}$ step and a
    $\mathsf{Compress}(\ell)$ step at the same time, adding the $\mathsf{Rake}$ vertices
    to $R_i$ and the $\mathsf{Compress}(\ell)$ vertices to $C_i$ -- this is possible as no vertex can
    be removed or even considered by both $\mathsf{Rake}$ and $\mathsf{Compress}(\ell)$ at
    the same time. We show that overall, we remove at least
    $\frac{n}{4\ell}$ vertices. Denote:
    \begin{align*}
        n_1&=\text{size of }V_1:=\{\text{vertices with indeg$\,=1$, outdeg$\,=0$}\}\\
        n_2&=\text{size of }V_2:=\{\text{vertices with indeg$\,=1$, outdeg$\,=1$}\}\\
        n_3&=\text{size of }V_3:=\{\text{vertices with indeg$\,=1$, outdeg}\geq 2\}\\
        n_A&=\text{size of }A:=\{\text{vertices in $V_2$ that are removed with $\mathsf{Compress}(\ell)$}\}\\
        n_B&=\text{size of }B:=\{\text{vertices in $V_2$ that are not removed with $\mathsf{Compress}(\ell)$}\}
    \end{align*}
    Observe that the root is not counted in any of these sets, as it has indegree $1$;
    then $n_1+n_2+n_3=n-1$. We denote by $d_r$ the outdegree of the root, then by adding
    the outdegrees:
    \begin{equation}
        n_2+2n_3\leq 0n_1+1n_2+2n_3+d_r\leq n-1=n_1+n_2+n_3\Rightarrow n_1\geq n_3.
    \end{equation}
    Consider the vertices of $B$: their connected components are directed paths of length $\leq\ell-1$
    where both endpoints are connected to a vertex in $V_1\cup V_3$ or the root. By contracting these
    paths to a single edge and removing the vertices of $A$, we obtain a forest with
    $n_1+n_3+1$ vertices and $\leq n_1+n_3$ edges: then
    \begin{gather*}
        n_B\leq(\ell-1)(n_1+n_3)\\
        n_A=n_2-n_B\geq n_2-(\ell-1)(n_1+n_3)=\ell n_2-(n-1)(\ell-1)\\
        n_2\leq \frac{n_A}{\ell}+(n-1)\frac{\ell-1}{\ell}\tag{2}
    \end{gather*}
    We apply inequalities (1) and (2) to obtain:
    \begin{gather*}
        2n_1+\frac{n_A}{\ell}+(n-1)\frac{\ell-1}{\ell}\geq n_1+n_2+n_3=n-1\\
        2n_1+2n_A\geq2n_1+\frac{n_A}{\ell}\geq \frac{n-1}{\ell}\\
        n_1+n_A\geq \frac{n-1}{2\ell}\geq \frac{n}{4\ell}
    \end{gather*}
    for $n\geq 2$. Since $n_1$ is the number of vertices removed by $\mathsf{Rake}$ and
    $n_A$ is the number of vertices removed by $\mathsf{Compress}$, we remove at
    least $\frac{n}{4\ell}$ vertices at each step until we have at most $2$ vertices left; they are removed by
    two extra $\mathsf{Rake}$ steps. Then in 
    $$L=\lceil\log_{\frac{4\ell}{4\ell-1}}(n)-\log_{\frac{4\ell}{4\ell-1}}(2)+2\rceil\in O(\log n)$$
    steps we remove all the vertices. 
    
    The resulting $R_1,C_1,\ldots,R_L$ is not yet a valid $(1,\ell,L)$ decomposition,
    as paths in $C_i$ might be longer than $2\ell$. We use Lemma~\ref{lem:postprocess} on
    each connected component of each $C_i$ \emph{separately} and \emph{in parallel}; this
    is possible even if the overall graph has unbounded degree, since each component is a
    path. Then for each $C_i$, we promote all nodes that are in the ruling set to
    $R_{i+1}$: these nodes are happy since they have no neighbors in higher layers and the
    remaining components of $C_i$ have length $[\ell,2\ell]$, so we now have a valid 
    $(1,\ell,L)$ decomposition.

    Overall, it takes $\ell+1\in O(1)$ rounds to do $\mathsf{Rake}$ and $\mathsf{Compress}(\ell)$
    and we repeat them for $O(\log n)$ times; the postprocessing step takes $O(\log^*n)$
    rounds, so the overall complexity is $O(\log n)$.
\end{proof}
\polydec*
\begin{proof}
    In this proof, we do $\gamma$ $\mathsf{Rake}$ steps and one $\mathsf{Compress}(\ell)$
    step for each pair of layers $R_i,C_i$; however, we cannot do them at the same time as
    repeated instances of $\mathsf{Rake}$ can remove vertices that started with
    outdegree$=1$, which affects the $\mathsf{Compress}$ step. We add all vertices removed
    by the $i$th iteration of applying $\mathsf{Rake}$ $\gamma$ times to $R_i$, and all
    vertices removed by the $i$th iteration of $\mathsf{Compress}(\ell)$ to $C_i$.

    Let $S_i$ be the vertices left after performing the above operations for
    $i-1$ times, and $S_i'$ be the result of applying $\mathsf{Rake}\ \gamma$ times --
    note that $S_1=V$. We
    want to prove that $\lvert S_k\rvert\leq \gamma$, meaning that all
    vertices will be removed by the last application of $\gamma\ \mathsf{Rake}$ steps. 

    Let $A_i$ be the vertices of $S_i'$ that have indegree$\,=1$ and
    outdegree$=\,1$ but are not removed by the $i$th $\mathsf{Compress}(\ell)$; let $B_i$
    be the vertices of $S_i'$ which have either indegree$\,=0$ (the root) or
    outdegree$\,\neq1$; then $\lvert S_{i+1}\rvert=\lvert A_i\rvert+\lvert B_i\rvert$. The connected components of
    the subgraph induced by $A_i$ are directed paths of length at most $(\ell-1)$: if we
    contract them in $S_{i+1}$, we obtain a forest with vertices in $B_i$ -- then $\lvert
    A_i\rvert\leq (\ell-1)(\lvert B_i\rvert-1)$.
    
    Let $B^1_i$ be the set of vertices that have outdegree$=0$ at the beginning of the
    $i$th $\mathsf{Compress}$; by adding outdegrees, we find that 
    \begin{gather*}
        0\lvert B_i^1\rvert+2\lvert B_i\smallsetminus B_i^1\rvert=2\lvert B_i\smallsetminus B_i^1\rvert\leq \lvert B_i\rvert\\
        \lvert B_i\smallsetminus B_i^1\rvert\leq \lvert B_i\rvert /2\\
        \lvert B_i^1\rvert\geq \lvert B_i\rvert/2
    \end{gather*}
    Additionally, each element of $B_i^1$ had to have been adjacent to a connected
    component of $S_i\smallsetminus S_i'$ (removed in the $i$th $\mathsf{Rake}$ step) of
    size at least $\gamma$: therefore 
    $$\lvert B_i\rvert/2\leq\lvert B_i^1\rvert\leq\lvert S_i\rvert/\gamma$$
    Finally, we show
    $$\lvert S_{i+1}\rvert= \lvert A_i\rvert+\lvert B_i\rvert\leq(\ell-1)(\lvert
    B_i\rvert-1)+\lvert B_i\rvert\leq \ell \lvert B_i\rvert\leq \frac{2\ell}{\gamma}\lvert
    S_i\rvert$$
    By applying this to every layer from $1$ to $k-1$, we get
    $$\lvert S_k\rvert\leq \left(\frac{2\ell}{\gamma}\right)^{k-1}n\leq \gamma,$$
    and all vertices get removed by the $k$th $\mathsf{Rake}$ step.

    The resulting $R_1,C_1,\ldots,R_k$ is not yet a valid $(\gamma,\ell,k)$ decomposition,
    as paths in $C_i$ might be longer than $2\ell$. We use Lemma~\ref{lem:postprocess} on
    each connected component of each $C_i$ \emph{separately} and \emph{in parallel}; this
    is possible even if the overall graph has unbounded degree, since each component is a
    path. Then for each $C_i$, we promote all nodes that are in the ruling set to
    $R_{i+1}$: these nodes are happy since they have no neighbors in higher layers and the
    remaining components of $C_i$ have length $[\ell,2\ell]$, so we now have a valid 
    $(\gamma,\ell,k)$ decomposition.

    Overall, it takes $\ell+\gamma\in O(n^{1/k})$ rounds to do $\mathsf{Rake}$ $\gamma$ times and $\mathsf{Compress}(\ell)$
    and we repeat them for $k\in O(1)$ times; the postprocessing step takes $O(\log^*n)$
    rounds, so the overall complexity is $O(n^{1/k})$.
\end{proof}
\section{Myhill-Nerode}\label{app:myhill}
\minratmtathm*
\begin{proof}
    $(1)$: let $\delta(t)$ denote the state of the root of $t$; observe that
    $\delta(t_1)=\delta(t_2)\Rightarrow t_1\equiv_\mathcal{G} t_2$, so the number of
    congruence classes of $\equiv_\mathcal{G}$ is bounded by the number of states, which
    is finite.

    $(2)$: let $\mathcal{A}=(\Sigma,Q,\delta,F)$ and let $Q'$ be the set of
    congruence classes of $\equiv_\mathcal{G}$ (defined equivalently over all trees or
    over the set of states; see (1)). Consider the quotient projection $h:Q\to Q'$; we can
    define $\dot{h}$ which applies $h$ to multisets pointwise, and $\ddot{h}$ which
    applies $\dot{h}$ to multiset languages pointwise. $\ddot{h}$ sends semilinear (rational) sets to semilinear
    (rational) sets, as in any vector defining a linear set we can find the index of
    $h(q)$ by adding the indices of all elements of $h(q)$. We define a transition
    $\delta_{min}:\Sigma\times Q'\to\mathrm{Rat}(\mathcal{M}(Q'))$ as $\delta_{min}(s,h(q_1))=\ddot{h}(\delta(s,q_1))$.
    Recall that $\mathcal{A}$ is complete and deterministic, so we can define an inverse function $\gamma:\Sigma\times\mathcal{M}(Q)\to
    Q$ where $\gamma(s,M)=q \Leftrightarrow M\in\delta(s,q)$. Similarly, we define
    $$\gamma_{min}(s,\dot{h}(M))=h(q)\Leftrightarrow \dot{h}(M)\in\delta_{min}(s,h(q))=\ddot{h}(\delta(s,q_1))$$ 
    and observe that $\delta_{min}$ is well-defined and induces a deterministic and
    complete RatMTA iff $\gamma_{min}$ is well-defined and total.

    Let $M_1,M_2$ such that $\dot{h}(M_1)=\dot{h}(M_2)$: we show that for all $s\in\Sigma$
    we have $h(\gamma(s,M_1))=h(\gamma(s,M_2))$. Consider rooted trees $t_1,t_2$ with
    $\lvert M_1\rvert=\lvert M_2\rvert$ children to the root, such that the children of
    $t_1,t_2$ have sub-root states corresponding to the elements of $M_1,M_2$
    respectively. Clearly $\delta(t_1)=q_1$ and $\delta(t_2)=q_2$. Since
    $\dot{h}(M_1)=\dot{h}(M_2)$, we can find a one-to-one correspondence $f:M_1\to M_2$
    such that $h(q)=h(f(q))$ for all $q\in M_1$. Let $C[x]$ be a context, and consider
    $C[t_1]$. We replace each child tree of the root of
    $t_1$ of type $q$ with the corresponding child tree of the root of $t_2$ of type
    $f(q)$: this operation is equivalent to replacing $t_1$ with $t_2$. Each step of this
    replacement preserves membership of the overall tree in $\mathcal{G}$, so
    $C[t_1]\in\mathcal{G}\Leftrightarrow C[t_2]\in\mathcal{G}$ and
    $t_1\equiv_{\mathcal{G}} t_2$. Then $\gamma_{min}$ is well defined; since
    $h,\dot{h},\ddot{h}$ are surjective it is also total.

    Finally, we define $\mathcal{A}_{min}=(\Sigma, Q', \delta_{min},
    F'=F/\equiv_{\mathcal{G}})$. We need to show that $\mathcal{A}_{min}$ correctly
    recognises $\mathcal{G}$: given a rooted tree $(V,E,r)$ any accepting run $l:V\to Q$
    of $\mathcal{A}$ corresponds to an accepting run $h\circ l$ of
    $\mathcal{A}_{min}$. Conversely, for each accepting run $\dot{l}$ of
    $\mathcal{A}_{min}$, we can label $\mathcal{A}$ by exploring
    the tree root-to-leaves and letting each subroot (with label $s$ and state $q$) choose a valid multiset of children
    states $M\in\delta(s,q)$ such that $\dot{h}(M)$ is the set of states of its children
    in $\dot{l}$; then, for each child subtree $t_1$ with $\delta(t_1)=q_1$ which was
    wrongly labeled $q_2$, we can replace it with $t_2$ with $\delta(t_2)=q_2$ since
    $q_1\equiv_\mathcal{G} q_2$. The result is a tree $T'$ labeled with an accepting run
    of $\mathcal{A}$ such that $T\in\mathcal{G}\Leftrightarrow T'\in\mathcal{G}$.
\end{proof}
\end{document}